\documentclass{article}
\usepackage[left=2.5cm,right=2.5cm,top=2.5cm,bottom=2.5cm]{geometry}
\usepackage{amssymb,amsmath,amsthm}

\usepackage{hyperref}
\usepackage{graphicx}
\usepackage{epstopdf}
\usepackage{xcolor}

\usepackage{mathtools}

\usepackage[capitalize,nameinlink]{cleveref} 

\crefformat{equation}{(#2#1#3)}

\newtheorem{theorem}{Theorem}[section]
\newtheorem{proposition}{Corollary}[theorem]

\newcommand{\N}{\mathbb{N}}

\newcommand{\R}{\mathbb{R}}
\newcommand{\Z}{\mathbb{Z}}

\newcommand{\Fref}[1]{\Cref{#1}}
\newcommand{\fref}[1]{\cref{#1}}

\newcommand{\eref}[1]{\cref{#1}}

\title{Chaotic switching in driven-dissipative Bose-Hubbard dimers: \\ 
when a flip bifurcation meets a T-point in $\R^4$}

\author{Andrus Giraldo\footnotemark[1] \footnotemark[3]
\and Neil G. R. Broderick\footnotemark[2] \footnotemark[3]
\and Bernd Krauskopf\footnotemark[1] \footnotemark[3]}

\date{}

\begin{document}

\renewcommand{\thefootnote}{\fnsymbol{footnote}}
\footnotetext[1]{Department of Mathematics, The University of
  Auckland, Private Bag 92019, Auckland 1142, New Zealand
  (\href{mailto:a.giraldo@auckland.ac.nz}{a.giraldo@auckland.ac.nz},
  \href{mailto:b.krauskopf@auckland.ac.nz}{b.krauskopf@auckland.ac.nz})}
\footnotetext[2]{Department of Physics, The University of
  Auckland, Private Bag 92019, Auckland 1142, New Zealand
  (\href{mailto:n.broderick@auckland.ac.nz}{n.broderick@auckland.ac.nz})}
\footnotetext[3]{Dodd-Walls Centre for Photonic and Quantum Technologies, New Zealand}

\renewcommand{\thefootnote}{\arabic{footnote}}

\maketitle

\begin{abstract}

The Bose--Hubbard dimer model is a celebrated fundamental quantum mechanical model that  accounts for the dynamics of bosons at two interacting sites. It has been realized experimentally by two coupled, driven and lossy photonic crystal nanocavities, which are optical devices that operate with only a few hundred photons due to their extremely small size. Our work focuses on characterizing the different dynamics that arise in the semiclassical approximation of such driven-dissipative photonic Bose--Hubbard dimers. Mathematically, this system is a four-dimensional autonomous vector field that describes two specific coupled oscillators, where both the amplitude and the phase are important. We perform a bifurcation analysis of this system to identify regions of different behavior as the pump power $f$ and the detuning $\delta$ of the driving signal are varied, for the case of fixed positive coupling. The bifurcation diagram in the $(f,\delta)$-plane is organized by two points of codimension-two bifurcations --- a $\Z_2$-equivariant homoclinic flip bifurcation and a Bykov T-point --- and provides a roadmap for the observable dynamics, including different types of chaotic behavior. To illustrate the overall structure and different accumulation processes of bifurcation curves and associated regions, our bifurcation analysis is complemented by the computation of kneading invariants and of maximum Lyapunov exponents in the $(f,\delta)$-plane. The bifurcation diagram displays a menagerie of dynamical behavior and offers insights into the theory of global bifurcations in a four-dimensional phase space, including novel bifurcation phenomena such as degenerate singular heteroclinic cycles.

\end{abstract}


\section{Introduction} \label{sec:intro}

The nonlinear response of a single isolated optical resonators has attracted considerable interest since the early works of Gibbs \cite{Gibbs1985} due to the 
plethora of behaviors they may display, including bistability, ultra-short pulses, cavity solitons, and most recently optical frequency combs \cite{10.5555/1593342}. 
While such resonators come in many shapes and sizes we focus here on optical microresonators based on photonic crystal structures, which are realized by introducing tiny defects in a regular lattice that generate refractive index changes on the order of the wavelength of the light \cite{Joannopoulos:08:Book}. Such photonic crystal nanocavities are extremely small and provide very strong confinement of light at a specified frequency, and this allows them to operate with unusually low numbers of photons --- down to a few hundred or even only tens of photons \cite{Hamel2015}. As such, these devices offer the possibility of investigating the interface between classical and quantum optics. More precisely, photonic crystal nanocavities have gained a lot of attention as prospective experimental realizations of an open quantum system, where an optical source replenishes photons lost from the system due to leaking \cite{MaiaThesis,BinMahmud, Casteels2017}. It is generally recognized that understanding the possible dynamics of such 
devices through careful analysis of appropiate mathematical models is of vital importance for explaining and guiding experimental observations. 

Moving beyond a single optical cavity, which has been well studied, we  focus here on a system of two resonantly driven identical coupled photonic crystal nanocavities. 
Such devices have been fabricated, and experiments with small to moderate driving (amount of input light) showed spontaneous symmetry breaking, in good 
agreement with the analysis of equilibria of the associated mathematical model \cite{MaiaThesis,And3,Hamel2015}. This system of two coupled optical microresonators is of wider importance, because it is known to realize the so-called driven-dissipative two-site Bose-Hubbard dimer model \cite{PhysRevA.93.033824}, which is a well-known quantum-optical model that also describes a two-mode approximation at low temperatues of a Bose-Einstein condensate in a double-well potential \cite{PhysRevLett.95.010402, PhysRevE.64.025202}. The semiclassical description of the Bose-Hubbard dimer is derived from its quantum-theoretical description by the semiclassical mean-field approximation of the so-called Lindblad master equation \cite{BinMahmud, Carmichael:1631392, Casteels2017}. This semiclassical approximation agrees with the macroscopic model for two identical  driven nonlinear optical resonator that are linearly coupled. After a suitable rescaling that links the physical parameters of the system to a reduce set of new parameters (see \cite{And3} for details), it takes the form
\begin{equation}  \label{eq:Couplednondim} 
\begin{cases}  
\dot{A} & = i\left(\delta + |A|^2 \right) A - A+ i\kappa B  + f, \\[2mm]
\dot{B} & = i\left(\delta  + |B|^2 \right) B - B + i\kappa A  + f,
\end{cases}
\end{equation}
for the complex variables $A$ and $B$, which are the (rescaled) envelopes of the electric fields in each cavity that describe the amplitudes and phases of the light in the two cavities; in particular, the intensities of light in each cavity are $|A|^2$ and $|B|^2$. Both cavities are assumed to be identical and coherently driven by a driving field of (rescaled) intensity $f$ and detuning $\delta$ (with respect to the frequency of the identical cavities); moreover, $\kappa$ is the (rescaled) coupling strength. The positive sign of the nonlinear term follows from considering a postive Kerr-type nonlinearity. We remark that our results can also be mapped to the case of negative nonlinearity by a suitable transformation; see \cite{And3} for more details. 

Photonic crystal nanocavities with different coupling strengths can and have been fabricated, but the respective coupling strength $\kappa$ remains fixed for any given device as it represents the geometry of the resonators and, in particular, the distance between them. The pump power $f$ and the detuning $\delta$ of the input light, on the other hand, are parameters that can readily be varied during an experiment. Moreover, since the cavities are identical and the driving is assumed to be the same for both of them, system~\eref{eq:Couplednondim} has the natural $\Z_2$-equivariance of interchanging the two cavities, that is, of interchanging the complex variables $A$ and $B$. 

Equations~\eref{eq:Couplednondim} define an autonomous real vector field with a four-dimensional phase space, and the main purpose of this paper is to present a detailed study of its bifurcation structure to determine and illustrate the dynamics the two microresonators may display. We focus here on the case of positive coupling, because photonic crystals nonocavities with $\kappa > 0 $ can be fabricated and studied experimentally \cite{Hamel2015}. More precisely, we consider throughout the representative value of $\kappa=2$ and present a comprehensive bifurcation diagram in the $(f,\delta)$-plane, where we identify, delimit and explain the different dynamical behaviors, from simple to chaotic, that system~\cref{eq:Couplednondim} exhibits. To accomplish this, we combine parameter sweeping of kneading invariants with the numerical continuation of global bifurcations.  We identify two codimension-two bifurcation points, a $\Z_2$-equivariant homoclinic flip bifurcation and a Bykov T-point, as organizing centers for infinitely many global bifurcations in the $(f,\delta)$-plane.  We compute a considerable number of the curves of such global bifurcations that emerge from these two special points. Here, we pay particular attention to the role of global bifurcations involving saddle periodic orbits for the emergence of chaotic attractors and their subsequent bifurcations, namely symmetry increasing bifurcations \cite{Chossat1988} and boundary crises \cite{GREBOGI1983181}. Moreover, we show how different global bifurcations interact with each other at additional codimension-two points of degenerate singular cycles \cite{Kirk1, Lohr1} and their generalizations. 

Overall, our bifurcation analysis of the semiclassical limit of the Bose-Hubbard dimer reveals a menagerie of dynamical behavior. Our findings identify the role of pretty exotic global bifurcation phenomena for the organization of chaotic behavior --- within reach of highly controlled future experiments with coupled nanocavities. More generally, our work showcases the important role of state-of-the-art numerical continuation techniques, in tandem with parameter sweeping, as a bridge between highly advanced results in dynamical systems and their relevance for applications. By studying and computing all these global bifurcations, we are able to present a comprehensive bifurcation diagram in the $(f,\delta)$-plane of system~\eref{eq:Couplednondim} that explains features identified by the computation of kneading invariants and of maximum Lyapunov exponents. It serves as a roadmap for symmetry increasing bifurcations of chaotic attractors and, moreover, identifies system~\cref{eq:Couplednondim} as a promising, concrete example vector field featuring novel types of global bifurcations. 
Specifically from the point of view of the application, we identify the following different types of chaotic behaviour of the two nanocavities, where the intensity of light inside each cavity fluctuates chaotically with the following additional characterizing properties:
\begin{itemize}
\item 
\emph{non-switching chaotic behavior}, where one cavity is dominant throughout, that is, traps more light and, hence, has a higher intensity; 
\item 
\emph{chaotic behavior with chaotic switching}, where there are epochs when one cavity traps more light, with irregular, unpredictable switching between which cavity is dominant;
\item 
\emph{chaotic behavior with regular switching}, characterized by regular, predictable switching between which cavity is dominant; and
\item 
\emph{chaotic behavior with intermittent regular and chaotic switching}, which features an irregular alternation of epochs with regular and with chaotic  switching  between which cavity is dominant. 
\end{itemize}

These results rely on state-of-the-art computational techniques that have been developed only quite recently. The computations of global bifurcations presented here are implemented in and performed with the pseudo-arclength continuation package \textsc{Auto} \cite{Doe1,Doe2} and its extension \textsc{HomCont} \cite{san2}.  More specifically, global manifolds are computed with a two-point boundary value problem setup \cite{Call1, Kra2}, and homoclinic and heteroclinic orbits are found with an implementation of Lin's method \cite{KraRie1,Kirk2}. The sweeping in the parameter plane to determine kneading sequences and Lyapunov exponents is done with the software package \textsc{Tides}~\cite{TIDES2012}. Finally, visualization and post-processing of the data are performed with \textsc{Matlab}\textsuperscript{\textregistered}.

The paper is organized as follows. In \cref{sec:Background}, we first discuss the symmetries of the semiclassical Bose--Hubbard dimer model and present the local bifurcation analysis of the symmetric equilibria. Here we provide in \cref{sec:LocBifSymm} analytical expressions for the saddle-node and pitchfork bifurcations as parameterized curves in the $(f,\delta)$-plane, which also yield expressions for associated bifurcations of higher codimension, while \cref{sec:OtherLocBifs} discusses Hopf bifurcation and properties of bifurcating periodic orbits. \Cref{sec:kappap2} is devoted to the identification of the different bifurcations that lead to the emergence of an asymmetric chaotic attractor and its subsequent symmetry increasing bifurcation. We first present one-parameter bifurcation diagrams in $f$ at different chosen  $\delta$-values for fixed $\kappa = 2$, and then continue the codimension-one bifurcations thus found to provide an initial bifurcation diagram in the $(f,\delta)$-plane parameter plane. The latter clearly identifies the homoclinic flip bifurcation and the Bykov T-point as organising centers for the curves of bifurcations thus obtained.

The kneading sequence, associated with the one-dimensional unstable manifold of a symmetric saddle focus, is introduced in \cref{sec:KneadSeq}. We then systematically construct the bifurcation diagram near the Bykov T-point by computing successive Shilnikov bifurcation curves that emanate from it and bound associated regions with finite kneading sequence. In \cref{sec:EtoPKneading}, we extend the bifurcation diagram near the Bykov T-point by computing curves of codimension-one heteroclinic orbits between a saddle equilibrium and different saddle periodic orbits, which we refer to as EtoP connections. Here \cref{sec:IsolasGamma_o} and \cref{sec:IsolasGamma_symm} concern EtoP connections to different types of periodic orbits, and we clarify their role in bounding certain regions with constant kneading sequences in the $(f,\delta)$-plane; in particular, we identify a new type of symmetric periodic orbit that, after a symmetry breaking bifurcation, results in a period-doubling route to new types of chaotic attractors. 

\Cref{sec:TangCycSadd} concerns degenerate singular cycles to a saddle focus; the associated tangency bifurcation of its three-dimensional stable manifold with the two-dimensional unstable manifolds of different saddle periodic orbits are discussed in \cref{sec:KneadSeqGamma_o} and \cref{sec:KneadSeqGamma_symm}, respectively. Here we show that associated curves of such tangency bifurcations not only delimit regions in the $(f,\delta)$-plan where Shilnikov bifurcations can be found, but are also responsible for symmetry increasing bifurcations and interior crises of the repesctive chaotic attractors. Moreover, we identify points of codimension-two degenerate singular cycles \cite{Kirk1,Lohr1}. This type of analysis is extended in \cref{sec:TanChainCycles} to tangencies between saddle periodic orbits, where the three-dimensional stable manifold and the two-dimensional unstable manifold of two saddle periodic orbits intersect tangentially. We show that these tangencies delimit regions in the $(f,\delta)$-plane where curves of codimension-one EtoP connections can be found; additionally, we find codimension-two points of degenerate singular cycles, but now involving two periodic orbits, which explain observed accumulations of the curves of EtoP connections. 

In \cref{sec:RegionDegSin}, we complement our analyses with plots of the $(f,\delta)$-plane that show the maximum Lyapunov exponent associated with the attractor that is approached by the one-dimensional unstable manifold of the saddle focus that also defines the kneading sequence. In this way, we identify parameter regions where different chaotic behavior arises, which highlights an intriguing overall structure of global bifurcations that explains the existence of and transition between different, large chaotic regions. Particularly, we find that a codimension-two degenerate singular cycle point from \cref{sec:TangCycSadd} is responsible for additional tangencies of homoclinic bifurcations of saddle periodic orbits that explain a boundary crisis transition of different chaotic attractors. We present a local unfolding of this point in the vector field~\cref{eq:Couplednondim}, whose unfolding has been theorized and numerically studied in a two-dimensional diffeomorphism model in \cite{Kirk1}; in particular, when only one parameter is varied close to this point, we find a snaking curve of homoclinic orbits to a saddle periodic orbit, which is a phenomenon reminiscent of the snaking curve of periodic orbits approaching a Shilnikov bifurcation \cite{Glendinning1984,Shil5}.

In the concluding \cref{sec:conclusions}, we summarize our results and present a final figure that shows how different curves of global bifurcations emerge from the $\Z_2$-equivariant homoclinic flip bifurcation point and the Bykov T-point, respectively, and cross each other transversely. Indeed, this illustrates succinctly that the overall organization of the rather complicated bifurcation diagram in the $(f,\delta)$-plane for positive coupling $\kappa$ is generated by a flip bifurcation meeting a T-point in a four-dimensional phase space. Some avenues for future research are also discussed briefly.

\section{Background and local bifurcations of symmetric equilibria}
\label{sec:Background}
We now consider different representations of system~\eref{eq:Couplednondim}, and derive conditions and formulas for the local bifurcations. Throughout this section we do not fix $\kappa$, that is, we consider it to be an arbitrary real number. To start, it is convenient to write system~\eref{eq:Couplednondim} as a real vector field.  This can be achieved in two ways. Firstly, splitting the complex variables $A$ and $B$ into their real and complex parts $A=A_r+i A_c$ and $B=B_r+i B_c$ gives the vector field 
\begin{equation} \label{eq:CoupledRC}
X_{rc}:
\begin{cases} 
    \dot{A}_r & =  -(\delta + A_r^2 + A_c^2)A_c - A_r - \kappa B_c + f,\\[2mm]
    \dot{A}_c & =  {\color{white} +}(\delta + A_r^2 + A_c^2)A_r - A_c + \kappa B_r ,\\[2mm]
    \dot{B}_r & =  -(\delta + B_r^2 + B_c^2)B_c - B_r  - \kappa A_c + f,\\[2mm]
    \dot{B}_c & =  {\color{white} +}(\delta + B_r^2 + B_c^2)B_r - B_c  + \kappa A_r ,
\end{cases}
\end{equation}
for the real variables $(A_r,A_c,B_r,B_c) \in \mathbb{R}^4$. Its Jacobian is
\begin{equation} \label{eq:CoupledRCJac}
DX_{rc} = 
 \left(\begin{array}{cccc} -2\,A_r\,A_c-1 & 
                                                -{A_r}^2-3\,{A_c}^2-\delta \
         &  0  & -\kappa \\
         3\,{A_r}^2+{A_c}^2+\delta \  & 2\,A_r\,A_c-1 & \kappa  &
                                                                        0
         \\ 0  & -\kappa  &
                                                      -2\,B_r\,B_c-1
                                 & -{B_r}^2-3\,{B_c}^2-\delta \ \\
         \kappa  & 0  &
                                                  3\,{B_r}^2+{B_c}^2+\delta \
                                 & 2\,B_r\,B_c-1 \end{array}\right),
\end{equation} 
and it follows that $\text{div}(X_{rc}) = \text{trace}(DX_{rc}) \equiv -4$ . Hence, the flow of system~\eref{eq:CoupledRC} is uniformly volume contracting at all points of phase space and for all values of the parameters. It follows that the sum of the eigenvalues of any equilibrium and the sum of the Lyapunov exponents along any trajectory of system~\eref{eq:CoupledRC} is always equal to $-4$. 

Alternatively, with $A=R_Ae^{i\phi_A}$ and $B=R_Be^{i\phi_B}$,  system~\eref{eq:Couplednondim} can be written in polar coordinates as
\begin{equation} \label{eq:CoupledAP}
X_{\rm pol}:
\begin{cases} 
    \dot{R}_A & =  -R_A + \kappa R_B \sin(\phi_A-\phi_B) + f\cos(\phi_A), \\[2mm]
    \dot{\phi}_A & =  {\color{white} +}R_A^2 + \delta + \dfrac{\kappa R_B \cos(\phi_A-\phi_B) -f\sin(\phi_A)}{R_A},  \\[2mm]
    \dot{R}_B & =  -R_B + \kappa R_A \sin(\phi_B-\phi_A) +   f\cos(\phi_A), \\[2mm]
    \dot{\phi}_B & = {\color{white} +}R_B^2 + \delta + \dfrac{\kappa R_A \cos(\phi_B-\phi_A) -f\sin(\phi_B)}{R_B}.
\end{cases}
\end{equation}
for the real variables $(R_A,\phi_A,R_B,\phi_B) \in \mathbb{R}^4$. We remark that here $\mathbb{R}^4$ is actually the covering space of the phase space $\mathbb{R} \times \mathbb{S} \times \mathbb{R} \times \mathbb{S} \equiv \mathbb{R}^2 \times \mathbb{T}^2$ of $X_{\rm pol}$. 

The $\Z_2$-equivariance of system~\eref{eq:Couplednondim} manifests itself as the invariance of systems~\eref{eq:CoupledRC} and~\eref{eq:CoupledAP} under the linear transformation 
\begin{eqnarray*}
\eta: & (A_r, A_c, B_r, B_c) & \mapsto (B_r, B_c, A_r, A_c), \\
& (R_A, \phi_A,R_B,\phi_B) & \mapsto (R_B, \phi_B,R_A,\phi_A).
\end{eqnarray*}
Hence, if $\psi(t)$ is a trajectory of system~\eref{eq:CoupledRC} or~\eref{eq:CoupledAP} then $\eta{(\psi)}$ is also a trajectory. The \emph{fixed-point subspace} of \eref{eq:CoupledRC} under the group $<\eta> \ = \Z_2$ is given by
\begin{equation*}
\text{Fix}(\eta) = \{(A_r,A_c,B_r,B_c) \in \R^4: A_r=B_r \text{ 
  and } A_c=B_c\},
\end{equation*}
and that of \eref{eq:CoupledAP} by
\begin{equation*}
\text{Fix}(\eta) = \{(R_A\phi_A,R_B,\phi_B) \in \R^4: R_A=R_B \text{ 
  and } \phi_A= \phi_B\}.
\end{equation*}
We refer to (sets of) solutions that are mapped to themselves by $\eta$ as \emph{symmetric}; otherwise, we refer to them as \emph{asymmetric} and they then come in pairs. In particular, solutions in the two-dimensional invariant plane $\text{Fix}(\eta)$ are symmetric solutions that have identical intensities $R_A(t)^2 = |A(t)|^2= |B(t)|^2 = R_B(t)^2$  and phase $\phi_A(t)=\phi_B(t)$ for all time.

\subsection{Parameterizations of local bifurcation loci}
\label{sec:LocBifSymm}

From the polar form~\eref{eq:CoupledAP} it follows that symmetric equilibria in $\text{Fix}(\eta)$ satisfy
\begin{equation} \label{eq:zeroProblem}
\begin{aligned}
    f\cos (\phi )- r &=& 0, \\
    r^2-\frac{ f \sin (\phi )}{r}+\delta +\kappa &=& 0, 
\end{aligned}
\end{equation}
where $r=R_A=R_B$ and $\phi=\phi_A=\phi_B$. It follows that these equilibria are the roots of the cubic equation 
\begin{equation} \label{eq:zeror2}
    I^3 + 2(\delta +\kappa) I^2 + [(\delta +\kappa)^2+1]I   - f^2 = 0 
\end{equation}
for the intensity $I=r^2$. Hence, there exist generically either one or three equilibria in $\text{Fix}(\eta)$, and their only codimension-one bifurcations are saddle-node and pitchfork bifurcations. They can be found analytically as parametrized curves \textbf{S} and \textbf{P} in the $(f,\delta)$-plane, which also gives expressions for other associated bifurcations of higher codimension. 

\begin{proposition}
\label{lem:local}
{\rm [Local bifurcations in $(f,\delta)$-plane of equilibria in $\text{Fix}(\eta)$]}
\begin{itemize}
\item[(i)]
The curve \textbf{S} of saddle-node bifurcation is parameterized by the phase $\phi$ as
\begin{equation} 
\label{eq:parProblem}
\left(f_{\rm \mathbf{S}} (\phi) ,\delta_{\rm \mathbf{S}}(\phi) \right) 
= \left( \sqrt{ \frac{- 1}{2 \sin (\phi ) \cos^3(\phi)}}, 
 \frac{1}{2} (3 \tan (\phi )+\cot (\phi )) - \kappa \right),
\end{equation}
where $\phi \in \left(-\frac{\pi}{2},0\right)$. 
\item[(ii)]
A cusp bifurcation point $\mathbf{CP}$ is found for $\phi = -\frac{\pi}{6}$ at 
\begin{equation}
\label{eq:cusp}
\left(f _{\mathbf{CP}},\delta_{\mathbf{CP}} \right) = \left (\frac{2 \sqrt{2}}{3^{3/4}}, -\sqrt{3} - \kappa\right).
\end{equation}
In particular, a saddle-node bifurcation can only occur for $f>2\frac{\sqrt{2}}{3^{3/4}}$ and $\delta< -\sqrt{3}-\kappa$.  
\item[(iii)]
The curve \textbf{P} of pitchfork bifurcation is parameterized by the phase $\phi$ as
\begin{equation} \label{eq:parProblemPitch}
\begin{aligned}
\left(f_{\rm \mathbf{P}} (\phi) ,\delta_{\rm \mathbf{P}}(\phi) \right) 
= \left(\frac{\sqrt{\cos^2(\phi)+(2\kappa\cos(\phi)-\sin(\phi))^2}}{\cos^2(\phi)\sqrt{4 \kappa -2  \tan (\phi )}}, \frac{10 \kappa  \tan (\phi ) - 8 \kappa ^2 - 1 - 3 \tan ^2(\phi )}{4 \kappa -2   \tan (\phi )} \right),
\end{aligned}
\end{equation}
where $\phi \in \left( -\frac{\pi}{2}, \arctan(2\kappa)\right)$. A necessary condition for a pitchfork bifurcation to occur is $\delta< -\sqrt{3}+\kappa$.  
\item[(iv)]
A codimension-two saddle-node-pitchfork bifurcation point $\mathbf{SP}$ is found at
\begin{equation} \label{eq:cod2SP}
    (f_{\mathbf{SP}}(\kappa),\delta_{\mathbf{SP}}(\kappa)) =\left(\sqrt{2(1+\kappa^2)(-\kappa+\sqrt{1+\kappa^2})},-2\sqrt{1+\kappa^2}\right).
\end{equation}
\item[(v)]
The curves $\mathbf{S}$ and $\mathbf{P}$ coincide for $\kappa=0$. However, the points $\mathbf{CP}$ and $\mathbf{SP}$ do not coincide for general $\kappa$, except at the codimension-three point $\mathbf{CSP}$ given by 
\begin{equation} \label{eq:cod3}
\left(f_{\mathbf{CSP}},\delta_{\mathbf{CSP}}, \kappa_{\mathbf{CSP}} \right)
= \left (\frac{2 \sqrt{2}}{3^{3/4}}, \dfrac{- 4\sqrt{3}}{3}, \dfrac{\sqrt{3}}{3} \right).
\end{equation}
\end{itemize}
\end{proposition}

\begin{proof}
The determinant of \eref{eq:CoupledRCJac} is
\begin{equation}
\label{eq:DetFull}
\det{\left( DX_{rc} \right)} = \left( 1+(\delta + \kappa)^2+4(\delta + \kappa)I +3I^2 \right)  \left( 1 +(\delta - \kappa)^2 + 4(\delta - \kappa)I + 3I^2 \right).
\end{equation}
The first factor is the determinant of the two-dimensional restriction of system~\eref{eq:CoupledRC} to the fixed-point subspace $\text{Fix}(\eta)$. Hence, its roots with $I$ given by \eref{eq:zeror2} correspond to saddle-node bifurcation inside $\text{Fix}(\eta)$, which provides an implicit equation for the locus \textbf{S}. To obtain a parameterization, one solves \cref{eq:zeroProblem} for $r^2$ as a function of $\phi$ to replace $I=r^2$ in the first factor of \cref{eq:DetFull}; solving for $f$ and $\delta$ as functions of $\phi$ gives the expressions in (i). This and subsequent quite intricate calculations have been performed with a computer algebra package.

Considering roots of the derivatives with respect to $\phi$ shows that $f_{\rm  \mathbf{S}}(\phi)$ has a unique minimum and $\delta_{\rm \mathbf{S}}(\phi)$ has a unique maximum, simultaneously for $\phi=-\frac{\pi}{6} \in[-\frac{\pi}{2},0]$; this is therefore the cusp point $\mathbf{CP}$ and evaluating $(f_{\rm  \mathbf{S}}(-\frac{\pi}{6}),\delta_{\rm \mathbf{S}}(-\frac{\pi}{6}))$ gives (ii).

The second factor of \eref{eq:DetFull} contains the information about the stability in the directions transverse to $\text{Fix}(\eta)$, and its roots with $I$ given by \eref{eq:zeror2} correspond to pitchfork bifurcations; this provides an implicit equation for the locus \textbf{P}. Its parameterization in (iii) is obtained by using instead $I=r^2$ as obtained from \cref{eq:zeroProblem} in the second factor of \cref{eq:DetFull}, and subsequently solving for $f(\phi)$ and $\delta(\phi)$. Note that the parametrisation in \eref{eq:parProblemPitch} is valid for any $\kappa \in \R$ and that the range of $\phi$ follows from $4 \kappa -2  \tan (\phi )>0$, that is, the requirement that the quantity inside the square root be positive. Fixing $\delta_{\rm \mathbf{P}}(\phi)=\delta$ in \eref{eq:parProblemPitch} and solving for $\tan(\phi)$ gives
$\tan(\phi)= (\delta+5\kappa \mp \sqrt{-3+(\delta-\kappa)^2})/3$. The requirement that the quantity inside the square root be positive gives the stated bound on $\delta$, and one checks that then $4 \kappa -2  \tan (\phi )$ is always positive.

Statement (iv) follows from equating the parameterizations~\eref{eq:parProblem} and \eref{eq:parProblemPitch}. Moreover, setting $\kappa=0$ and simplifying shows that these parameterizations for the curves \textbf{S} and \textbf{P} are identical for this special case. On the other hand, for $\kappa=0$ in \cref{eq:cusp} and \cref{eq:cod2SP} the respective points $\mathbf{CP}$ and $\mathbf{SP}$ are still different in general; the condition $\kappa = \sqrt{3}/3$ follows from equating them when $f=2\sqrt{2} / 3^{3/4}$, which completes (v).
\end{proof}

The case $\kappa = 0$ when the curves \textbf{S} and \textbf{P} coincide is special: the two resonators are actually uncoupled, and the equations, hence, describe a single passive Kerr cavity --- a system that has been studied extensively in both the classical and quantum regimes \cite{BinMahmud, DrummondWalls1980}. According to \cref{lem:local}, depending on whether the detuning $\delta$ is above or below the value $\delta_{\mathbf{CP}} = -\sqrt{3}$ of the curve \textbf{S}, one either finds for $\kappa = 0$ a unique stable equilibrium for all values of the forcing $f$, or a hysteresis loop with a region of bistability between the values of $f$ corresponding to the two branches of the curve \textbf{S}; this agrees with \cite{DrummondWalls1980} where the existence of bistability  was established.

In the subsequent sections, we will discuss the bifurcations and dynamics found for positive coupling $\kappa$. Notice from the parametrisation \cref{eq:parProblem} that changing $\kappa$ simply translates the saddle-node bifurcation curve \textbf{S} in the $(f,\delta)$-plane along the $\delta$-axis.

\subsection{Hopf bifurcation and bifurcating periodic orbits}
\label{sec:OtherLocBifs}

There cannot be any symmetric periodic orbit $\Gamma$ contained in $\text{Fix}(\eta)$ because the restriction of system~\eref{eq:CoupledRC} in this space is area contracting (it has divergence $-2$). This implies that symmetric equilibria cannot have a Hopf bifurcation. 

As we will see, Hopf bifurcations can be found for asymmetric equilibria.
Given a periodic orbit $\Gamma \not\subset \text{Fix}(\eta)$ with (minimal) period $P$ there are two different cases \cite{Kuz1}.
\begin{enumerate}
\item If $\Gamma$ is mapped to itself as a set under the symmetry transformation $\eta$. Here, $\eta$ acts as the symmetry $\eta(\Gamma)(t) = \Gamma (t+P/2)$, and $\Gamma$ is called an S-invariant periodic orbit.
\item Otherwise we refer to $\Gamma$ as an asymmetric periodic orbit, and denote its symmetric counterpart $\Gamma^* = \eta(\Gamma)$.
\end{enumerate}

While it is possible to determine local bifurcations of asymmetric equilibria analytically, the resulting expressions for saddle-node and Hopf bifurcations are very unwieldy and impractical. In what follows, we compute them by means of numerical continuation techniques implemented in the software package \textsc{Auto} \cite{Doe2}. Solving for these (only implicitly given) local bifurcations in this way allows us to readily switch to finding and continuing bifurcating periodic orbits and their different types of global bifurcations (which do not have analytic expressions to begin with).

\section{Bifurcation structure for positive coupling}
\label{sec:kappap2}

Photonic crystals with a positive value of $\kappa$ can be fabricated and studied experimentally \cite{Hamel2015}. We consider from now on the value of $\kappa=2$ to illustrate the dynamics displayed by the two microresonators for this coupling configuration. Since it is common in experiments to keep the detuning $\delta$ fixed and ramp the forcing $f$ up and down, we start by presenting in \cref{fig:bifSetPlus} one-parameter bifurcation diagrams in $f$ of system~\eref{eq:Couplednondim} for four (decreasing) values of $\delta$. Branches of equilibria and periodic solutions are presented here in terms of their (maximal) values of the two intensities $|A|^2$ and $|B|^2$. 

\begin{figure}
\centering
\includegraphics[scale=0.95]{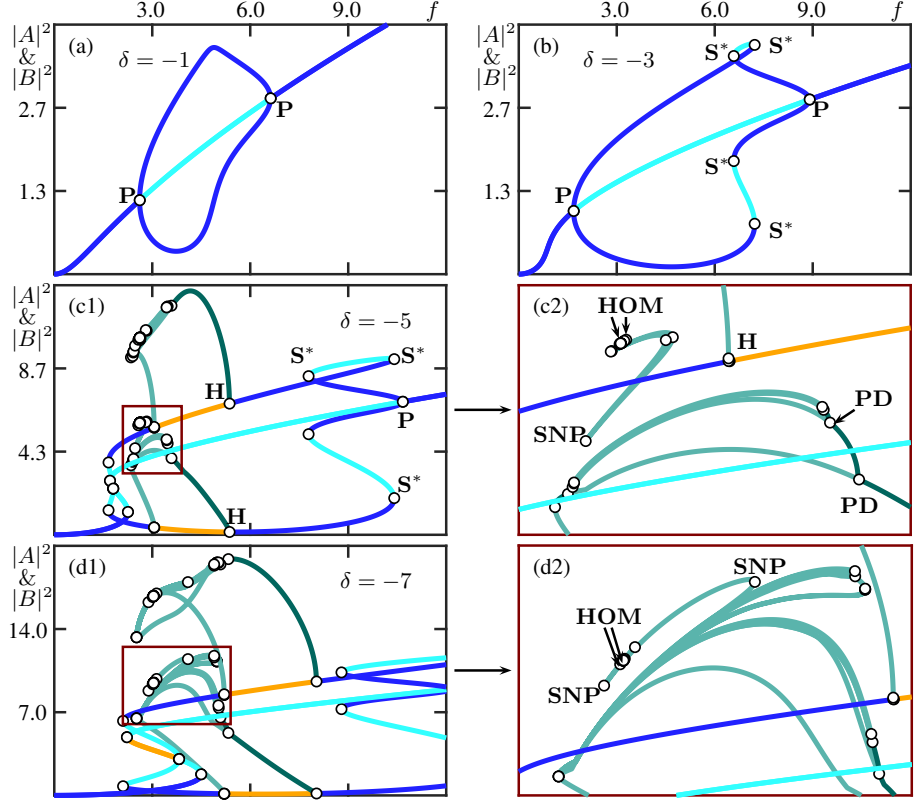}
\caption{One-parameter bifurcation diagrams in $f$ of system~\eref{eq:Couplednondim} for $\kappa=2$ at different values of $\delta$ as stated, where solutions are represented by the intensities $|A|^2$ and $|B|^2$. Shown are stable equilibria (blue), saddle equilibrium with a single unstable eigenvalue (cyan) and with two unstable eigenvalues (orange), stable periodic solutions (dark green), and saddle periodic solution with one unstable Floquet multiplier (light green). Regions in red frames are enlarged in corresponding subpanels.} 
\label{fig:bifSetPlus} 
\end{figure} 

For sufficiently large positive detuning $\delta$ there is a single branch of stable symmetric equilibria. As $\delta$ is decreased sufficiently, as in \cref{fig:bifSetPlus}(a) for $\delta = -1$, one first finds that two pitchfork bifurcation points $\mathbf{P}$ give rise to a pair of branches of asymmetric equilibria over a growing and considerably large $f$-range. When $\delta$ is decreased further, as in \cref{fig:bifSetPlus}(b) for $\delta = -3$, two points of saddle-node bifurcations of asymmetric equilibria, labelled  $\mathbf{S^*}$, give rise to a growing hysteresis loop of asymmetric equilibria. For yet smaller values of $\delta$, as in panels~(c) and (d), there are two points $\mathbf{H}$ of Hopf bifurcation of the asymmetric equilibria, which are connected by a branch of periodic orbits. The right-most point $\mathbf{H}$ is supercritical. It generates stable periodic orbits (for lower values of $f$) that then become a saddle peridic orbit by undergoing a period-doubling bifurcation at the point $\mathbf{PD}$; see the enlargement panel~(c2). The branch then has a fold at a point $\mathbf{SNP}$ of saddle-node bifurcation of periodic orbits at $f \approx 2.2$ and then connects back to the left-most Hopf bifurcation point $\mathbf{H}$, which is subcritical. Notice further that there are many coexisting periodic orbits, which arise from subsequent period doublings; moreover, there is an isola of periodic solutions in the range $f \in [2.3, 3.3]$, which is bounded by further points $\mathbf{SNP}$ of saddle-node bifurcation of periodic orbits. As $\delta$ is decreased even further, as in \cref{fig:bifSetPlus}(d) for $\delta = -7$, this complicated picture of multiple periodic orbits persists, while the isola now appears to be connected to the original branch of periodic solutions arising from the Hopf bifurcation points; see the enlargement panel~(d2).

\begin{figure}
\centering
\includegraphics[scale=0.95]{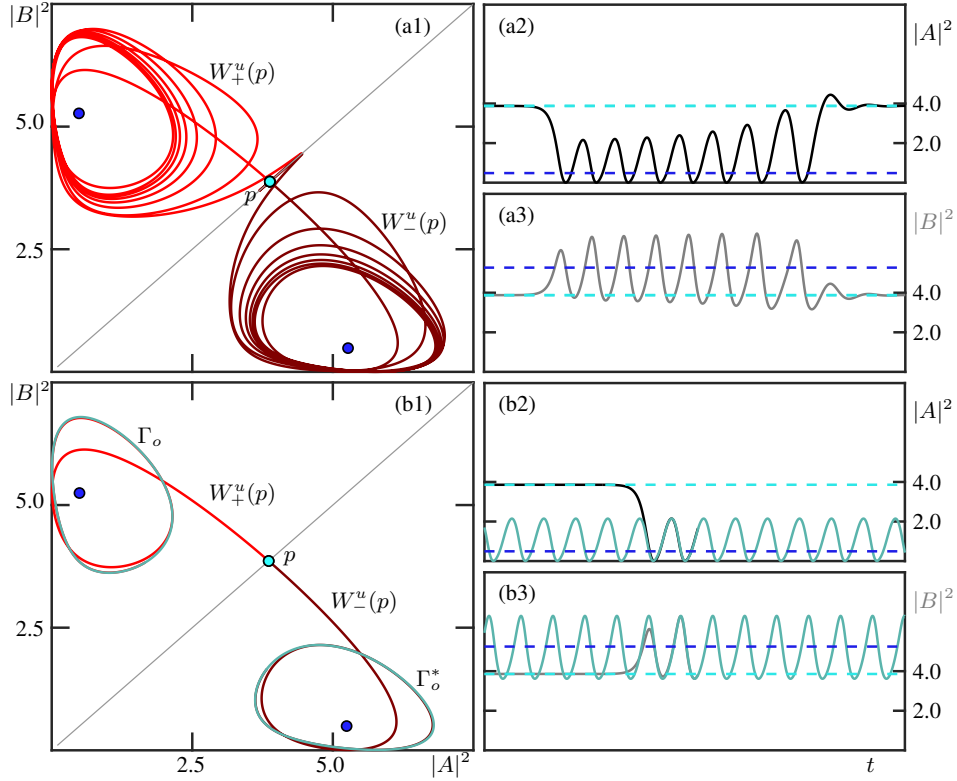}
\caption{Pairs of connecting orbits of system~\eref{eq:Couplednondim} for $\kappa = 2$ and $\delta=-5$, shown in the $(|A|^2, |B|^2)$-plane and as time series of $|A|^2$ and $|B|^2$, respectively. Panels~(a) show the Shilnikov homoclinic orbits of the symmetric saddle equilibrium $p$ at $f \approx 2.6227$, labeled $\mathbf{HOM}$ in \cref{fig:bifSetPlus}(c1); and panels~(b) show the nearby heteroclinic EtoP connections from $p$ to two asymmetric saddle periodic orbits at $f \approx 2.5900$.} 
\label{fig:homOrbits}  
\end{figure} 

A feature in both panels~(c2) and~(d2) of \cref{fig:bifSetPlus} is the existence of points $\mathbf{HOM}$ of homoclinic bifurcation on a branch of periodic orbits. The corresponding limiting situation in phase space for $\delta=-5$ and $f \approx 2.6227$ in \cref{fig:homOrbits}(a) shows that there is a pair of Shilnikov orbits formed by the two branches $W_+^u(p)$ and $W_-^u(p)$ of the one-dimensional unstable manifold $W^u(p)$ of a symmetric saddle equilibrium $p$. Examination of the eigenvalues of $p$ shows that this Shilnikov bifurcation is of chaotic type, meaning that there must be further Shilnikov bifurcations of $p$ nearby \cite{Shil5}. Observe in \cref{fig:homOrbits}(a1) how each homoclinic orbit oscillates eight times around an asymmetric equilibrium before spiraling back into the saddle $p$; moreover, panels~(a2) and (a3) suggest that the homoclinic orbit to $p$ comes very close to a periodic orbit of saddle type. This suggests that there is a codimension-one bifurcation, which we refer to as an EtoP connection, where $W_+^u(p)$ and $W_-^u(p)$ form heteroclinic connections to a pair of saddle periodic orbit $\Gamma_o$ and $\Gamma_o^*$ with three-dimensional stable manifolds $W^s(\Gamma_o)$ and $W^s(\Gamma^*_o)$. We found the pair of EtoP connections with a numerical implementation of Lin's method \cite{KraRie1} at $\delta=-5$ and $f \approx 2.590037$; we refer to this bifurcation as $\mathbf{Hep}$ and it is shown in phase space in \cref{fig:homOrbits}(b). It has been shown in \cite{And2, Krauskopf_2004, RADEMACHER2005390} that an EtoP connection is an accumulation point in parameter space of sequences of $n$-homoclinic bifurcations for increasing $n$, which return to the equilibrium only after $n-1$ near passes. Hence, system~\eref{eq:Couplednondim} features infinitely many Shilnikov bifurcations with an increasing number of loops, which agrees with the fact that the homoclinic bifurcation at $\mathbf{HOM}$ and at nearby points is of chaotic type. 

Note that the connecting orbits shown in \cref{fig:homOrbits} are characterized by the fact that the two branches $W_+^u(p)$ and $W_-^u(p)$ of the one-dimensional manifold $W^u(p)$ remain in the region of the $(|A|^2, |B|^2)$-plane where $|A|^2 < |B|^2$ and where $|B|^2 < |A|^2$, respectively. Since any pair of Shilnikov homoclinic orbits is to the single symmetric equilibrium $p$, they are a mechanism for creating S-invariant periodic orbits and associated highly complicated dynamics, including chaotic attractors; these will be discussed in later sections.

\begin{figure}
\centering
\includegraphics[scale=0.95]{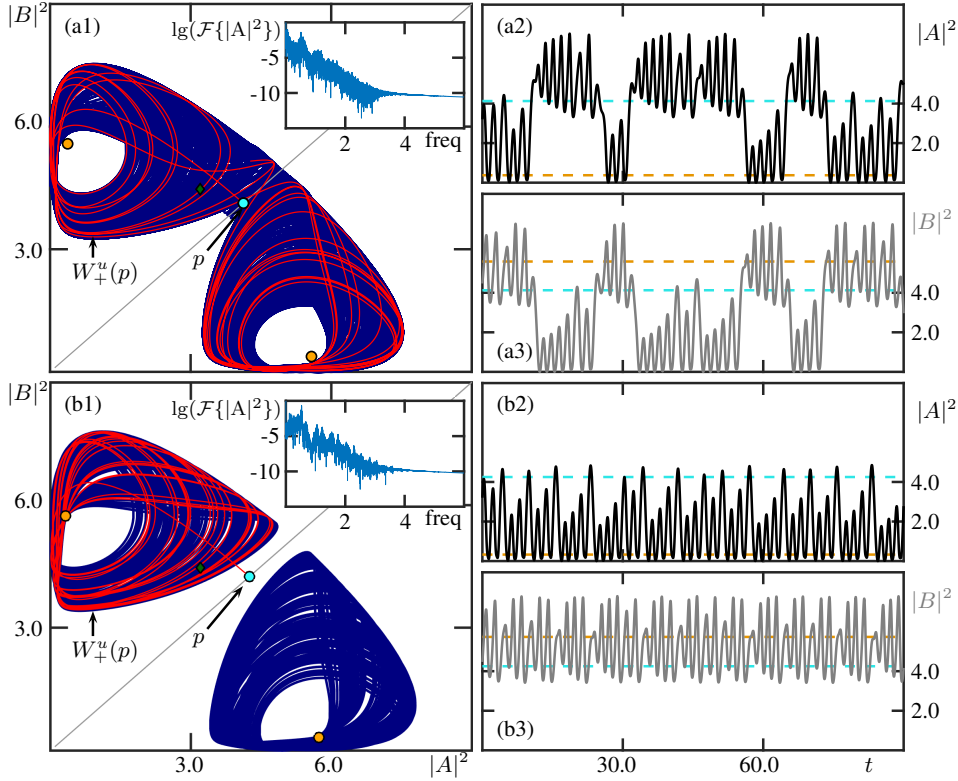}
\caption{Chaotic attractors (blue curves) of system~\eref{eq:Couplednondim} for $\kappa = 2$ and $\delta=-5$ at $f=3.06$ (a) and at $f=3.3$ (b), shown in the $(|A|^2, |B|^2)$-plane together with the branch $W_+^u(p)$ (red curve) and as time series of $|A|^2$ and $|B|^2$, respectively; the insets show the respective power spectrum on a logarithmic scale.}
\label{fig:chaos} 
\end{figure} 

For $\delta=-5$ and $f=3.06$ we find a symmetric chaotic attractor; see \cref{fig:chaos}(a). Here the attractor is represented in panel~(a1) by a long computed trajectory, the power spectrum of the amplitude $A$ of which is shown in the inset; also shown is the branch $W_+^u(p)$ of the symmetric saddle equilibrium $p$, which can clearly be seen to switch across the symmetry line given by $|A|^2 = |B|^2$. Indeed, the time series of $|A|^2$ and $|B|^2$ in panels~(a2) and~(a3) show that the intensity switches irregularly between being higher in cavity A as opposed to cavity B and vice versa. We refer to this type of dynamics as \emph{chaotic behavior with chaotic switching}. For the larger value of $f=3.3$, one finds the pair of chaotic attractors shown in \cref{fig:chaos}(b). Note from panel~(b1) of \cref{fig:chaos} that the branch $W_+^u(p)$ converges to one of the two attractors and, hence, is now confined to the region where $|A|^2 < |B|^2$. The time series of $|A|^2$ and $|B|^2$ in panels~(b2) and~(b3) do not switch between the two cavities; we refer to this chaotic behavior as \emph{non-switching chaotic behavior}.  On the other hand, they have similar short-term characteristics to those in panels~(a2) and~(a3), which explains why the two cases cannot be distinguished by their power spectra; compare the insets of panels~(a1) and~(b1). The transition from \cref{fig:chaos}(a) to \cref{fig:chaos}(b) for increasing $f$ can be interpreted as a symmetry breaking of a chaotic attractor. However, it is more common to think of the transition for decreasing $f$ as a symmetry increasing bifurcation \cite{Chossat1988} of chaotic attractors as is discussed in \cref{sec:TangCycSadd}. Conditions when these may arise have been studied in, for example, \cite{Ben-Tal2002,Dellnitz1995,Heinrich2008, King1992,Melbourne1993}, and such a bifurcation has been shown experimentally to occur in electrical circuits \cite{Ashwin1990, Matsumoto1987}.

\begin{figure}
\centering
\includegraphics[scale=0.95]{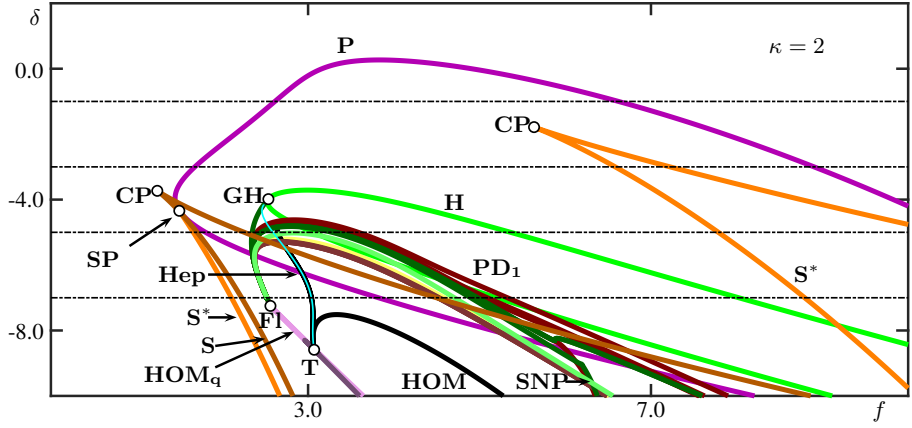}
\caption{Bifurcation diagram in the $(f,\delta)$-plane of system~\eref{eq:Couplednondim} for $\kappa=2$. Shown are curves of saddle-node bifurcations of symmetry equilibria $\mathbf{S}$ (brown) and asymmetric equilibria $\mathbf{S^*}$ (orange), pitchfork $\mathbf{P}$ (purple) and Hopf $\mathbf{H}$ (green) bifurcations, saddle-node bifurcations of periodic solutions $\mathbf{SNP}$ (dark-green), period-doubling bifurcations $\mathbf{PD}$ (red), Shilnikov bifurcations $\mathbf{Hom}$ to symmetric focus equilibria (black), homoclinic bifurcation to real symmetric equilibria (lilac and dark-lilac), and EtoP connections (cyan); also shown are codimension-two points of cusp $\mathbf{CP}$, saddle-node-pitchfork $\mathbf{SP}$, generalised Hopf $\mathbf{GH}$, Bykov T-point $\mathbf{T}$ and homoclinic flip $\mathbf{Fl}$ bifurcations; the horizontal dashed lines correspond to the panels of \cref{fig:bifSetPlus}.} 
\label{fig:bifSetKappa+2} 
\end{figure} 

The fact that the unstable manifold $W^u(p)$ changes quite dramatically between the two cases shown in \cref{fig:chaos}(a) to \cref{fig:chaos}(b) already hints that this transition is associated with different types of connecting orbits involving the symmetric saddle equilibrium $p$. Indeed, families of different types of codimension-one phenomena organize the $(f, \delta)$-plane, as will be discussed in considerable detail in the sections that follow. As a starting point for this discussion, we present in \cref{fig:bifSetKappa+2} a bifurcation diagram in the $(f, \delta)$-plane for $\kappa=2$ with a focus on the loci of the bifurcations that we considered so far; in particular, the slices for fixed $\delta$ shown in \cref{fig:bifSetPlus} are indicated by horizontal dashed lines. In \cref{fig:bifSetKappa+2} the curves $\mathbf{S}$ of saddle-node and $\mathbf{P}$ of pitchfork bifurcation of symmetric equilibria are rendered from the formula in \cref{lem:local}, and the remaining curves have been obtained by numerical continuation of the respective bifurcations. 
\begin{figure}
\centering
\includegraphics[scale=0.95]{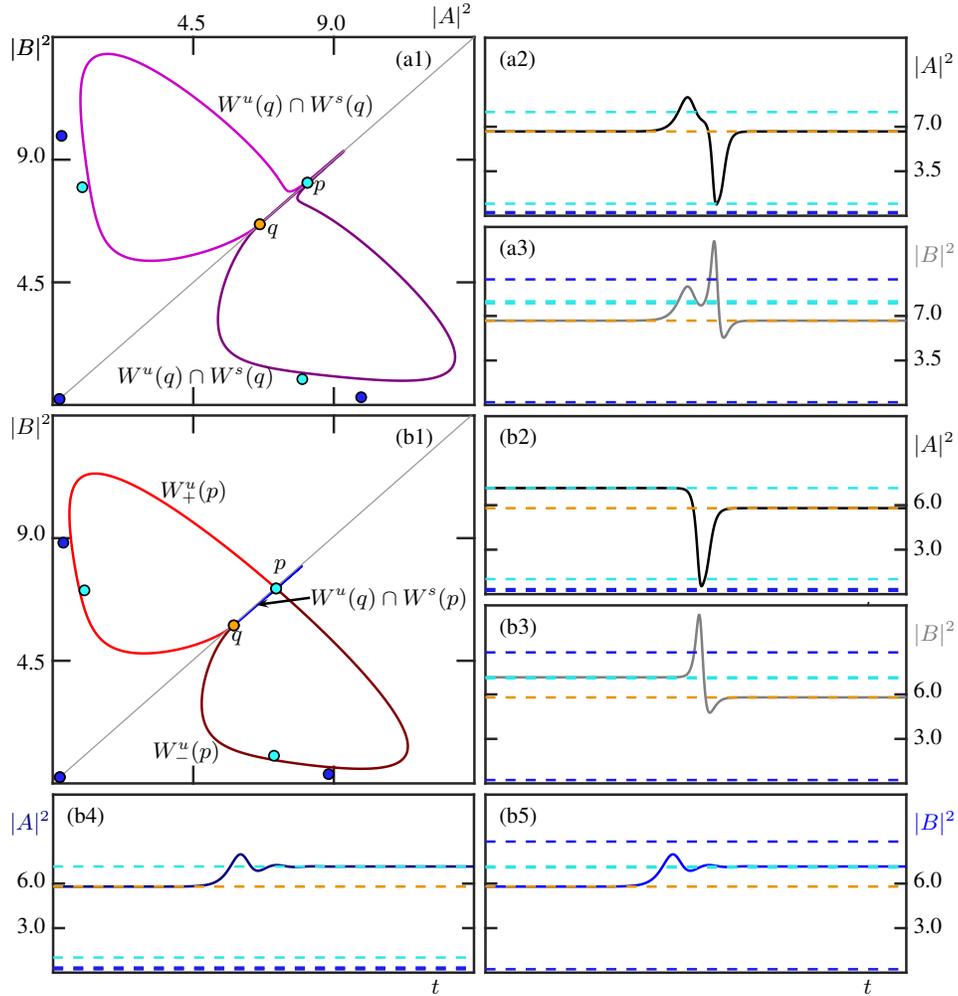}
\caption{Pairs of connecting orbits of system~\eref{eq:Couplednondim} for $\kappa = 2$, represented as in \cref{fig:homOrbits}. Panels~(a) show the homoclinic orbits to the real symmetric saddle $q$ at $\delta=-9.5$ and $f \approx 3.4139$; and panels~(b) show the heteroclinic cycle at the Bykov T-point $\mathbf{T}$ at $\delta \approx -8.5888$ and $f \approx 3.0705$, consisting of a pair of codimension-two connections from $p$ to $q$ (see (b2) and (b3) for the time series) and a single structurally stable connection in $\text{Fix}(\eta)$ from $q$ to $p$ (see (b4) and (b5) for the time series).} 
\label{fig:bykov} 
\end{figure} 

The curve $\mathbf{P}$ in \cref{fig:bifSetKappa+2} encloses a large region in the $(f, \delta)$-plane, and it is the first bifurcation curve that is encountered in one-parameter slices in $f$ for fixed and decreasing $\delta$. The curve $\mathbf{P}$ is tangent at a codimension-two saddle-node pitchfork point $\mathbf{SP}$ to the curve $\mathbf{S}$ of saddle-node bifurcations of symmetric equilibria, which also has a cusp point $\mathbf{CP}$. Within the region bounded by $\mathbf{P}$, one finds the curve $\mathbf{S^*}$ of saddle-node and $\mathbf{H}$ of Hopf bifurcations of asymmetric equilibria. The curve $\mathbf{S^*}$ has a cusp point $\mathbf{CP}$ and on the curve $\mathbf{H}$ there is a codimension-two point $\mathbf{GH}$ of generalized Hopf bifurcation, from which a curve $\mathbf{SNP}$ of saddle-node bifurcations of periodic orbits emerges. Also shown in \cref{fig:bifSetKappa+2} are the curves $\mathbf{SNP}$ and $\mathbf{PD}$ obtained by continuing the saddle-node and period-doubling bifurcations of the different periodic orbits identified in \cref{fig:bifSetPlus}(c) and~(d). Notice that all these curves end up, together with the curve $\mathbf{SNP}$ emerging from $\mathbf{GH}$, at the point labeled $\mathbf{Fl}$ in \cref{fig:bifSetKappa+2}, which is a point of flip bifurcation on the curve $\mathbf{HOM_q}$ where one finds a pair of homoclinic orbits to a symmetric saddle equilibrium $q$ with real eigenvalues. 

A pair of homoclinic orbits along $\mathbf{HOM_q}$ of the saddle $q$, which has two-dimensional stable and unstable manifolds, is shown in \cref{fig:bykov}(a). Note that the connecting orbits pass very close to the symmetric saddle equilibrium $p$. Also shown in \cref{fig:bifSetKappa+2} are the curves $\mathbf{HOM}$ and $\mathbf{Hep}$ of the connecting Shilnikov and EtoP orbits from \cref{fig:homOrbits}, which end up at the point on $\mathbf{HOM_q}$ that is labeled $\mathbf{T}$. This codimension-two point is called a Bykov T-point, and it represents the moment where there exists a heteroclinic cycle between two saddle equilibria \cite{Glendinning1984,HomSan,Knobloch2013}. The heteroclinic cycle is shown in  \cref{fig:bykov}(b), and it consists of a pair of codimension-two connections from $p$ to $q$, the time series of which are presented in panels (b2) and (b3), and of a single structurally stable connection from $q$ to $p$, which lies in $\text{Fix}(\eta)$ and whose time series are presented in panels (b4) and (b5). As the bifurcation diagram in \cref{fig:bifSetKappa+2} already hints at, the codimension-two global bifurcation points $\mathbf{Fl}$ and $\mathbf{T}$ emerge as the main organizing centers for global bifurcations. The next sections are devoted to the study of associated families of different types of homoclinic and heteroclinic bifurcations that organize the overall dynamics of system~\eref{eq:Couplednondim} in the $(f, \delta)$-plane for $\kappa = 2$. As we will see, they are intimately related to transitions between chaotic attractors with different symmetry properties.

\section{Kneading sequences and curves of Shilnikov bifurcations in the $(f, \delta)$-plane}
\label{sec:KneadSeq}

\begin{figure}
\centering
\includegraphics[scale=0.95]{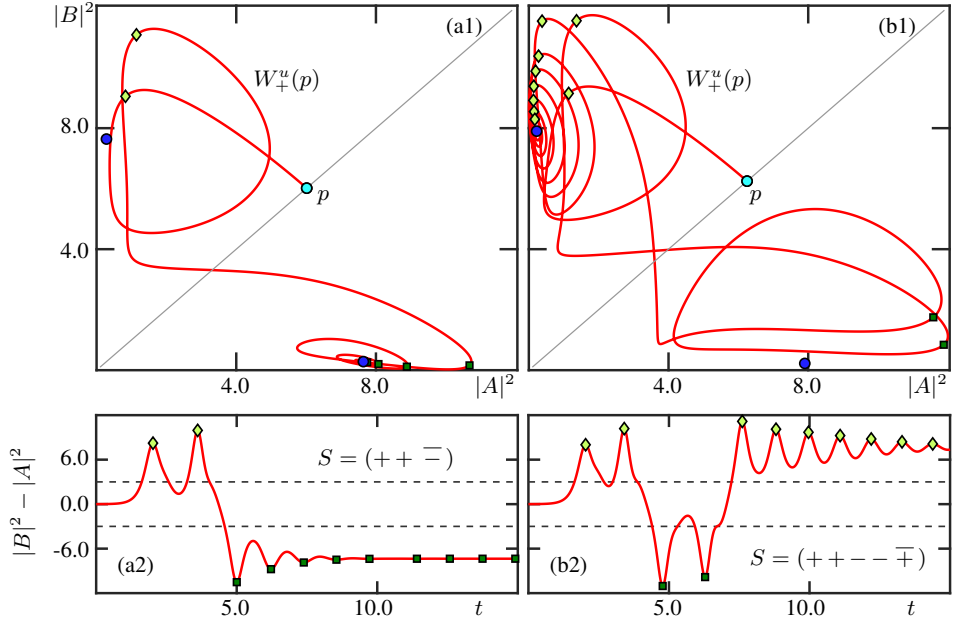} 
\caption{The branch $W^u_+(p)$ (red curve) of the symmetric equilibrium $p$ of system~\eref{eq:Couplednondim} for $\kappa = 2$ and $\delta = -7$ as it spirals to one of a pair of stable asymmetric equilibria for $f = 3.5$ (a) and for $f = 4.0$ (b), shown in the $(|A|^2,|B|^2)$-plane (top row) as temporal trace of $|B|^2-|A|^2$ (bottom row). The maxima (light green rhombi) and minima (dark green squares) of this time series define the respective shown kneading sequence $S$.} 
\label{fig:kneadTrajSample} 
\end{figure} 
We now combine the computation and continuation of a representative number of different types of codimension-one global bifurcations with a parameter sweeping technique that determines regions in the $(f,\delta)$-plane where a topological invariant is constant. More specifically, we consider here the kneading sequence defined by the itinerary of the positive branch $W^u_+(p)$ of the saddle symmetric equilibrium $p$ of system~\eref{eq:Couplednondim} that is involved in (pairs of) Shilnikov homoclinic orbits and EtoP connections of codimension one, as well as in forming the codimension-two connections at the Bykov T-point $\mathbf{T}$. Indeed, it is known that the point $\mathbf{T}$ is responsible for the creation of infinitely many Shilnikov bifurcations and, thus, generates a complicated arrangement of global bifurcations in a parameter plane \cite{Shil3,Shil1}.  In light of the role of the one-dimensional manifold $W^u(p)$ in this, it is natural to consider a symbolic or kneading sequence $S = S(f,\delta)$ that records subsequent maxima and minima of the time series of $|B|^2-|A|^2$ generated by the specific trajectory $W^u_+(p)$. In other words, the kneading sequence $S$ records and distinguishes local oscillations with a dominant intensity of either $|B|^2$ or $|A|^2$ for every point in the $(f,\delta)$-plane. 

More specifically, every maximum with $|A|^2< |B|^2$ of the time series of $|B|^2-|A|^2$ generated by $W^u_+(p)$ is recorded as $+$, and every minimum with $|B|^2<|A|^2$ is recorded as a $-$, thus, generating (generically) a kneading sequence 
\begin{equation}
S = (s_1 \, s_2 \, s_3 \, \ldots) \in \Sigma = \{+,-\}^\N.
\end{equation}
The map from the $(f,\delta)$-plane to the kneading sequence $S = S(f,\delta)$ is well defined and locally constant close to a T-point $\mathbf{T}$. Moreover, it follows from the choice of $W^u_+(p)$ as the generating object that the kneading sequence $S$  always starts with the symbol $+$, that is, $s_1=+$. Indeed, one could consider instead the kneading sequence of the other branch $W^u_-(p)$ which has the opposite symbols due to the reflectional symmetry $\eta$. As is common, we denote by $\overline{+}$ and $\overline{-}$ the infinite repetition of the respective symbol. 

This definition is illustrated in \cref{fig:kneadTrajSample} with two examples where $W^u_+(p)$ converges to an attracting equilibrium in the region with $|A|^2< |B|^2$ in panels~(a) and with $|B|^2<|A|^2$ in panels~(b).  Specifically in panel~(a), $W_+^u(p)$ has two maxima with $|A|^2<|B|^2$, after which the trajectory switches to and remains in the region with $|B|^2<|A|^2$, yielding the kneading sequence $S= (++\overline{-})$. For the case shown in \cref{fig:kneadTrajSample}(b), on the other hand, the trajectory switches to the region with $|B|^2<|A|^2$, has two minima there and then switches back and converges to the attracting equilibrium with $|A|^2<|B|^2$, giving $S= (++--\overline{+})$. 

Notice that these two kneading sequences agree up to including the fourth symbol $s_4$ and then start to differ from their fifth symbol $s_5$. Hence, one may suspect that there are global bifurcations in between the two values $f = 3.5$ and $f = 4.0$ that generate this change in the kneading sequence. As we will see below, this is indeed the case: there is a Shilnikov bifurcation where $W_+^u(p)$ has maxima and minima leading to the finite sequence $S=(++--)$ before it returns to the point $p$.

More generally, the $(f,\delta)$-plane is divided into open regions of given kneading sequences, and a change of kneading sequence requires a global bifurcation involving the point $p$. In particular, any Shilnikov bifurcation of $p$ provides a mechanism that changes the kneading sequence. When system~\eref{eq:Couplednondim} exhibits a Shilnikov bifurcation, there is a finite number of extrema of the time series before $W_+^u(p)$ returns back close to $p$. To avoid picking up tiny extrema near the fixed-point subspace as $W^u_+(p)$ converges back to $p$, we consider and detect only maxima and minima of $|B|^2-|A|^2$ in the trajectory of $W_+^u(p)$ that also satisfy $3 < | |A|^2-|B|^2 |$, that is, are sufficiently far from the dividing case $ |A|^2 = |B|^2$. In this way, we are able to detect the relevant and sufficiently large maxima and minima that determine a finite kneading sequence at the respective homoclinic bifurcation; for notational convenience, we represent such a (nongeneric) finite kneading sequence of length $k$ by $S = (s_1 \, s_2 \, s_3 \, \ldots \, s_k )$, which we complete (in a slight abuse of notation) with an infinite string of zeros as $S = (s_1 \, s_2 \, s_3 \, \ldots \, s_k \, \overline{0})$. This allows us to define consistently for any kneading sequence $S$ the associated kneading invariant
\begin{equation}
\label{eq:kneadingInv}
K = \sum_{i=1}^{\infty}  s_i\,\frac{1}{2^i},
\end{equation} 
where the $s_i$ represent the sign of the corresponding term in the sum. Since $s_1=+$ the kneading invariant $K$ takes values in the interval $[0,1]$. For any given point in the $(f,\delta)$-plane, we only ever consider its kneading sequence up to $n$ symbols for a given $n \in \N$. The finite kneading sequences $S^n \in \Sigma^n = \{+,-\}^n$ define the $n$-cylinders, consisting of all infinite sequences starting with $S^n$. We refer to the kneading invariant of $S^n$ as $K^n$; by construction, the numbers $K^n$ divide the interval $[0,1]$ into $2^{(n-1)}$ subintervals of length $2^{-(n-1)}$, each of which represents a different symbol sequence of length $n$ that is assigned a specific color, as represented in the figures that follow by a discrete color bar.

Increasing $n$ refines the symbol sequences and subdivides this color scheme. We will use this fact to build up an increasingly complex picture of the division of the $(f,\delta)$-plane into regions of detected finite kneading sequences in $\Sigma^n$. More specifically, we perform parameter sweeps with a $1000 \times 1000$ grid in the parameter range shown and determine $S^n$ from the time series of $|B(t)|^2-|A(t)|^2$ of the trajectory $W_+^u(p)$ as computed by numerical integration from a point near $p$ that lies in its one-dimensional unstable eigenspace. We consider here kneading sequences of up to length $n = 12$ with a corresponding color scheme of $2^{11}$ colors defined by $K^n$. To perform these computation we use the software package \textsc{Tides}~\cite{TIDES2012}, which is able to compute $W_+^u(p)$ to high-precision and find the relevant maxima and minima of $|B(t)|^2-|A(t)|^2$ as $t$ increases. 

We remark that parameter sweeping techniques of topological invariants are a quite common way of detecting regions in a parameter plane with different qualitative behavior.  For example, parameter sweeping of kneading invariants has been used to showcase the complicated bifurcation structure near Bykov T-points in the Lorenz and Shimizu-–Morioka systems~\cite{Shil3,Shil1}, and to understand complicated bursting behavior in the Hindmarsh--Rose model \cite{Shil4}. Indeed, the package \textsc{Tides} has been explicitly designed to enable such parameter sweeps. 

An important aspect of the work presented here is that we use parameter sweeps for increasing $n$ to inform us which global bifurcations arise at the increasingly many boundaries between neighboring regions in the $(f,\delta)$-plane. Different families of Shilnikov bifurcations and EtoP connections are then detected and continued as curves. For this purpose, we represent connecting orbits as solutions of suitably defined boundary value problems, which are implemented and solved within the continuation package \textsc{Auto} \cite{Doe2}. In particular, we use a numerical implementation of Lin's method to detect connecting orbits between saddle periodic orbits and equilibria; see \cite{KraRie1} an entry point to this quite general and efficient approach. This combined approach has also been used in \cite{And2}, specifically, to characterize the two-parameter bifurcation diagram of the most complicated case of a homoclinic flip bifurcation known as case~\textbf{C}. This particular case involves infinite families of homoclinic orbits, as well as chaotic dynamics of a vector field model with a three-dimensional phase space \cite{HomSan}. The work presented here is in the same spirit. It uses parameter sweeping and systematic computations of relevant curves of global bifurcations in a complementary and interactive way to build up a comprehensive picture in the $(f,\delta)$-plane of the overall behavior of system~\eref{eq:Couplednondim}.

\subsection{Kneading sequences for increasing $n$}

\begin{figure}
\centering
\includegraphics[scale=0.95]{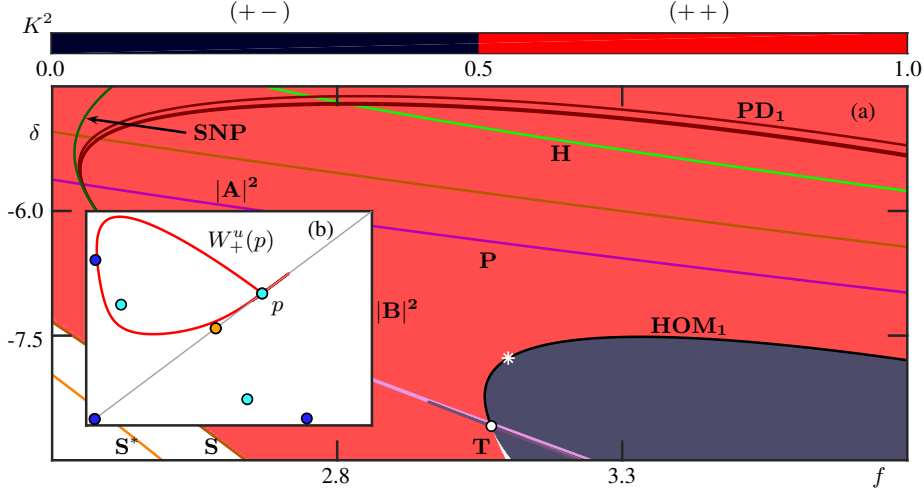} 
\caption{Coloring by finite kneading sequences $S^2$ and associated Shilnikov bifurcation near the Bykov T-point $\mathbf{T}$. Panel~(a) shows the $(f,\delta)$-plane of system~\eref{eq:Couplednondim} with the curves of local bifurcation of equilibria and periodic orbits from \fref{fig:bifSetKappa+2} (not labelled) and additionally the Shilnikov bifurcation curve $\mathbf{HOM_1}$ (black curve). The inset~(b) shows $W^u_+(p)$ in the $(|A|^2,|B|^2)$-plane at $\mathbf{HOM_1}$ for $(f,\delta) \approx (3.1000,-7.7740)$, as indicated by the white asterisk in panel~(a).} 
\label{fig:Knead2} 
\end{figure} 

\Fref{fig:Knead2} shows the result of parameter sweeping with the identification of the two kneading sequences of length two. The $(f,\delta)$-plane in panel~(a) features two regions: that with $S^2=\left(\mathrel{+} \, \mathrel{+} \right)$ and that with $S^2=\left(\mathrel{+} \, \mathrel{-} \right)$, as indicated by the color coding of $K^2$. Superimposed we show the codimension-one local bifurcation curves of equilibria and periodic orbit from \fref{fig:bifSetKappa+2} as well as the codimension-one Shilnikov bifurcation curve $\mathbf{HOM_1}$. It has been obtained by continuation of the homoclinic orbit of the symmetric equilibrium $p$ with one loop, which is shown in panel~(b) for the specific parameter point indicated by the white asterisk in panel~(a). The kneading sequence along the curve $\mathbf{HOM_1}$ is $S=(\mathrel{+} \, \overline{0})$, which indeed implies that it bounds the two colored regions of different kneading sequences. Note that the separation of the $(f,\delta)$-plane into exactly two colored regions provides numerical evidence that no further Shilnikov bifurcations of $p$ with a single loop exist near the Bykov T-point $\mathbf{T}$.

\begin{figure}
\centering
\includegraphics[scale=0.95]{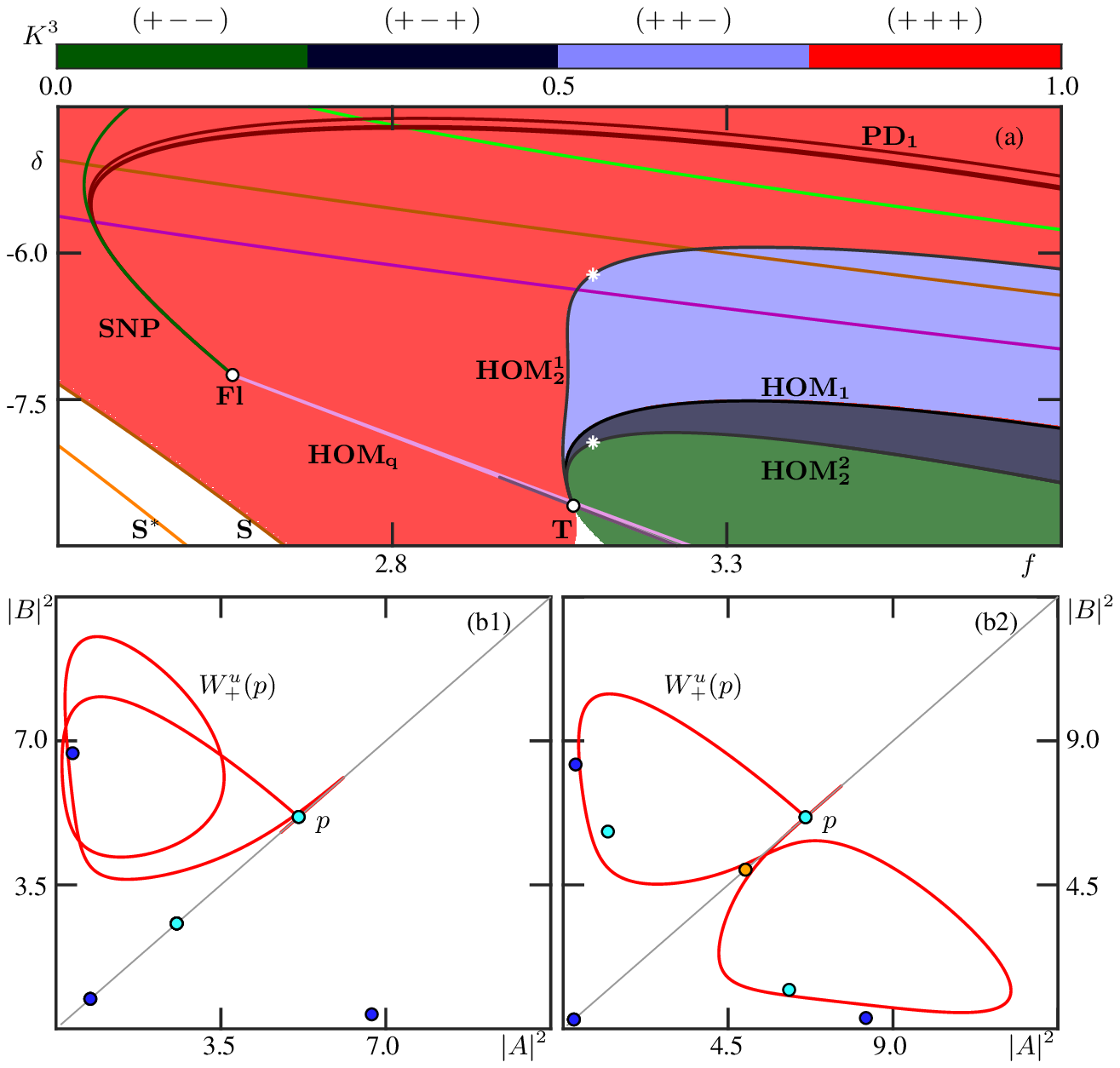} 
\caption{Coloring by finite kneading sequences $S^3$ and associated Shilnikov bifurcations near the Bykov T-point $\mathbf{T}$. Panel~(a) shows the $(f,\delta)$-plane of system~\eref{eq:Couplednondim} with the bifurcation curves from \fref{fig:Knead2} and additionally the Shilnikov bifurcation curves $\mathbf{HOM^1_2}$ and $\mathbf{HOM^2_2}$ (grey curves). The insets~(b1) and~(b2) show $W^u_+(p)$ in the $(|A|^2,|B|^2)$-plane at $\mathbf{HOM^1_2}$ for $(f,\delta) \approx (3.1000,-6.2153)$ and at $\mathbf{HOM^2_2}$ for $(f,\delta) \approx (3.10,-7.9449)$, respectively, as indicated by the white asterisks in panel~(a).} 
\label{fig:Knead2sub2} 
\end{figure} 

\begin{figure}
\centering
\includegraphics[scale=0.95]{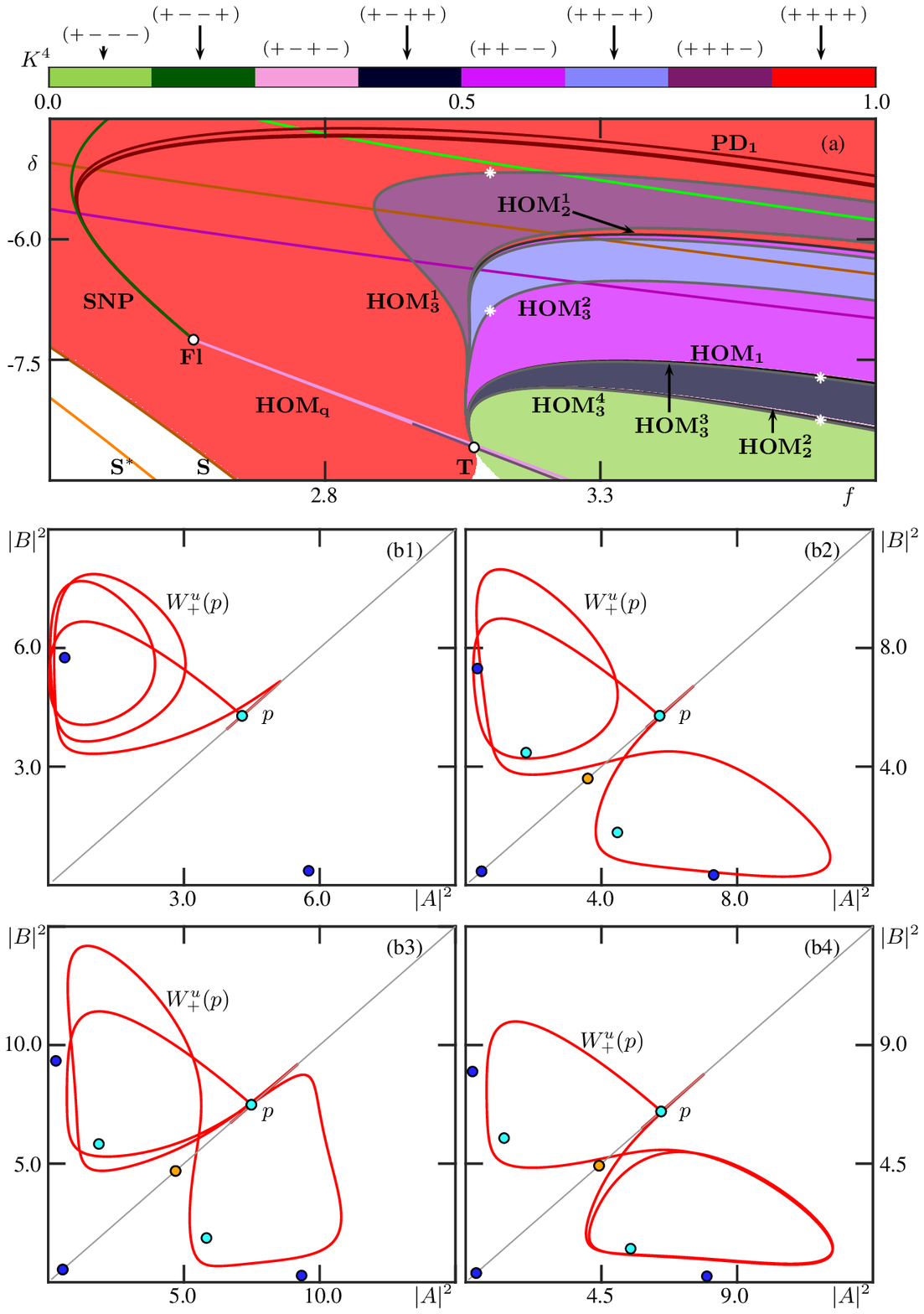} 
\caption{Coloring by finite kneading sequences $S^4$ and associated Shilnikov bifurcations near the Bykov T-point $\mathbf{T}$. Panel~(a) shows the $(f,\delta)$-plane of system~\eref{eq:Couplednondim} with the bifurcation curves from \fref{fig:Knead2sub2} and additionally the Shilnikov bifurcation curves $\mathbf{HOM^1_3}$ to $\mathbf{HOM^4_3}$ (grey curves). The insets~(b1) and~(b4) show $W^u_+(p)$ in the $(|A|^2,|B|^2)$-plane at $\mathbf{HOM^1_3}$ to $\mathbf{HOM^4_3}$, for $(f,\delta) \approx (3.1000,-5.1716)$, $(f,\delta) \approx (3.1000,-6.8865)$, $(f,\delta) \approx (3.7000,-7.7238)$ and $(f,\delta) \approx (3.7000, -8.2547)$, respectively, as indicated by the white asterisks in panel~(a).} 
\label{fig:Knead4}
\end{figure} 

\Fref{fig:Knead2sub2} shows that, when finite kneading sequence $S^3$ are considered, the regions corresponding to $S^2$ are each split up into two subregions. The red region of $S^2=\left(\mathrel{+} \, \mathrel{+} \right)$ in \fref{fig:Knead2}(a) is split in \fref{fig:Knead2sub2}(a) into the regions with $S^3=\left(\mathrel{+} \, \mathrel{+} \, \mathrel{+} \right)$ (red) and $S^3=\left(\mathrel{+} \, \mathrel{+} \, \mathrel{-} \, \right)$ (purple), which are separated by a Shilnikov bifurcation curve $\mathbf{HOM^1_2}$ (dark grey). The branch $W^u_+(p)$ in the $(|A|^2,|B|^2)$-plane at $\mathbf{HOM^1_2}$ in panel~(b1) shows that along this curve one indeed finds the finite kneading sequence $S=\left(\mathrel{+} \, \mathrel{+} \, \overline{0} \right)$. Similarly, the grey region corresponding to $S^2=\left(\mathrel{+} \, \mathrel{-} \right)$ in \fref{fig:Knead2}(a) is split in \fref{fig:Knead2sub2}(a) into the regions with $S^3=\left(\mathrel{+} \, \mathrel{-} \, \mathrel{+} \right)$ (grey) and $S^3=\left(\mathrel{+} \, \mathrel{-} \, \mathrel{-} \right)$ (green) by the Shilnikov bifurcation curve $\mathbf{HOM^2_2}$ (dark grey), which has the finite kneading sequence $S^2=\left(\mathrel{+} \, \mathrel{-} \, \overline{0} \right)$ as is illustrated in \fref{fig:Knead2sub2}(b2).

\begin{figure}
\centering
\includegraphics[scale=0.95]{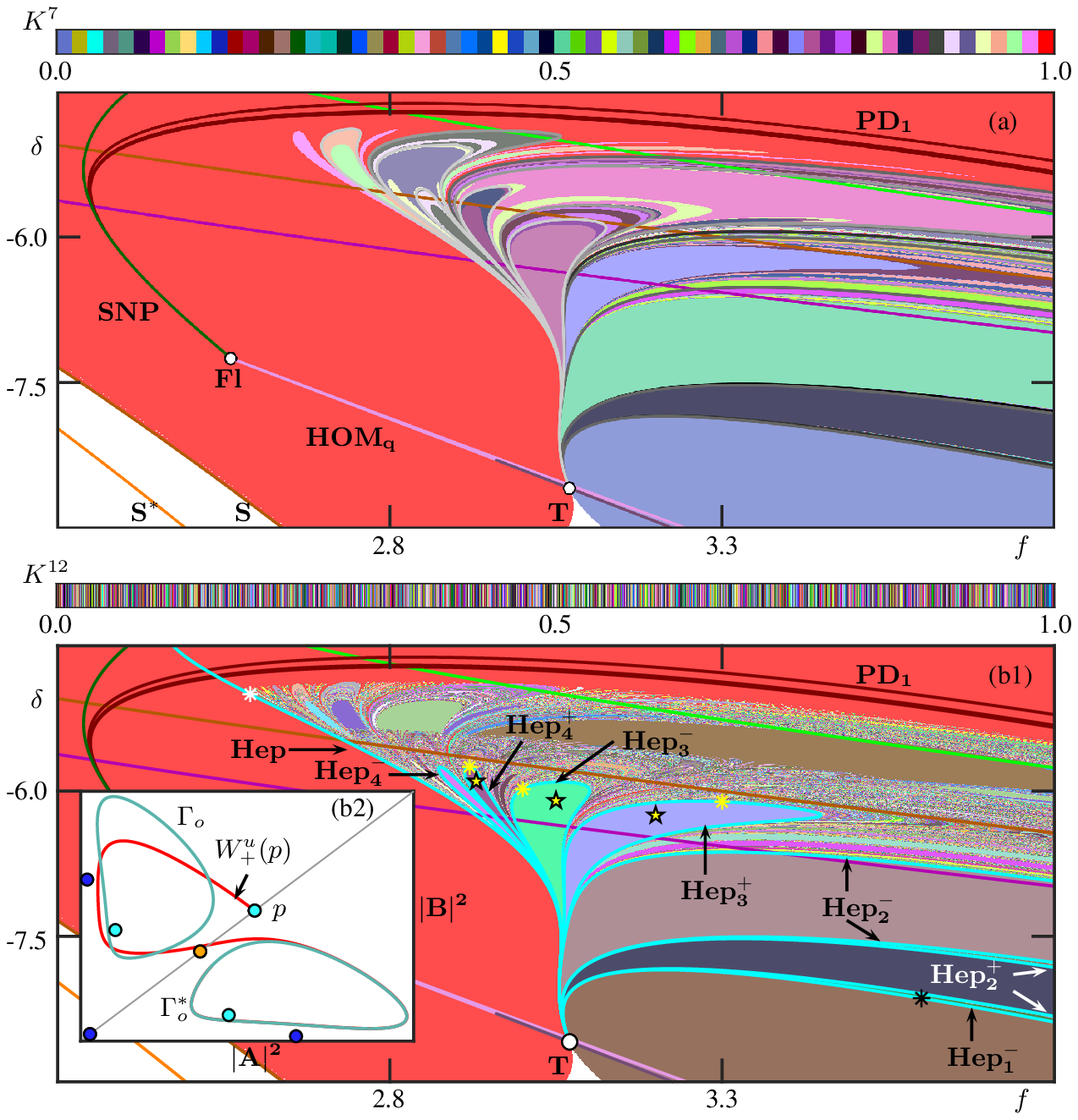} 
\caption{Coloring by finite kneading sequences $S^n$ and representative curves of Shilnikov bifurcations (grey and unlabelled) in the $(f,\delta)$-plane near the Bykov T-point $\mathbf{T}$ for $n=7$ (a) and for $n=12$ (b1); compare with \fref{fig:Knead2}(a) to \fref{fig:Knead4}(a). Panel~(b1) shows additionally curves $\mathbf{Hep}$,  $\mathbf{Hep_1^-}$ and $\mathbf{Hep^{+/-}_i}$, $\mathbf{i} = \mathbf{2} ,..,\mathbf{4}$ of codimension-one heteroclinic EtoP connections between $p$ and the orientable saddle periodic orbits $\Gamma_o$ and $\Gamma_o^*$. In panel~(b1) parameter values are indicated for which the curves $W^u_+(p)$, $\Gamma_o$ and $\Gamma_o^*$ are shown in separate figures, namely: for the white asterisk along the curve $\mathbf{Hep}$ see \fref{fig:homOrbits}(b); and for the yellow asterisks along the curves $\mathbf{Hep_3^+}$,  $\mathbf{Hep_3^-}$ and $\mathbf{Hep_4^+}$ and the yellow stars inside the respective bounded regions with constants kneading sequence see \fref{fig:EtoPSample}. The inset~(b2) shows $W^u_+(p)$ in the $(|A|^2,|B|^2)$-plane at $\mathbf{Hep^-_1}$ at the parameter values indicated by the black asterisk in panel~(b1).} 
\label{fig:BifEtoP}
\end{figure} 

It cannot be seen clearly on the scale of \fref{fig:Knead2sub2}(a), but the purple region with kneading sequence $S^3=\left(\mathrel{+} \, \mathrel{+} \, \mathrel{-} \right)$ is an \emph{isola} within the region with $S^3=\left(\mathrel{+} \, \mathrel{+} \, \mathrel{+} \right)$ (red). That is, there exists another homoclinic bifurcation very close to the curve $\mathbf{HOM_1}$ with the same kneading sequence as $\mathbf{HOM^1_2}$.  Moreover, the region with kneading sequence $S^3=\left(\mathrel{+} \, \mathrel{-} \, \mathrel{+} \right)$ is also an isola that lies within the region with $S^3=\left(\mathrel{+} \, \mathrel{-} \, \mathrel{-} \right)$. In fact, isolas of given finite kneading sequences are a phenomenon that we will encounter repeatedly for increasing $n$; see already \fref{fig:BifEtoP}. 

When considering finite kneading sequence $S^4$, as in \fref{fig:Knead4}, we find again that each subregion of $S^3$ in \fref{fig:Knead2sub2}(a) is split up into two subregions each. This happens again via the transition through additional Shilnikov bifurcation curves, of which there are now four: the grey curves $\mathbf{HOM^1_3}$ to $\mathbf{HOM^4_3}$. As can be checked from panels~(b1)--(b4), the respective Shilnikov bifurcation indeed has the finite symbol sequence given by the three leading symbols of its two neighboring regions. In \fref{fig:Knead4} the curve $\mathbf{HOM^1_3}$ is a single curve that starts and ends at the Bykov T-point $\mathbf{T}$ and bounds the isola with $S^4=\left(\mathrel{+} \, \mathrel{+} \, \mathrel{+} \, \mathrel{-} \right)$ (dark purple) within the region with $S^4=\left(\mathrel{+} \, \mathrel{+} \, \mathrel{+} \, \mathrel{+} \right)$. Similarly, the single curve $\mathbf{HOM^2_3}$ bounds the isola with $S^4=\left(\mathrel{+} \, \mathrel{+} \, \mathrel{-} \, \mathrel{+} \right)$ within the region with $S^4=\left(\mathrel{+} \, \mathrel{+} \, \mathrel{-} \, \mathrel{-} \right)$. Notice that the regions with $S^4=\left(\mathrel{+} \, \mathrel{-} \, \mathrel{-} \, \mathrel{+} \right)$ and $S^4=\left(\mathrel{+} \, \mathrel{-} \, \mathrel{+} \, \mathrel{+} \right)$ are not discernible in \fref{fig:Knead4}(a). Nevertheless, these regions exist and are bounded by the curves $\mathbf{HOM^3_3}$ and $\mathbf{HOM^4_3}$ of Shilnikov bifurcations that allow the corresponding transition between the regions that differ in the fourth symbol of the kneading sequence; see panels~(b3) and (b4).

When the number of kneading symbols $n$ is increased, we find a repetition of the process of previous regions of constant kneading sequences being subdivided by additional curves of Shilnikov bifurcations with finite kneading sequences of length $n-1$. This is illustrated in \fref{fig:BifEtoP}(a), where the respective coloring in the $(f,\delta)$-plane near the Bykov T-point $\mathbf{T}$ for $S^7$; also shown are representative bounding curves of Shilnikov bifurcations as identified with Lin's method and then continued in parameters. \Fref{fig:BifEtoP}(a) gives an impression of how additional isolas of constant kneading sequence arise and how they are organized in the parameter plane. One can clearly observe accumulation processed where the space outside and between isolas is being filled with new isolas as the number $n$ of symbols increases, such as to $S^{12}$ in \Fref{fig:BifEtoP}(b) which is discussed in more detail below. Notice, in particular, how quite a few of such isolas reach a vertical maximum at the top of $(f,\delta)$-plane, around $\delta=-5$.

\section{Kneading sequences and curves of EtoP connections in the $(f,\delta)$-plane}\label{sec:EtoPKneading}

Understanding the different accumulation processes in the organization of the $(f,\delta)$-plane near the Bykov T-point is a considerable challenge that requires the study of additional global bifurcations of different kinds. A natural starting point is the study of codimension-one heteroclinic or EtoP connections from the symmetric saddle equilibrium $p$ to different saddle periodic orbits with a single unstable Floquet multiplier.

\subsection{Isolas bounded by EtoP connections from $p$ to $\Gamma_o$ and $\Gamma^*_o$}
\label{sec:IsolasGamma_o}
 
The first type of EtoP connection we consider is that to the pair of basic saddle periodic orbits $\Gamma_o$ and $\Gamma^*_o$ that bifurcate from the subcritical branch of the curve $\mathbf{H}$ of Hopf bifurcation. These symmetry-related periodic orbits are orientable, and each has a three-dimensional stable manifold and a two-dimensional unstable manifold. Hence, any connection from $p$ to $\Gamma_o$ (and $\Gamma^*_o$ by symmetry) is indeed of codimension one, while the connections back from $\Gamma_o$ and $\Gamma^*_o$ to $p$ are generically structurally stable (if they exists). \Fref{fig:BifEtoP}(b1) already shows a number of curves of such codimension-one EtoP connections, which are labelled  $\mathbf{Hep}$, $\mathbf{Hep_1^-}$ and $\mathbf{Hep^{+/-}_i}$ with $\mathbf{i} = \mathbf{1} ,..,\mathbf{4}$. Here, the ``$+$'' superscript means that the EtoP connection is from $p$ to $\Gamma_o$, while ``$-$'' superscript corresponds to the EtoP connection with $\Gamma^*_o$

In fact, we already encountered in \fref{fig:homOrbits}(b) the basic EtoP connection, where $W^u_+(p)$ stays in the region with $|A|^2 < |B|^2$ and connects directly to $\Gamma_o$; we continue to refer to this global bifurcation as $\mathbf{Hep}$ from now on. It arose as the limiting case of Shilnikov homoclinic orbits with an increasing number of loops of $W^u_+(p)$ also in the region with $|A|^2 < |B|^2$ and, hence, finite kneading sequences consisting of $n$ repetitions of the symbol $+$; see \fref{fig:homOrbits}(a). Accordingly, we find in \fref{fig:BifEtoP}(b1) that the isolas with kneading sequence $S^n=\left(\mathrel{+} \, \cdots \, \mathrel{+} \, \mathrel{-} \right)$  accumulate on the curve $\mathbf{Hep}$ of this basic EtoP connection. The curve $\mathbf{Hep}$ emerges from the Bykov T-point $\mathbf{T}$ and can be continued to the generalized Hopf bifurcation point $\mathbf{GH}$, outside the range shown in \fref{fig:BifEtoP}(b1), where the criticality of the Hopf bifurcation $\mathbf{H}$ changes compare with \fref{fig:bifSetKappa+2}.

Notice that during the transition through $\mathbf{Hep}$, or through any other EtoP connection from $p$ to $\Gamma_o$ and $\Gamma^*_o$, the one-dimensional unstable manifold $W^u(p)$ moves from the inside to the outside of the three-dimensional hypercylinders $W^s(\Gamma_o)$ and $W^s(\Gamma^*_o)$. Since the equilibrium $p$ is outside these hypercylinders, a Shilnikov bifurcation with a given finite kneading sequence can only occur after the corresponding EtoP connection has occurred that ensures that $W^u(p)$ is also outside $W^s(\Gamma_o)$ and $W^s(\Gamma^*_o)$. This explains why all  Shilnikov bifurcations lie to one side of the curve $\mathbf{Hep}$ in the $(f,\delta)$-plane. The other bounding curve in \fref{fig:BifEtoP}(b1) for the existence of Shilnikov bifurcations is the curve $\mathbf{Hep_1^-}$ for which the branch $W_+^u(p)$ has a single maximum in $||B|^2 - |A|^2|$ before moving to the region of phase space with $|B|^2 < |A|^2$ to accumulate on $W^s(\Gamma^*_o)$, see \fref{fig:BifEtoP}(b2). 

\begin{figure}[t!]
\centering
\includegraphics[scale=0.95]{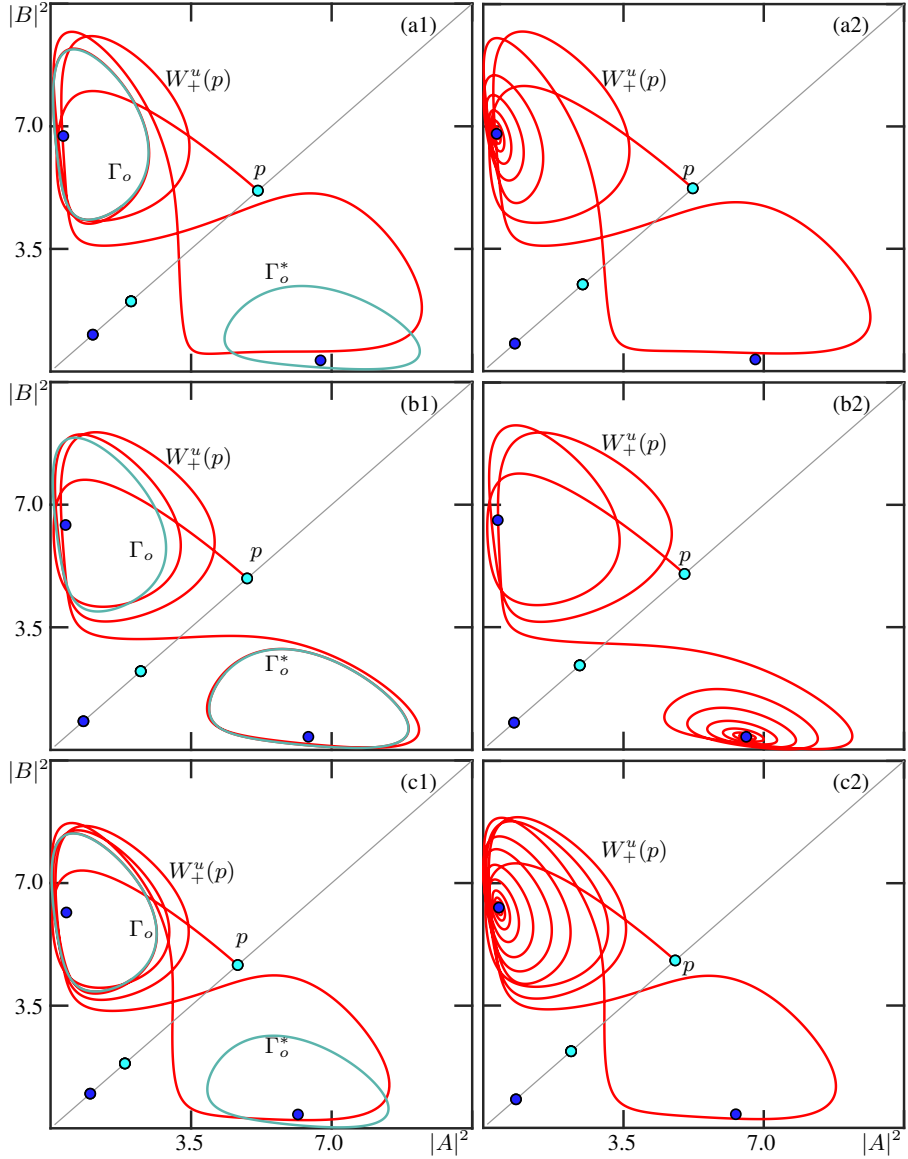} 
\caption{Heteroclinic orbits along the curves $\mathbf{Hep_3^+}$,  $\mathbf{Hep_3^-}$ and $\mathbf{Hep_4^+}$ (left column) and the behavior of $W^u_+(p)$ inside the respective enclosed regions of constant kneading sequence (right column); shown in the $(|A|^2,|B|^2)$-plane are $W^u_+(p)$ (red curve), $\Gamma_o$ and $\Gamma_o^*$ (cyan curves), and the different equilibria. Panels (a1), (b1) and (c1) are for $(f,\delta) \approx (3.3000, -6.1106)$ on $\mathbf{Hep_3^+}$, $(f,\delta) \approx (3.0000,-5.9816)$ on $\mathbf{Hep_3^-}$, and  $(f,\delta) \approx (2.9200, -5.7506)$ on $\mathbf{Hep_4^+}$, respectively; panels (a2), (b2) and (c2) are for $(f,\delta) = (3.2, -6.25)$, $(f,\delta) = (3.05, -6.1)$ and $(f,\delta) = (2.93, -5.9)$, respectively.}  \label{fig:EtoPSample}
\end{figure} 

The argument above also implies that other types of EtoP connections are accumulated by Shilnikov bifurcations as well.  In \fref{fig:BifEtoP}(b1) we show examples of curves $\mathbf{Hep^{+/-}_i}$ of EtoP connections.  Each of these curves encloses an isola of constant kneading sequence in the $(f,\delta)$-plane, inside which no Shilnikov bifurcations occur; moreover, these boundary curves of EtoP connections are accumulated from the outside by curves of Shilnikov bifurcations. These properties, which we also found for the basic EtoP connection curve $\mathbf{Hep}$, can be understood by considering the kneading sequence of $W_+^u(p)$ at the moment of respective EtoP connection and within the enclosed isola. For the points highlighted by yellow asterisks and yellow stars in \fref{fig:BifEtoP}(b1), the branch $W_+^u(p)$ is shown in \fref{fig:EtoPSample} for points on the bifurcation curves $\mathbf{Hep_3^+}$,  $\mathbf{Hep_3^-}$  and $\mathbf{Hep_4^+}$ in the left column, and for the corresponding enclosed region of constant kneading sequence in the right column. By comparing panels (a1), (b1) and (c1) with panels (a2), (b2) and (c2) of \fref{fig:EtoPSample}, one observes that each curve $\mathbf{Hep^{+/-}_i}$ bounds a region with the same constant symbol sequence, characterized by the fact that $W_+^u(p)$ now reaches the respective attracting equilibrium inside the topological hypercylinder formed by $W^s(\Gamma_o)$ or $W^s(\Gamma_o^*)$, respectively. The chosen curves $\mathbf{Hep^{+/-}_i}$ and the associated isolas are characterized by the kneading sequences:
\begin{itemize}
\item 
$S=\left(\mathrel{+} \, \cdots \, \mathrel{+} \, \mathrel{-} \, \overline{\mathrel{+}} \right)$ for $\mathbf{Hep^{+}_i}$ and the region bounded by it;
\item 
$S=\left(\mathrel{+} \, \cdots \, \mathrel{+}  \, \mathrel{+} \, \overline{\mathrel{-}} \right)$ for $\mathbf{Hep^{-}_i}$ and the region bounded by it.
\end{itemize}
For both families, the subindex $\mathbf{i}$ denotes the number of symbols before the sequence reaches its repeating symbol. Notice in \fref{fig:BifEtoP}(b1) that for these two specific families of EtoP connections the respective curves and regions accumulate on the curve $\mathbf{Hep}$, which is compatible with the fact that the kneading sequence at the EtoP connection $\mathbf{Hep}$, and in the region to its left, is $S=\left(\mathrel{+} \overline{\mathrel{+}} \right)$. 

Isolas and associated (families of) EtoP connections from $p$ to $\Gamma_o$ and $\Gamma_o^*$ with initial kneading sequences other than $\mathrel{+} \, \cdots \, \mathrel{+}$ exist and can be found and continued in the same way. We remark that we have found this intricate interplay between homoclinic bifurcations and EtoP connections previously in the bifurcation diagram of an inclination flip bifurcation of case~\textbf{C}, where the region of homoclinic bifurcations (to a real saddle and not a focus in this case) is similarly organised by the stable manifold of codimension one of an orientable saddle periodic orbit \cite{And2}.

\subsection{Isolas bounded by EtoP connections from $p$ to $\Gamma^s_o$}
\label{sec:IsolasGamma_symm}

\begin{figure}
\centering
\includegraphics[scale=0.95]{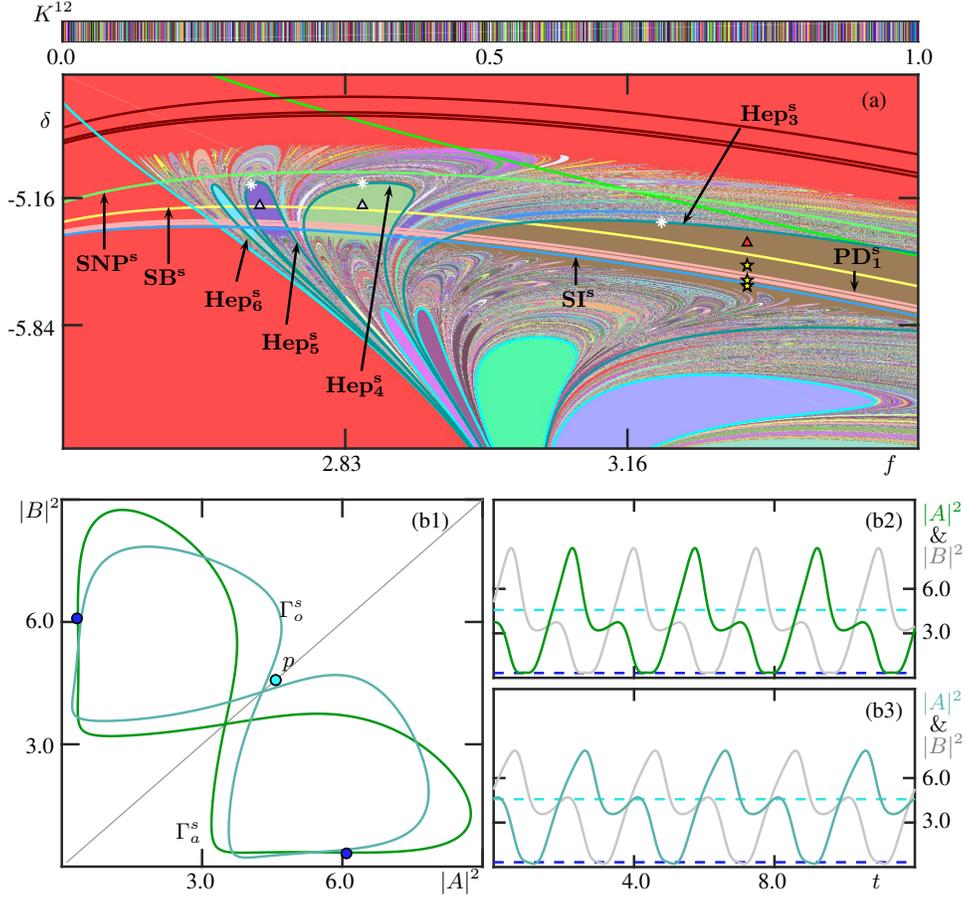} 
\caption{Bifurcations associated with S-invariant periodic orbits. Panel~(a) shows the coloring of the $(f,\delta)$-plane by finite kneading sequences $S^{12}$ with associated curves of saddle node of periodic orbit bifurcation $\mathbf{SNP^s}$ (light-green) that creates the orientable S-invariant periodic orbits $\Gamma^s_a$ and $\Gamma^s_o$, of symmetry breaking bifurcation $\mathbf{SB^s}$ (yellow) of $\Gamma^s_a$, of successive period-doubling bifurcations $\mathbf{PD_i^s}$ (brown), of symmetry increasing bifurcation $\mathbf{SI^s}$, and of EtoP connections $\mathbf{Hep^{s}_i}$ with $\mathbf{i} = \mathbf{3} ,\ldots,\mathbf{6}$ (dark cyan) from $p$ to $\Gamma^s_o$. The periodic orbits $\Gamma^s_a$ (green) and $\Gamma^s_o$ (cyan) at the red triangle with $(f,\delta) = (3.3, -5.4)$ are shown in the $(|A|^2,|B|^2)$-plane in panel~(b1) and as time series of $|A|^2$ and $|B|^2$ in panels~(b2) and~(b3), respectively. In panel~(a) parameter values are indicated for which the curves $W^u_+(p)$, $\Gamma^s_a$ and $\Gamma^s_o$ are shown in separate figures, namely: for the white triangles and the white asterisks along the curves $\mathbf{Hep_3^s}$, $\mathbf{Hep_4^s}$ and $\mathbf{Hep_5^s}$ see \fref{fig:EtoPSymSample}; and for the yellow stars see \fref{fig:TrueSymChaos}.} 
\label{fig:BifEtoPSym}
\end{figure} 

\begin{figure}
\centering
\includegraphics[scale=0.95]{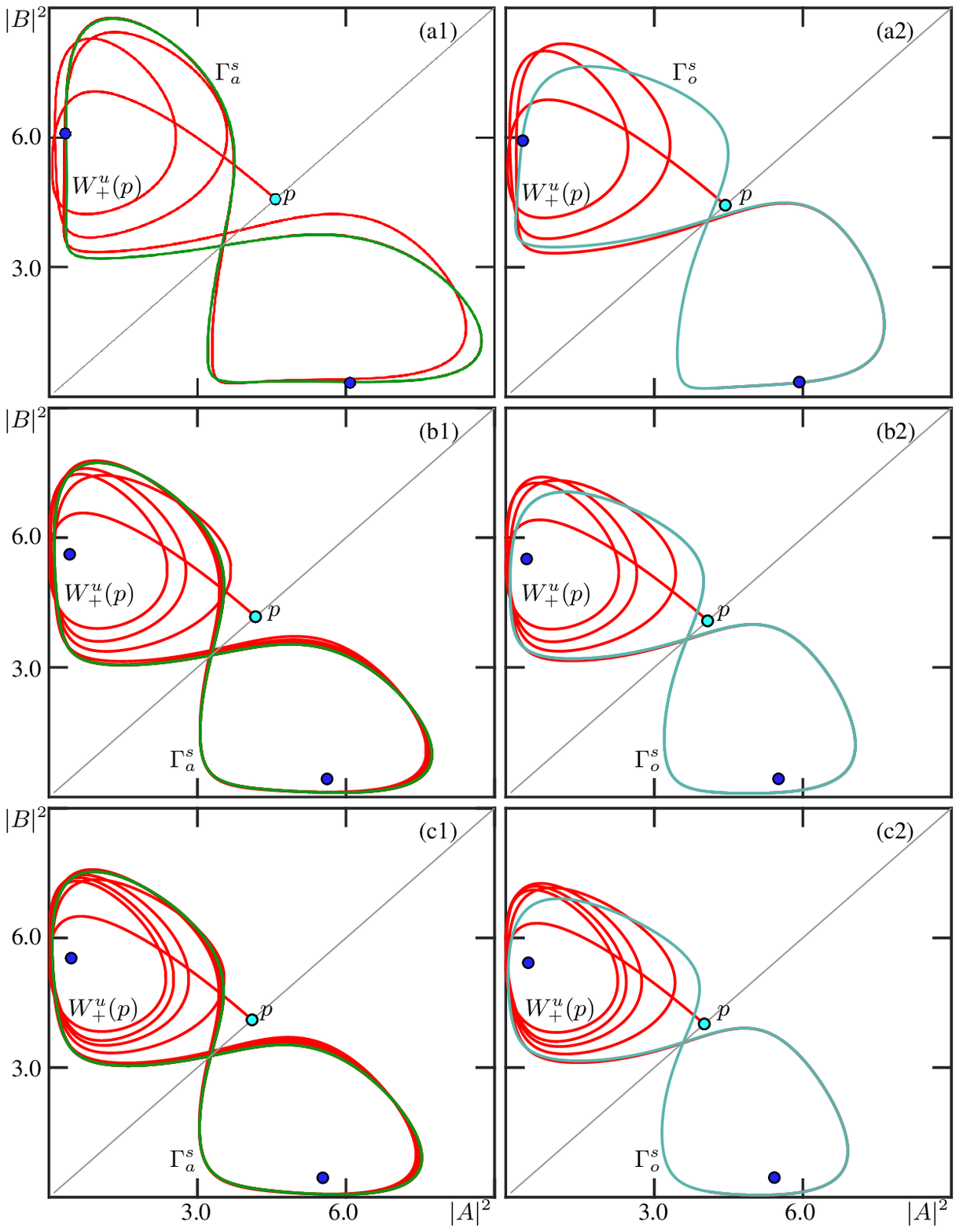} 
\caption{The behavior of $W^u_+(p)$ at the white triangles in \fref{fig:BifEtoPSym}(a) inside regions of constant kneading sequence (left column), and the associated heteroclinic orbits along the bounding curves $\mathbf{Hep_3^s}$ to $\mathbf{Hep_5^s}$ (right column); shown in the $(|A|^2,|B|^2)$-plane are $W^u_+(p)$ (red curve), the 
periodic orbits $\Gamma^s_a$ (green curve) 
and $\Gamma^s_o$ (cyan curve), and the different equilibria. Panels (a1), (b1) and (c1) are for $(f,\delta) = (3.30, -5.4)$, $(f,\delta) = (2.85, -5.2)$ and $(f,\delta)= (2.73, -5.2)$, respectively; panels (a2), (b2) and (c2) are for $(f,\delta) \approx (3.2000, -5.2902)$ on $\mathbf{Hep_3^s}$, $(f,\delta) \approx (2.8500, -5.0840)$ on $\mathbf{Hep_4^s}$, and $(f,\delta) \approx (2.7200, -5.0877)$ on $\mathbf{Hep_5^s}$, respectively.} 
\label{fig:EtoPSymSample} 
\end{figure} 

There are other larger and similar regions of constant kneading sequences in  the $(f,\delta)$-plane of \fref{fig:BifEtoP}(b1) for smaller values of the detuning $\delta$ around $\delta \approx -5.0$. \Fref{fig:BifEtoPSym} focuses on the associated bifurcations and objects. The bifurcation diagram in panel~(a) is an enlargement, where three of these regions feature points marked by triangles. At the triangles, one finds a pair of S-invariant periodic orbits, which are created (for decreasing $\delta$) at the saddle-node of periodic orbit bifurcation curve $\mathbf{SNP^s}$. The two S-invariant periodic orbits at the red triangle, the attractor $\Gamma^s_a$ and the orientable saddle periodic orbit $\Gamma^s_o$, are shown in \fref{fig:BifEtoPSym}(b). The image in the $(|A|^2,|B|^2)$-plane in panel~(b1) is accompanied by time series of $\Gamma^s_a$ and $\Gamma^s_o$ that show $|A|^2$ and $|B|^2$ in panels~(b2) and (b2), respectively, which clearly show the phase shift over half a period between the two cavities that is characteristic for S-invariant periodic orbits. 

\begin{figure}[th!]
\centering
\includegraphics[scale=0.95]{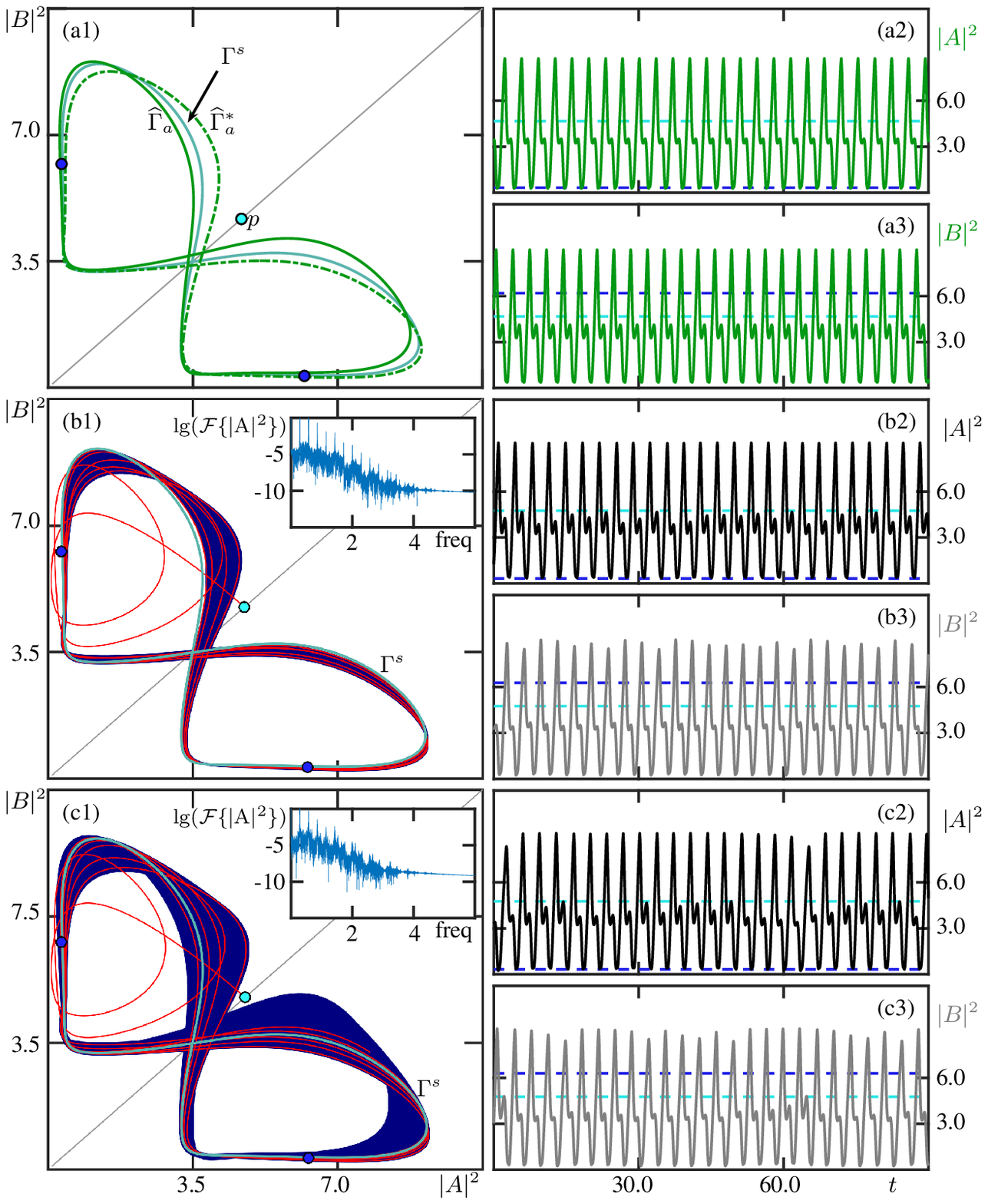} 
\caption{Attracting symmetry-broken periodic orbits $\widehat{\Gamma}_a$ (green curve) and chaotic attractors (blue curves) of system~\eref{eq:Couplednondim} for $\kappa = 2$, $f=3.3$, and $\delta = -5.52$ in row (a), $\delta = -5.61$ in row (b) and $\delta = -5.63$ in row (c); shown in the $(|A|^2, |B|^2)$-plane with the branch $W_+^u(p)$ (red curve) and as time series of $|A|^2$ and $|B|^2$, respectively; the insets show the power spectra of the chaotic attractors.}
\label{fig:TrueSymChaos}
\end{figure} 

\Fref{fig:EtoPSymSample} shows the branch $W^u_+(p)$ at parameter points inside and at the boundary of these regions of constant kneading sequences with the triangles. As is shown in panels~(a2), (b2) and (c2) for the parameter points at the triangles inside each region, $W^u_+(p)$ accumulates on the S-invariant attracting periodic orbit $\Gamma^s_a$, after initially making three, four and five loops in the region with $|A|^2 < |B|^2$, respectively. Panels~(a2), (b2) and (c2) of \fref{fig:EtoPSymSample} show that each of these regions is associated with an EtoP connection that connects $p$ with the orientable S-invariant saddle periodic orbit $\Gamma^s_o$, whose stable manifold $W^s(\Gamma^s_o)$ bounds the basin of attraction of  $\Gamma^s_a$. We refer to this family of EtoP connections as $\mathbf{Hep^{s}_i}$, owing to the fact that they and the regions they enclose are characterized by the kneading sequences
\begin{itemize}
\item 
$S=\left(\mathrel{+} \, \cdots \, \mathrel{+} \, \overline{\mathrel{-}\,\mathrel{+}} \right)$ where the initial finite sequence is of length $\mathbf{i}$.
\end{itemize}
The associated curves $\mathbf{Hep^{s}_i}$ for $\mathbf{i} = \mathbf{3} ,\ldots,\mathbf{6}$ are shown in \fref{fig:BifEtoPSym}(a), where the parameter points of the EtoP connections shown in \fref{fig:EtoPSymSample}(a2)--(c2) are marked by white asterisks. The curves $\mathbf{Hep^{s}_i}$ of EtoP connections to $\Gamma^s_o$ start and end at the Bykov T-point $\mathbf{T}$ (outside the range shown) and, thus, delimit isolas in the $(f,\delta)$-plane. Inside each of these isolas the unstable manifold $W^u(p)$ accumulates on the attractor with basin boundary $W^s(\Gamma^s_o)$.

For sufficiently large $\delta$, the attractor inside the region bounded by $\mathbf{Hep^s_i}$ is the attracting S-invariant periodic orbit $\Gamma^s_a$, which is the periodic orbit $\Gamma^s_a$ in \fref{fig:EtoPSymSample}(a1), (b1) and (c1). \Fref{fig:BifEtoPSym}(a) also shows associated bifurcation curves, namely those of symmetry breaking bifurcation $\mathbf{SB}^s$ of $\Gamma^s_a$, of period-doubling bifurcations $\mathbf{PD}_i^s$ of nonsymmetric attracting periodic orbits, and of a subsequent symmetry increasing bifurcation $\mathbf{SI^s}$; the three stars are the parameter points for which the respective attractors are shown in \fref{fig:TrueSymChaos}. The S-invariant periodic orbit $\Gamma^s_a$ from \fref{fig:BifEtoPSym}(b) becomes a saddle periodic orbit, denoted $\Gamma^s$, at the symmetry-breaking bifurcation $\mathbf{SB}^s$. As is shown in \fref{fig:TrueSymChaos}(a1) for the top yellow star in \fref{fig:BifEtoPSym}(a), it creates for lower values of $\delta$ the pair of nonsymmetric periodic orbits $\widehat{\Gamma}_a$ and $\widehat{\Gamma}^*_a$, which are each others counterparts under the reflection $\eta$. This pair coexists with the S-invariant saddle periodic orbit $\Gamma^s$ (the continuation of $\Gamma_a^s$ past $\mathbf{SB}^s$). Panels~(a2) and~(a3) for $\widehat{\Gamma}_a$ illustrate the broken S-invariance of $\widehat{\Gamma}_a$: the time series of $|A|^2$ and $|B|^2$ are indeed no longer phase shifts over half a period of one another. The boundary between the basins of attraction of the symmetry-broken attractors $\widehat{\Gamma}_a$ and $\widehat{\Gamma}^*_a$ is the three-dimensional stable manifold $W^s(\Gamma^s)$ of $\Gamma^s$. Note from \fref{fig:EtoPSymSample}(a1) that the branch of $W_+^u(p)$ accumulates on $\widehat{\Gamma}_a$. 

When the (cascade of) subsequent period-doubling bifurcations in \fref{fig:BifEtoPSym}(a) is crossed, $\widehat{\Gamma}_a$ and $\widehat{\Gamma}^*_a$ successively period-double and a pair of chaotic attractors is created. \Fref{fig:TrueSymChaos}(b1) shows the chaotic attractor, associated with $\widehat{\Gamma}_a$, on which $W_+^u(p)$ accumulates. Notice that this pair of chaotic attractors is very different from the pair shown in \fref{fig:chaos}(b), which are each restricted to one of the two cavities and hence never show switching between cavities. Rather, as panels~(b2) and~(b3) show, this attractor is characterized by the consistent alternation between the two cavities, that is, the regions with  $|A|^2 < |B|^2$ and with $|B|^2 < |A|^2$. In other words, the kneading sequence generated by $W_+^u(p)$ remains unchanged, which is represented in \fref{fig:BifEtoPSym}(a) by the fact that the corresponding parameter point lies in the same region of constant kneading sequence. We refer to this chaotic behavior as \emph{chaotic behavior with regular switching}.  The boundary between the two symmetry-related attractors is still $W^s(\Gamma^s)$, and they merge into a single symmetric chaotic attractor soon after they are created at a codimension-one homoclinic tangency where two-dimensional unstable manifold $W^u(\Gamma^s)$ becomes tangent to the three-dimensional stable manifold $W^s(\Gamma^s)$ at the bifurcation curve $\mathbf{SI^s}$. This is a symmetry-increasing bifurcation of chaotic attractors, where the two chaotic attractors collide with their common basin boundary $W^s(\Gamma^s)$ to become a single chaotic attractor. \Fref{fig:TrueSymChaos}(c1) illustrates that this symmetric chaotic attractor does still not contain the saddle point $p$. As panels~(c2) and~(c3) show, any trajectory on this symmetric attractor, which exists in the same region of constant kneading sequence, still alternates between the two cavities. In other words, the kneading sequence generated by $W_+^u(p)$ remains unchanged throughout the entire transition from $\Gamma^s_a$ via $\widehat{\Gamma}_a$ and the two separate chaotic attractors to this single chaotic attractor. There are also no discernible changes in the time evolutions of the intensities compared to the situation in panels~(b).

In contrast to the isolas bounded by the curves $\mathbf{Hep^{+/-}_i}$ of EtoP connections from $p$ to the symmetric pair of saddle periodic orbits $\Gamma_o$ and $\Gamma_o^*$, inside the isolas bounded by the curves $\mathbf{Hep^{s}_i}$ in \fref{fig:BifEtoPSym}(a) the kneading sequence is not constant throughout. Rather, there are infinitely many further isolas with different kneading sequences inside each of these isolas for values of $\delta$ below what appears to be a well-defined curve. In particular, notice that these include the isolas bounded by $\mathbf{Hep^{+/-}_i}$; for example, the isola of $\mathbf{Hep^s_3}$ clearly contains the two sub-isolas bounded by  $\mathbf{Hep^-_3}$ and $\mathbf{Hep^+_4}$. As we will discuss in the next section, such sudden changes of the kneading sequence are due to a global bifurcation involving the symmetric saddle equilibrium $p$.

\section{Tangency bifurcation of $W^s(p)$ and degenerate singular cycles to a saddle focus}\label{sec:TangCycSadd}

Up to now, we considered different codimension-one Shilnikov bifurcations of the symmetric equilibrium $p$, as well as EtoP connections between $p$ and saddle periodic orbits with different symmetry properties. In combination with determining the kneading sequences generated by $W^u(p)$, we found an emerging picture of infinitely many curves of such global bifurcations in the $(f,\delta)$-plane that accumulate on one another in a complicated way. We now address two important questions regarding the observed phenomena:\\[-3mm]
\begin{enumerate}
\item Why and when do curves of Shilnikov bifurcations accumulate on curves of EtoP connections?\\[-4mm]
\item What determines different observed upper boundaries of isolas in the $(f,\delta)$-plane?\\[-3mm]
\end{enumerate}
The answers to these two questions are intimately related and concern the existence of heteroclinic cycles involving the point $p$. For such heteroclinic cycles to exist, apart from the codimension-one connection defining the respective curve, there must also exit structurally stable connections between $p$ and the respective saddle objects \cite{ HomSan, KraRie1, RADEMACHER2005390,RADEMACHER2010305}. These additional connections are structurally stable heteroclinic orbits, which are created at codimension-one (generically quadratic) tangencies or folds between two global invariant manifolds.

As we saw in \cref{sec:KneadSeq}, the three-dimensional stable manifolds of orientable saddle periodic orbits, such as $\Gamma_o$ and $\Gamma^s_o$, are separatrices, and their relative positions with $W^u(p)$ determines whether $W^u(p)$ and $W^s(p)$ can intersect or not, that is, have a Shilnikov bifurcation or not. The fact that a Shilnikov bifurcation is possible, however, is not sufficient to explain the accumulation process of Shilnikov bifurcations onto the respective EtoP connection. Such an accumulation requires additionally that, at the codimension-one heteroclinic EtoP connection, the separatrix $W^s(p)$ comes arbitrary close to the stable manifold(s) of the respective saddle periodic orbit(s). Indeed, in this case, an intersection between $W^u(p)$ and $W^s(p)$ and, hence, a Shilnikov bifurcation can be created by an arbitrarily small perturbation in parameters. 

\subsection{First heteroclinic tangency between $W^s(p)$ and $W^u(\Gamma_o)$ and $W^u(\Gamma_o^*)$}
\label{sec:KneadSeqGamma_o}

For the observed accumulation of Shilnikov bifurcations onto the curve $\mathbf{Hep}$ of the basic EtoP connection between $p$ and the pair of cycles $\Gamma_o$ and $\Gamma_o^*$ there must exist structurally stable heteroclinic orbits that converge backwards in time to $\Gamma_o$ and forwards in time to $p$ (and similarly for $\Gamma_o^*$). This means that the intersection set of the two-dimensional unstable manifold $W^u(\Gamma_o)$ and the three-dimensional stable manifold $W^s(p)$ is non-empty and transversal, so that there is a pair of heteroclinic cycles. The transversality of this intersection implies, as a consequence of the $\lambda$-Lemma \cite{Palis1,Wigg1}, that $W^s(p)$ is indeed arbitrary close to $W^s(\Gamma_o)$ and $W^s(\Gamma_o^*)$ near $p$; see also \cite{And1, And2} where this general phenomenon is demonstrated close to homoclinic flip bifurcations.  

\begin{figure}
\centering
\includegraphics[scale=0.95]{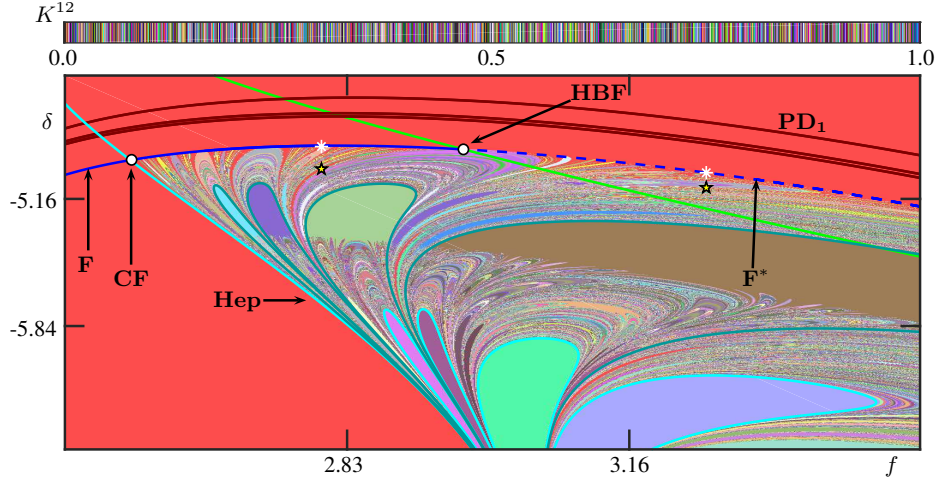} 
\caption{Bifurcations associated with the symmetric attractors arising from $\Gamma_a$ and $\Gamma^*_a$. Shown is the coloring of the $(f,\delta)$-plane by finite kneading sequences $S^{12}$ with curves of period-doubling bifurcations $\mathbf{PD_i}$ (brown), of Hopf bifurcation $\mathbf{H}$ where the saddles $\Gamma_o$ and $\Gamma_o^*$ bifurcate with the two stable equilibria to create the pair of saddle equilibria $q$ and $q^*$, and of folds $\mathbf{F}$ (blue) and $\mathbf{F^*}$ where $W^s(p)$ is tangent to $W^u(\Gamma_o)$ and $W^u(\Gamma_o^*)$, and to $W^u(q)$ and $W^u(q^*)$, respectively. Also labelled are the codimension-two points $\mathbf{CF}$ where $\mathbf{F}$ intersects the curve $\mathbf{Hep}$, and $\mathbf{HBF}$ where the tangency of $W^s(p)$ coincides with the Hopf bifurcation. The stars and asterisks indicate the parameter points chosen for
the panels in \fref{fig:BifTransSub2}.} 
\label{fig:BifTrans}
\end{figure} 

\begin{figure}
\centering
\includegraphics[scale=0.95]{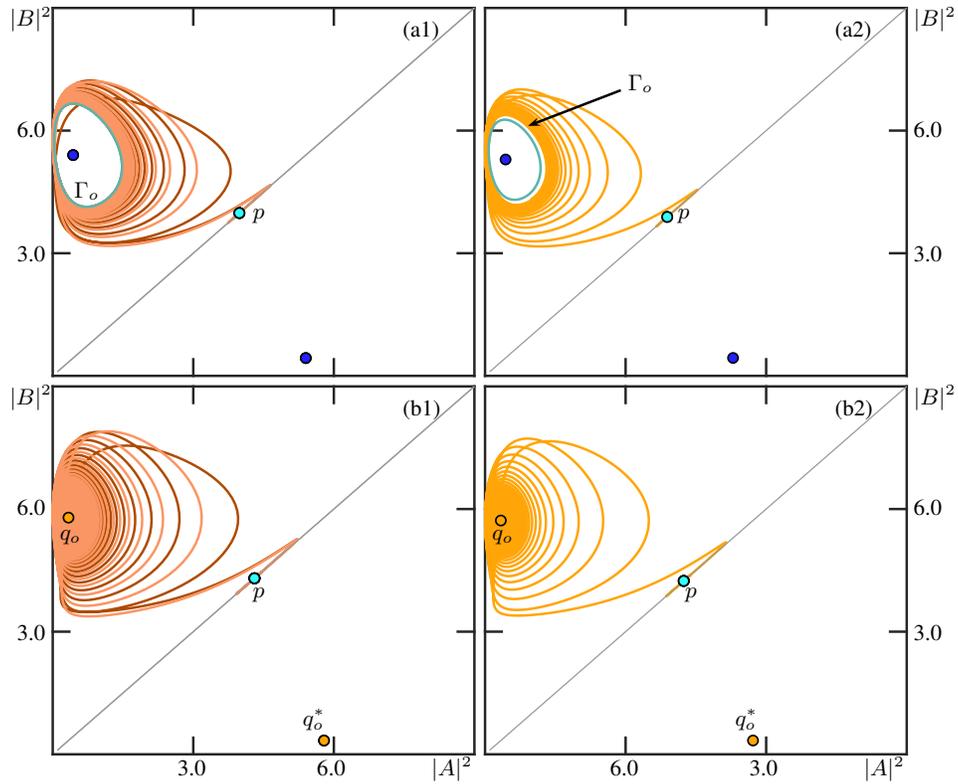} 
\caption{Heteroclinic orbits in the $(|A|^2, |B|^2)$-plane between the symmetric saddle equilibrium $p$ and the periodic orbit $\Gamma^s_o$ (a) and the asymmetric saddle equilibrium $q_o$ (b). Panels~(a1) and~(b1) show pairs of structurally stable heteroclinic orbits at the stars in \fref{fig:BifTrans}, and panels~(a2) and~(b2) show the respective tangent heteroclinic orbits at the asterisks on the curves $\mathbf{F}$ and $\mathbf{F^*}$; here $(f,\delta) = (2.80, -5.0)$ in (a1), $(f,\delta) \approx (2.800, -4.8752)$ in (a2), $(f,\delta) = (3.25, -5.10)$ in (b1), and $(f,\delta) \approx (3.2500, -5.0192)$ in (b2).} 
\label{fig:BifTransSub2}
\end{figure} 

We indeed find a pair of structurally stable heteroclinic orbits between $p$ and $\Gamma_o$ with Lin's method in the region to the right of the curve $\mathbf{Hep}$ of the basic EtoP connection; their continuation in $\delta$ then allows us to detect their codimension-one fold bifurcation, which we subsequently continue to obtain the curve $\mathbf{F}$ in the $(f,\delta)$-plane that is shown in \fref{fig:BifTrans}. For these particular heteroclinic orbits, $\mathbf{F}$ represents the situation where the three-dimensional stable manifold of $p$ has a quadratic tangency with the two-dimensional unstable manifold of $\Gamma_o$. Note that the curve $\mathbf{F}$ crosses $\mathbf{Hep}$ at a point $\mathbf{CF}$ and ends on the curve $\mathbf{H}$ at the point $\mathbf{HBF}$. Along $\mathbf{H}$ the pair of saddle periodic orbits $\Gamma_o$ and $\Gamma_o^*$ undergoes a subcritical Hopf bifurcation with the nonsymmetric attracting equilibria to create the pair of saddle focus equilibria $q_o$ and $q_o^*$. The two-dimensional unstable manifold $W^u(q_o)$ is tangent to $W^s(p)$ along the curve $\mathbf{F^*}$ (and likewise for $q_o^*$), which therefore should be seen as the continuation of the curve $\mathbf{F}$ past the Hopf bifurcation. The respective heteroclinic connections in the $(|A|^2, |B|^2)$-plane at the two stars and the two asterisks are shown in \fref{fig:BifTransSub2}. Panel~(a1) shows the pair of heteroclinic orbits that are the transverse intersections of $W^s(p)$ and $W^u(\Gamma_o)$, which emerge from the nearby fold curve $\mathbf{F}$; the associated single nontransverse heteroclinic orbit in the tangential intersections of $W^s(p)$ and $W^u(\Gamma_o)$ is shown in panel~(a2). \Fref{fig:BifTransSub2}(b1) shows the heteroclinic orbit that form the transverse intersections of $W^s(p)$ and $W^u(q_o)$ near the fold curve $\mathbf{F^*}$, where one finds the associated single nontransverse heteroclinic orbit in panel~(b2). Note that \fref{fig:BifTransSub2} illustrates the fact that heteroclinic orbits between $p$ and $q_o$ in panels~(b1) and~(b2) are topological continuations through the Hopf bifurcation curve $\mathbf{H}$ of those in panels~(a1) and~(a2).

\Fref{fig:BifTrans} shows that the heteroclinic fold curves $\mathbf{F}$ and $\mathbf{F^*}$ bound from above the region of the $(f,\delta)$-plane where homoclinic bifurcations of the symmetric saddle equilibrium $p$ exists. Moreover, the curves $\mathbf{F}$ and $\mathbf{F^*}$ are also responsible for the symmetry increasing/decreasing bifurcation of the two symmetric chaotic attractors in \fref{fig:chaos}. This is so because the unstable manifolds $W^u(\Gamma_o)$ and $W^u(\Gamma^*_o)$, which initially accumulate on the periodic attractors $\Gamma_a$ and $\Gamma^*_a$ for sufficiently large $\delta$ in \fref{fig:BifTrans}, also accumulate on the two separate chaotic attractors that are created for decreasing $\delta$ by the transition through the period-doubling cascade, indicated by the curves $\mathbf{PD_i}$. That is, the closure of $W^u(\Gamma_o)$  and $W^u(\Gamma^*_o)$ contain the separated chaotic attractors.

\begin{figure}
\centering
\includegraphics[scale=0.95]{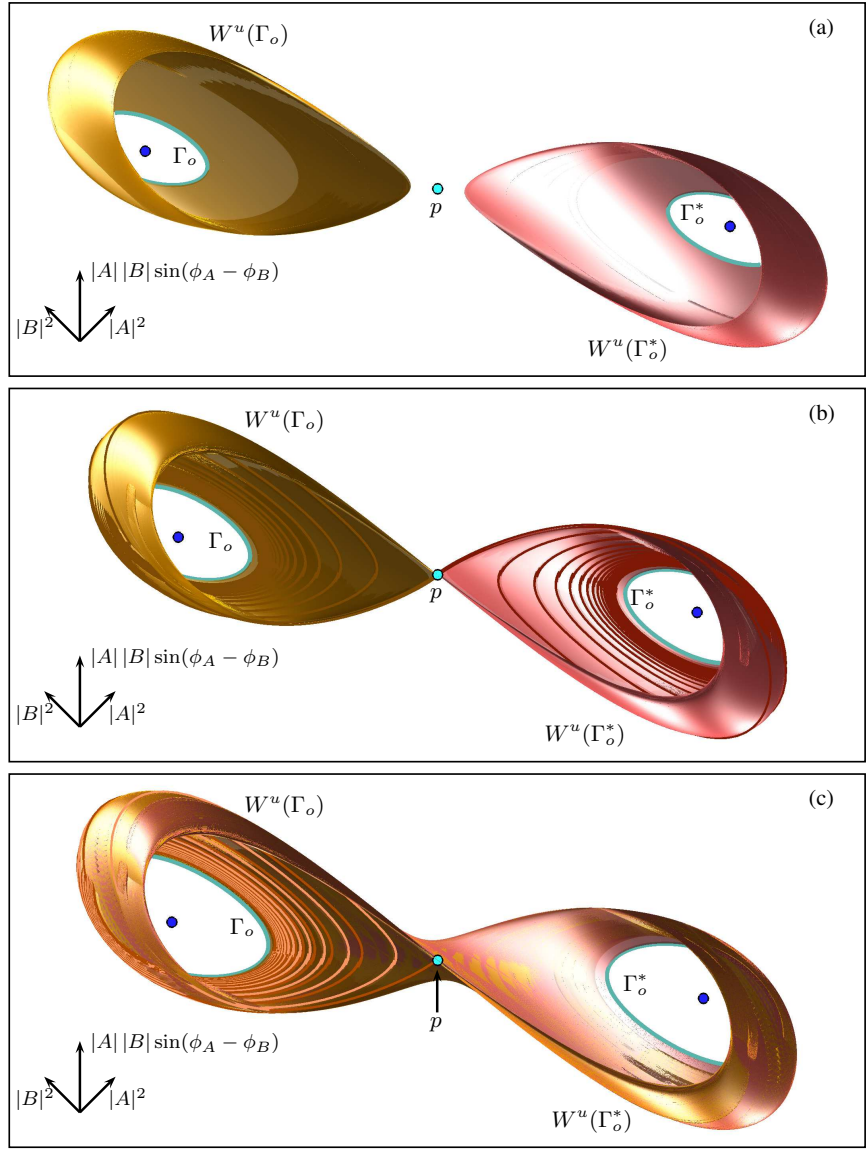} 
\caption{Symmetry increasing bifurcation of chaotic attractors at the transition of system \eref{eq:Couplednondim} through the heteroclinic fold curve $\mathbf{F}$, before $\mathbf{F}$ at $(f,\delta)=(2.8,-4.8)$ in panel~(a), at $\mathbf{F}$ at $(f,\delta)\approx (2.8000,-4.8752)$  in panel~(b), and after $\mathbf{F}$ at $(f,\delta)=(2.8,-5.0)$ in panel~(c). Shown in $(|A|^2,|B|^2,|A|\,|B|\sin(\phi_A-\phi_B))$-space are the saddle equilibrium $p$, the asymmetric attracting equilibria (blue dots), the saddle periodic orbits $\Gamma_o$ and $\Gamma^*_o$ (cyan curves) and the parts of their unstable manifolds $W^{u}(\Gamma_o)$ (gold surface) and $W^{u}(\Gamma^*_o)$ (red surface) that accumulate on the respective chaotic attractors; panel~(b) also shows the tangent heteroclinic orbits in $W^s(p) \cap W^{u}(\Gamma_o)$ (brown  curves) and $W^s(p) \cap W^{u}(\Gamma_o^*)$ (red curves), and panel~(c) also shows the pair of transversal heteroclinic orbits in $W^s(p) \cap W^{u}(\Gamma_o)$ (pink and maroon curves).} 
\label{fig:ManifoldsTrans}
\end{figure} 

\Fref{fig:ManifoldsTrans} shows $W^u(\Gamma_o)$ and $W^u(\Gamma^*_o)$ in projection onto $(|A|^2,|B|^2,|A|\,|B|\sin(\phi_A-\phi_B))$-space during the transition through the heteroclinic fold curve $\mathbf{F}$. For clarity of illustration, we do not show the parts of $W^u(\Gamma_o)$ and $W^u(\Gamma^*_o)$ that converge to the asymmetric stable equilibria but rather only those that accumulate on the respective chaotic attractors. The shown surfaces $W^u(\Gamma_o)$ and $W^u(\Gamma^*_o)$ were computed via the continuation of a one-parameter family of orbit segments defined by a two-point boundary value problem \cite{Doe3,Kra2}.  Panel~(a) shows the situation before the heteroclinic fold bifurcation, that is, above the curve $\mathbf{F}$ in $(f,\delta)$-plane; here, the manifolds $W^u(\Gamma_o)$ and $W^u(\Gamma^*_o)$ do not come near the equilibrium $p$ as they accumulate on the two separate chaotic attractors well within the regions with $|A|^2<|B|^2$ and $|B|^2<|A|^2$, respectively. At the moment of the heteroclinic fold bifurcation $\mathbf{F}$ in \fref{fig:ManifoldsTrans}(b), the two manifolds meet at $p$ due to the existence of a single tangent heteroclinic orbit that converges backward in time to $\Gamma_o$ and forward in time to $p$. This illustration clearly showcases the moment where the symmetry-increasing bifurcation occurs. After the heteroclinic tangency, as in panel~(c), the two manifolds $W^u(\Gamma_o)$ and $W^u(\Gamma^*_o)$ are no longer separated and are each able to access both regions with $|A|^2<|B|^2$ and that with $|B|^2<|A|^2$, as they now accumulate on the symmetric chaotic attractor. Notice further that the tangent heteroclinic orbit gives rise to a pair of transverse heteroclinic orbits; this is illustrated in \fref{fig:ManifoldsTrans}(b) and (c) for the intersection set $W^s(p) \cap W^{u}(\Gamma_o)$.

\begin{figure}
\centering
\includegraphics[scale=0.95]{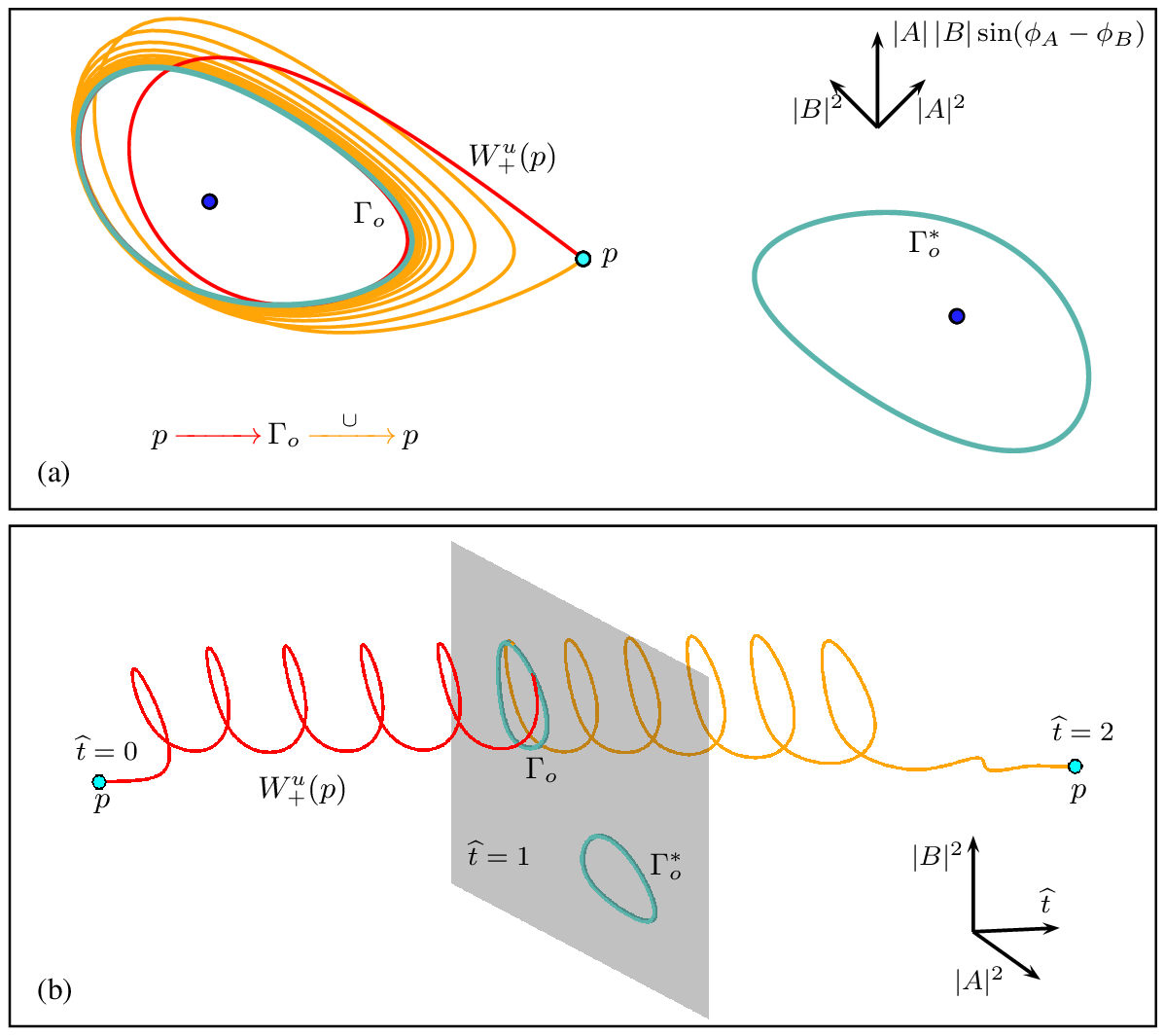} 
\caption{The singular EtoP cycle between $p$ and $\Gamma_o$ of system~\eref{eq:Couplednondim} at the codimension-two bifurcation point $\mathbf{CF}$ with $(f,\delta) \approx (2.577, -4.952)$ where the curves $\mathbf{F}$ and $\mathbf{Hep}$ intersect, shown in $(|A|^2,|B|^2,|A|\,|B|\sin(\phi_A-\phi_B))$-space in panel~(a) and as a compact time connection (CTC) plot in panel~(b). It consists of the codimension-one EtoP connection formed by $W_u^+(p)$ (red curve) and the tangent return connection (orange curve) that converges forward in time to $p$ and backward in time to $\Gamma_o$.} 
\label{fig:CodTwoHF}
\end{figure} 

As was already mentioned, the curves $\mathbf{F}$ and $\mathbf{Hep}$ intersect in the $(f,\delta)$-plane of \fref{fig:BifTrans} at a codimension-two point denoted $\mathbf{CF}$. At this point, system~\eref{eq:Couplednondim} exhibits a symmetric pair of singular heteroclinic EtoP cycles between $p$ and the pair of orbits $\Gamma_o$ and $\Gamma^*_o$. The singular cycle involving $\Gamma_o$ is shown in \fref{fig:CodTwoHF} and consists of a codimension-one EtoP connection formed by the one-dimensional unstable manifold $W_u^+(p)$ lying in the three-dimensional stable manifold $W^s(\Gamma_o)$, as well as a single return EtoP connection in the tangential intersection between the three-dimensional stable manifold $W^s(p)$ and the two-dimensional unstable manifold $W^u(\Gamma_o)$ (and similarly for $\Gamma^*_o)$ due to symmetry; not shown). In the literature, such connections are also called degenerate singular cycles. We find it advantageous to represent this pair of EtoP cycles (and subsequent other heteroclinic cycles) symbolically, by
\begin{eqnarray*} 
p \xrightarrow{ {\color[HTML]{000000} \phantom{va} \phantom{\pitchfork} \phantom{va} }} \Gamma_o \xrightarrow{ {\color[HTML]{000000} \phantom{va} \cup \phantom{va} }} p,\\
p \xrightarrow{ {\color[HTML]{000000} \phantom{va} \phantom{\pitchfork} \phantom{va} }} \Gamma^*_o \xrightarrow{ {\color[HTML]{000000} \phantom{va} \cup \phantom{va} }} p, \end{eqnarray*} 
where the arrows represent the respective heteroclinic connections between corresponding saddle invariant objects with the direction of time. Moreover, the symbol $\cup$ above the second arrow indicates that this connection is due to a quadratic tangency. \Fref{fig:CodTwoHF}(a) shows this singular EtoP cycle in projection onto $(|A|^2,|B|^2,|A|\,|B|\sin(\phi_A-\phi_B))$-space. To illustrate the nature of the connections involved in line with their symbolic representation, panel~(b) shows the EtoP cycle between $p$ and $\Gamma_o$ in what we call a \emph{compact time connection plot}, or CTC plot for short. The idea, which is motivated by the way we compute connecting orbits as solutions to boundary value problems, is to show (a suitable part of) each connection over a unit time interval in $(|A|^2,|B|^2,\widehat{t})$-space, where $\widehat{t}$ represents the compact, truncated time as used in the computation. Here each invariant object lies in a plane with $\widehat{t}=n$ with $n \in \Z$, specifically, $n = 0,1,2$ in \fref{fig:CodTwoHF}(b). The advantage of a CTC plot, which we will use extensively in what follows, is that it illustrates more readily the nature of the individual connections compared to a projection such as that in panel~(a).

The role of the codimension-two point $\mathbf{CF}$ is that it forms the corner of the region, bounded by the respective parts of the curves $\mathbf{F}$ and $\mathbf{Hep}$ in the $(f,\delta)$-plane of \fref{fig:CodTwoHF}(a), where Shilnikov bifurcations of $p$ occur. We remark that this type of codimension-two global bifurcation has been studied in \cite{Kirk1, Lohr1} for the case of a three-dimensional vector field with a saddle equilibrium with real eigenvalues, and a conjectured unfolding at the level of a return map was proposed in \cite{Kirk1} for the saddle-focus case. Indeed it was shown that sequences of curves of Shilnikov bifurcations accumulate on the curve that we denote $\mathbf{Hep}$ and that their maxima reach the fold curve $\mathbf{F}$. Furthermore,  codimension-two singular heteroclinic cycles with a real saddle have also been encountered close to homoclinic flip bifurcations \cite{And2}.  However, to our knowledge, there is no account in the literature regarding an explicit vector field with an explicit degenerate singular cycle to a saddle focus. Moreover, the situation encountered here is quite special as the phase space is of dimension four,  and there is the additional symmetric connection due to the reflectional symmetry of system~\eref{eq:Couplednondim}.

\subsection{Heteroclinic tangency between $W^s(p)$ and $W^u(\Gamma^s)$}
\label{sec:KneadSeqGamma_symm}

\begin{figure}
\centering
\includegraphics[scale=0.95]{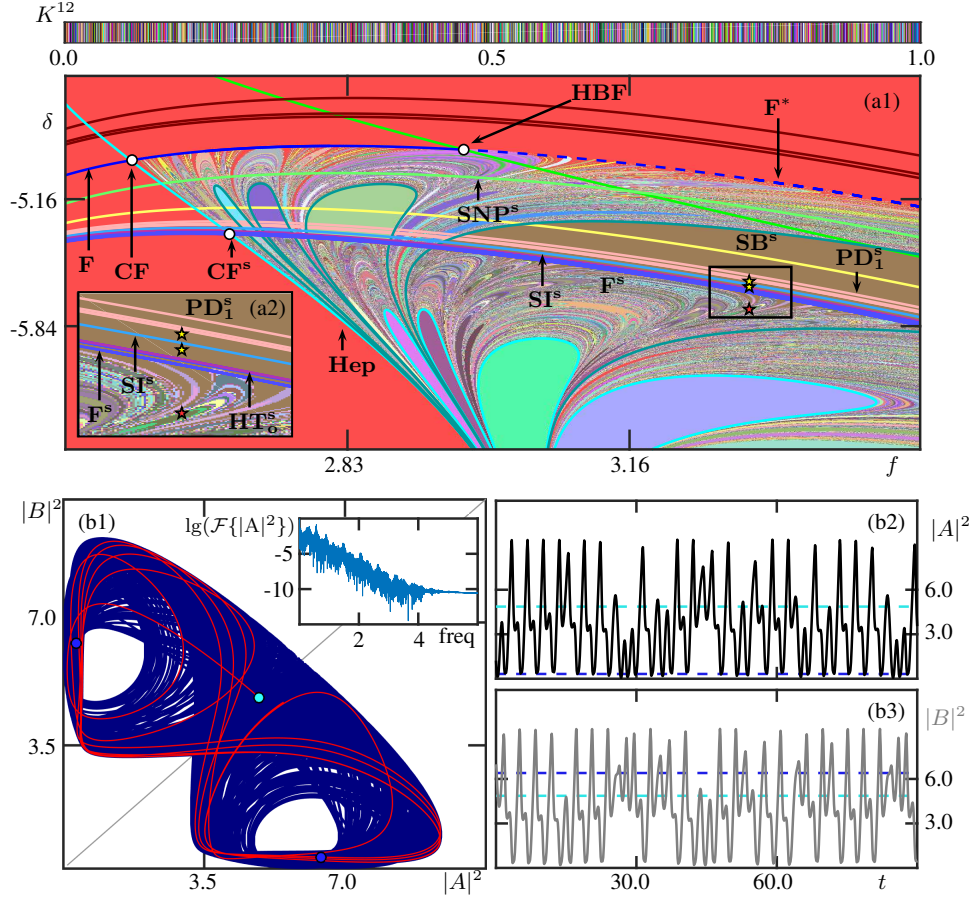} 
\caption{Bifurcations associated with the symmetric attractors arising from $\widehat{\Gamma}_a$ and $\widehat{\Gamma}^*_a$. Panel~(a1) shows the coloring of the $(f,\delta)$-plane by finite kneading sequences $S^{12}$ with the curves from \fref{fig:BifTrans} and additionally the curves $\mathbf{SB^s}$ (yellow) of symmetry breaking bifurcation, $\mathbf{PD_i^s}$ (brown) of successive period-doubling bifurcations, $\mathbf{SI^s}$ (light blue) of symmetry increasing bifurcation, $\mathbf{F^s}$ (blue) of fold bifurcation where $W^s(p)$ and $W^u(\Gamma^s)$ are tangent, and $\mathbf{HT^s_o}$ (magenta) of homoclinic tangency of $\Gamma^s_o$; see also the inset~(a2) and compare with \fref{fig:BifEtoPSym}. At the yellow stars in panels~(a) one finds the chaotic attractors from \fref{fig:TrueSymChaos}(b) and~(c), and at the red star at $(f,\delta)=(3.3,-5.75)$ there is the chaotic attractor shown in panels~(b); see also \fref{fig:ManifoldsTransSInv}.} 
\label{fig:BifFoldHFs}
\end{figure} 

\begin{figure}[h!]
\centering
\includegraphics[scale=0.95]{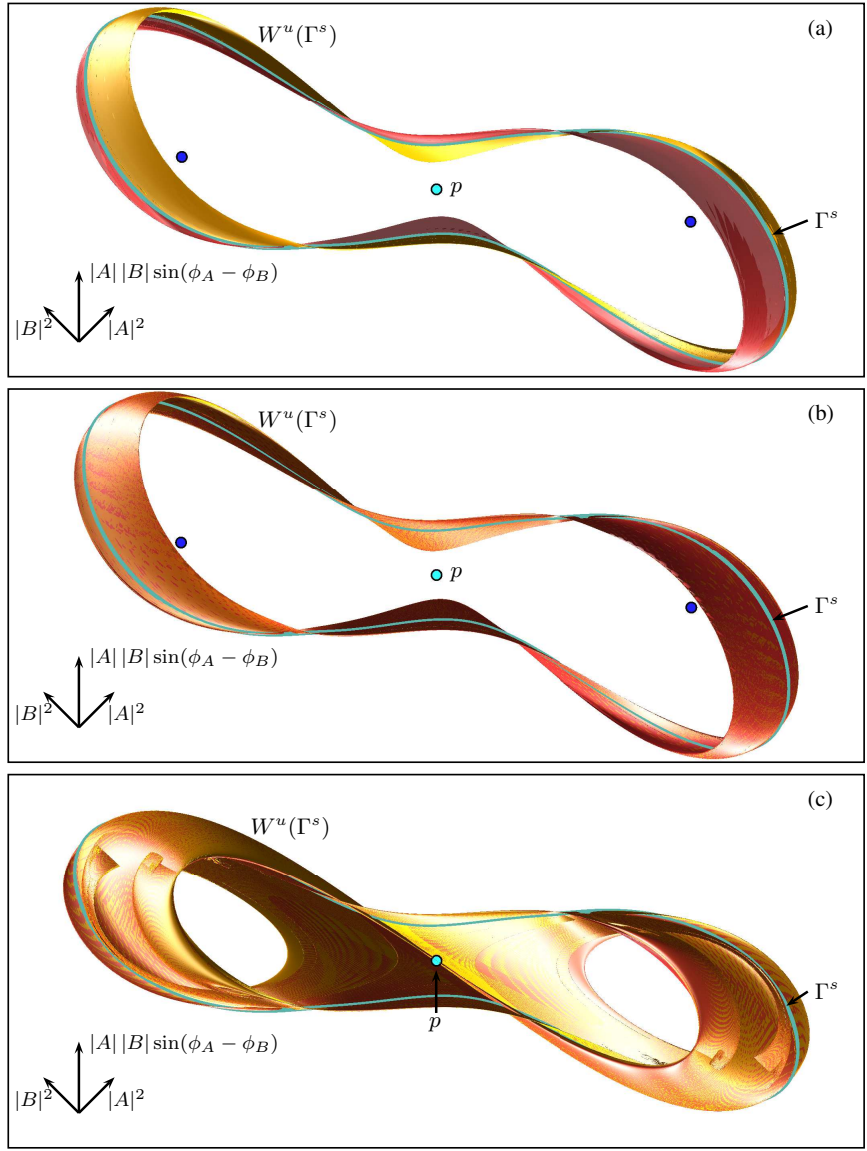} 
\caption{Symmetry increasing and interior crisis bifurcations of chaotic attractors of system \eref{eq:Couplednondim} due to crossing the curves $\mathbf{SI^s}$ and $\mathbf{F^s}$, shown before the curve $\mathbf{SI^s}$ at $(f,\delta)=(3.3,-5.61)$ in panel~(a), after $\mathbf{SI^s}$ but before $\mathbf{F^s}$ at $(f,\delta)= (3.3,-5.63)$ in panel~(b), and after $\mathbf{F^s}$ at $(f,\delta)=(3.3,-5.75)$ in panel~(c); see the stars in \fref{fig:BifFoldHFs}. Shown in $(|A|^2,|B|^2,|A|\,|B|\sin(\phi_A-\phi_B))$-space are the saddle equilibrium $p$, the asymmetric attracting equilibria (blue dots), the S-invariant saddle periodic orbit $\Gamma^s$ (cyan curves) and the two sides of its unstable manifolds $W^{u}(\Gamma^s)$ (gold and red surfaces).} 
\label{fig:ManifoldsTransSInv}
\end{figure} 

\Fref{fig:BifFoldHFs} shows bifurcations associated with further bifurcations of the attractors arising from $\widehat{\Gamma}_a$ and $\widehat{\Gamma}^*_a$ to a chaotic attractor with reflectional symmetry. Panels~(a) show the $(f,\delta)$-plane of \fref{fig:BifTrans} with the additional curves $\mathbf{SB^s}$, $\mathbf{PD_i^s}$ and $\mathbf{SI^s}$ from \fref{fig:BifEtoPSym}, as well as a curve $\mathbf{F^s}$ of fold bifurcation where $W^s(p)$ and $W^u(\Gamma^s)$ are tangent, and a curve $\mathbf{HT^s_o}$ of homoclinic tangency where $W^s(\Gamma^s_o)$ and $W^u(\Gamma^s_o)$ are tangent. 

As we have seen in \fref{fig:TrueSymChaos}(b) and~(c), when $\delta$ is decreased, the two non-symmetric chaotic attractors born in the sequence of period-doublings then merge at the symmetry increasing bifurcation $\mathbf{SI^s}$ to form a chaotic attractor with symmetry. Recall that at $\mathbf{SI^s}$, the manifolds $W^s(\Gamma^s)$ and $W^u(\Gamma^s)$ of the S-invariant saddle periodic orbit $\Gamma^s$, created at the symmetry-breaking bifurcation $\mathbf{SB^s}$, intersect tangentially. This implies that $\Gamma^s$ gets incorporated into the chaotic attractor, which therefore becomes symmetric. \Fref{fig:ManifoldsTransSInv} illustrates what this means for the two-dimensional unstable manifold $W^u(\Gamma^s)$. Before $\mathbf{SI^s}$ as in panel~(a), each of the two sides of $W^u(\Gamma^s)$ accumulates on one of the two non-symmetric chaotic attractors. In contrast, immediately after $\mathbf{SI^s}$ as in panel~(b), both sides of $W^u(\Gamma^s)$ accumulate on the symmetric chaotic attractor, which is contained in the closure of $W^u(\Gamma^s)$. Indeed, this chaotic attractor does not include the saddle equilibrium $p \in \text{Fix}(\eta)$. We also find that the three-dimensional stable manifold $W^s(p)$ is part of the basin boundary of the symmetric chaotic attractor, while both branches of the one-dimensional unstable manifold $W^u(p)$ accumulate on it. In particular, $W^s(p)$ and $W^u(p)$ cannot intersect, which means that there are no Shilnikov bifurcations of $p$. 

When $\delta$ is decreased further, the point $p$ is already incorporated into the chaotic attractor past the curve $\mathbf{F^s}$ where $W^s(p)$ and $W^u(\Gamma^s)$ have a tangency. This curve was computed by identifying a structurally stable heteroclinic connection between $p$ and $\Gamma^s$ below $\mathbf{F^s}$ with Lin's method and subsequently detecting and continuing its fold. The resulting chaotic attractor below $\mathbf{F^s}$ is shown in \fref{fig:ManifoldsTransSInv}(c) as the closure of $W^u(\Gamma^s)$ and in \fref{fig:BifFoldHFs}(b1) as the closure of $W^u_+(p)$. These images show that crossing the fold curve $\mathbf{F^s}$ plays the same role as the fold curve $\mathbf{F}$ from the previous section, in that it incorporates the equilibrium $p$. In particular, the time series in \fref{fig:BifFoldHFs}(b2) and~(b3) clearly demonstrate that the kneading sequence generated by $W^u_+(p)$ is no longer restricted to be an alternation between the two cavities; rather the switching between cavities is now chaotic itself, which is consistent with the fact that the corresponding parameter point in the $(f,\delta)$-plane of \fref{fig:BifEtoPSym}(a) now lies in the region with further sub-isolas, bounded by Shilnikov bifurcation of $p$, within the isola bounded by the curve $\mathbf{Hep^{s}_3}$. Indeed, since it has epochs of regular switching interspersed with those of chaotic switching, we refer to this dynamics as \emph{chaotic behavior with intermittent regular and chaotic switching}. 

Since $W^u(\Gamma^s)$ accumulates on the chaotic attractor, we know that at $\mathbf{F^s}$ the chaotic attractor has already hit its basin boundary, specifically, the stable manifold $W^s(p)$. This is a type of interior crisis; the curve $\mathbf{F^s}$ in \fref{fig:BifFoldHFs}(a) appears to approximate it and, hence, the boundary of the sub-isolas quite well. However, close inspection shows that the chaotic attractor may undergo the interior crisis close to but before reaching the fold $\mathbf{F^s}$ between $W^u(\Gamma^s)$ and $W^s(p)$. This is due to the fact that there are many more saddle periodic orbits in the chaotic attractor with two-dimensional unstable manifolds that all will become tangent to $W^s(p)$ near $\mathbf{F^s}$. Moreover, the three-dimensional unstable manifold $W^s(\Gamma_o^s)$ of the S-invariant saddle periodic orbit $\Gamma_o^s$ (born in the saddle-node bifurcation $\mathbf{SNP}$) also lies in the basin boundary, and one side of its two-dimensional unstable manifold $W^u(\Gamma^s_o)$ accumulates on the chaotic attractor. Hence, a first homoclinic tangency between $W^s(\Gamma^s_o)$ and $W^u(\Gamma^s_o)$ is also a interior crisis. This fold bifurcation can also be identified and continued as the curve $\mathbf{HT^s_o}$ in \fref{fig:BifFoldHFs}(a). Taken together, the curves $\mathbf{F^s}$ and $\mathbf{HT^s_o}$, which are very close together, give a suitable impression of the locus of interior crisis, the exact details of which are well beyond the scope of what is discussed here.

\begin{figure}
\centering
\includegraphics[scale=0.95]{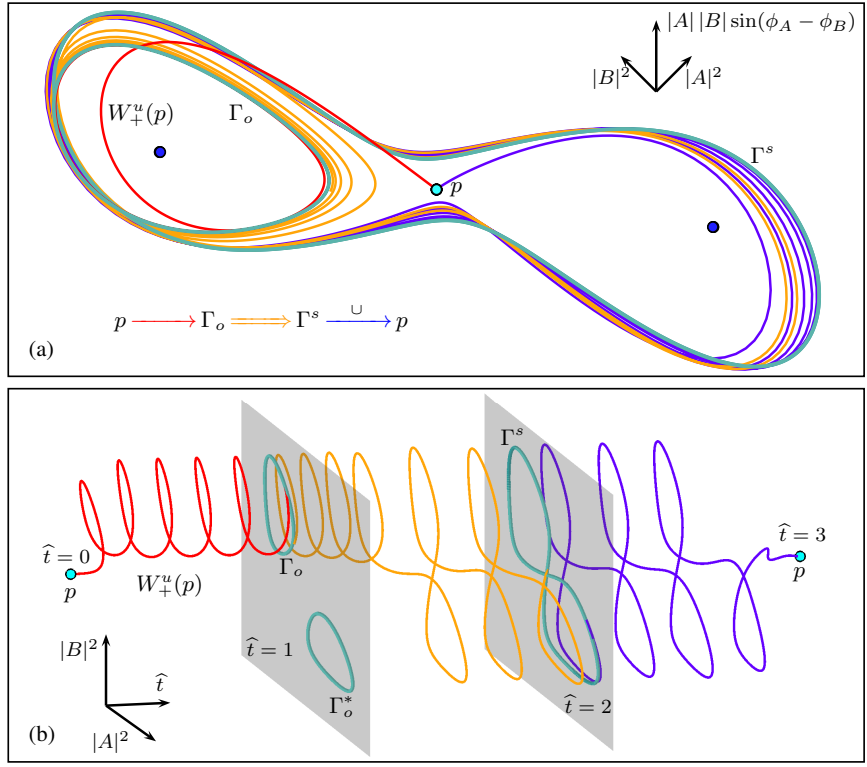} 
\caption{The singular heteroclinic cycle between $p$, $\Gamma_o$ and $\Gamma^s$ at the codimension-two bifurcation point $\mathbf{CF_1^s}$ with $(f,\delta) \approx (2.6932,-5.347)$ where the curves $\mathbf{F^s}$ and $\mathbf{Hep}$ intersect, shown in $(|A|^2,|B|^2,|A|\,|B|\sin(\phi_A-\phi_B))$-space in panel~(a) and as a CTC plot in panel~(b). It consists of the codimension-one EtoP connection from $p$ to $\Gamma_o$ (red curve), a structurally stable connection from $\Gamma_o$ to $\Gamma^s$ (orange curve), and a tangent connection from $\Gamma^s$ back to $p$ (blue curve).} 
\label{fig:CodTwoHFs}
\end{figure} 

As was the case with the curve $\mathbf{F}$, the fold curve $\mathbf{F^s}$ intersects the curve $\mathbf{Hep}$ in a codimension-two point that we label $\mathbf{CF^s}$. Here, one finds the pair of singular heteroclinic cycles 
\begin{eqnarray*}
p \xrightarrow{ {\color[HTML]{000000} \phantom{va \pitchfork va} }} \Gamma_o & \xRightarrow{ {\color[HTML]{000000} \phantom{va  \pitchfork  va} }} & \Gamma^s \xrightarrow{ {\color[HTML]{000000} \phantom{va} \cup \phantom{va} }}  p, \\
p \xrightarrow{ {\color[HTML]{000000} \phantom{va \pitchfork va} }} \Gamma^*_o & \xRightarrow{ {\color[HTML]{000000} \phantom{va  \pitchfork  va} }} & \Gamma^s \xrightarrow{ {\color[HTML]{000000} \phantom{va} \cup \phantom{va} }}  p,
\end{eqnarray*}
where the middle connection from $\Gamma_o$ or $\Gamma_o^*$ to $\Gamma^s$ denote by double arrows is structurally stable, while the two other are of codimension one, respectively. The singular heteroclinic cycle involving $\Gamma_o$ is depicted in \fref{fig:CodTwoHFs}, where panel~(a) is its projection onto $(|A|^2,|B|^2,|A|\,|B|\sin(\phi_A-\phi_B))$-space and panel~(b) the associated CTC plot in $(|A|^2,|B|^2,\widehat{t})$-space. Notice that these heteroclinic cycles are more complicated than the EtoP cycle presented in \fref{fig:CodTwoHF} because they involve two saddle periodic orbits, that is, one additional saddle object. Nevertheless, they serve the same role in the bifurcation diagram as an organizing center for the region, bounded by the respective parts of the curves $\mathbf{F^s}$ and $\mathbf{Hep}$, where infinitely many Shilnikov bifurcation curves exist and accumulate in the $(f,\delta)$-plane.

\section{Tangency bifurcations of periodic orbits and generalized degenerate singular cycles} \label{sec:TanChainCycles}

\begin{figure}
\centering
\includegraphics[scale=0.95]{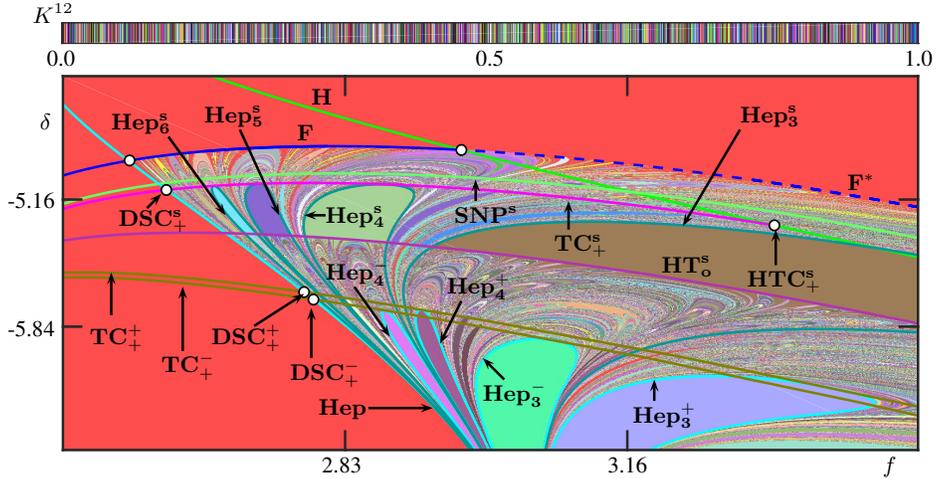} 
\caption{Bifurcation curves of first tangencies of periodic orbits in the $(f,\delta)$-plane, namely the curve $\mathbf{TC_+^s}$ of heteroclinic tangency between $W^u(\Gamma_o)$ and $W^s\Gamma^s_o)$, the curve $\mathbf{TC_+^+}$ of homoclinic tangency between $W^u(\Gamma_o)$ and $W^s\Gamma_o)$, and the curve $\mathbf{TC_+^-}$ of heteroclinic tangency between $W^u(\Gamma_o)$ and $W^s\Gamma_o^*)$; also shown are the fold curves $\mathbf{F}$, $\mathbf{F^*}$ and $\mathbf{F^s}$ from \fref{fig:BifFoldHFs}.
The intersection points of these fold curves with the curve $\mathbf{Hep}$ are denoted $\mathbf{DSC_+^s}$, $\mathbf{DSC_+^+}$ and $\mathbf{DSC_+^-}$, respectively; the curve $\mathbf{TC_+^s}$ ends at the codimension-two point $\mathbf{HTC_+^s}$ on the Hopf bifurcation curve $\mathbf{H}$.} 
\label{fig:BifFoldHFpe}
\end{figure} 

\begin{figure}
\centering
\includegraphics[scale=0.95]{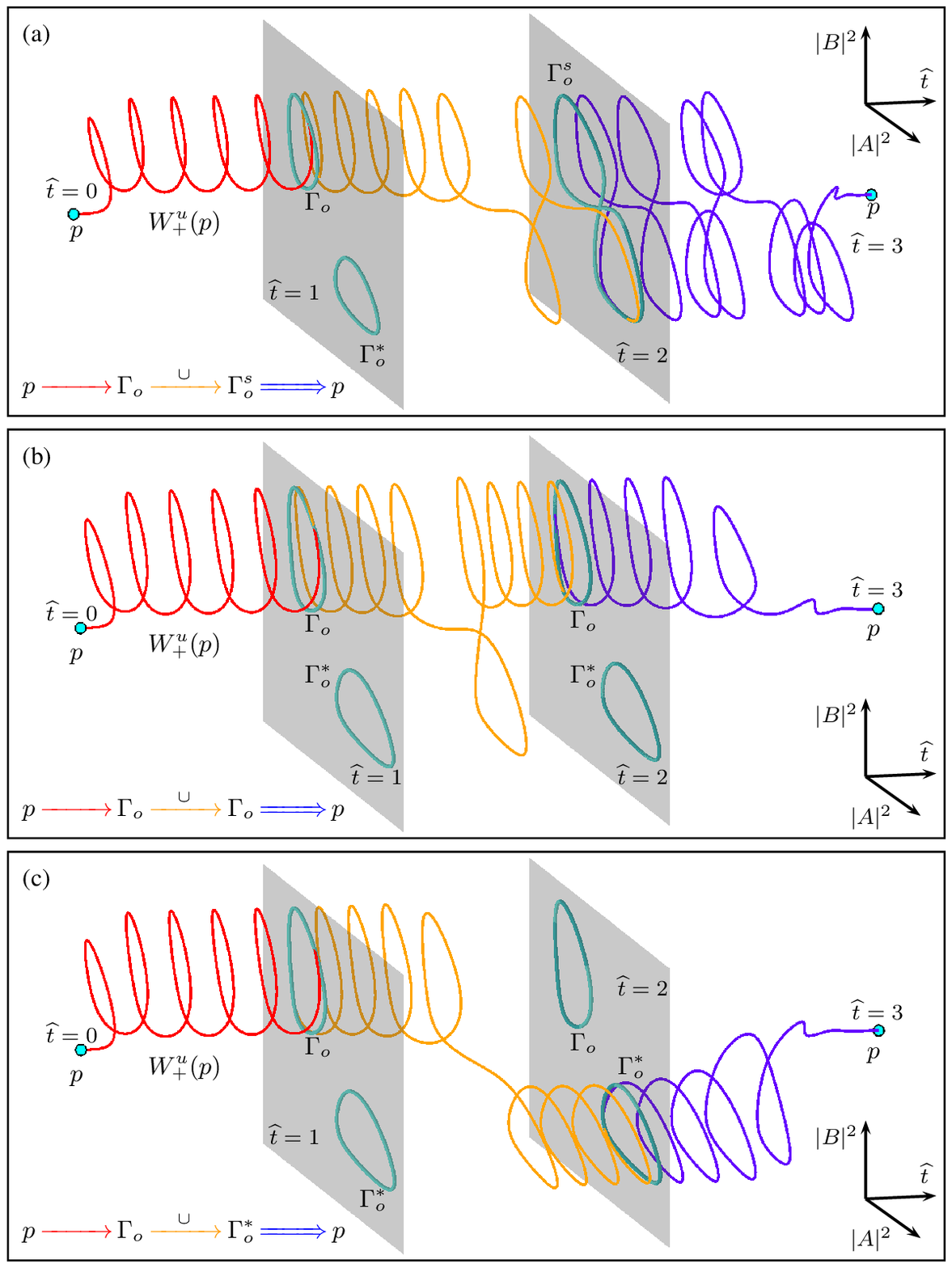} 
\caption{CTC plots of singular heteroclinic cycles on the curve $\mathbf{Hep}$ at the codimension-two points labelled in \cref{fig:BifFoldHFpe}: in panel~(a) between $p$, $\Gamma_o$ and $\Gamma_o^s$ at $\mathbf{DSC_+^s}$ with $(f,\delta) \approx (2.6219,-5.1092)$; in panel~(b) between $p$ and $\Gamma_o$ with a homoclinic loop of $\Gamma_o^s$ at $\mathbf{DSC_+^+}$ with $(f,\delta) \approx (2.7822,-5.6551)$; and in panel~(c) between $p$, $\Gamma_o$ and $\Gamma_o^*$ at $\mathbf{DSC_+^-}$  at $(f,\delta) \approx (2.7923,-5.6952)$. These cycles consists of the EtoP connection from $p$ to the first periodic orbit (red curve), a tangent connection between the two respective periodic orbits (orange curve), and a structurally stable connection from the second periodic orbit back to $p$ (blue curve).} 
\label{fig:CodTwoDSCall}
\end{figure} 

We now identify additional fold curves that act as boundaries for other types of isolas in the $(f,\delta)$-plane and, in particular, the ones related to the bifurcation curves $\mathbf{Hep^{S}_i}$ and $\mathbf{Hep^{+/-}_i}$. In contrast to the fold curves $\mathbf{F}$, $\mathbf{F^*}$ and $\mathbf{F^s}$ these bifurcation curves do not concern tangencies of $W^s(p)$ but tangencies between global manifolds of different saddle periodic orbits. All these saddle periodic orbits have a two-dimensional unstable and three-dimensional stable manifold. Specifically, we find and show in \fref{fig:BifFoldHFpe} the curve $\mathbf{TC_+^s}$ along which there is a heteroclinic tangency between $W^u(\Gamma_o)$ and $W^s(\Gamma^s_o)$, the curve $\mathbf{TC_+^+}$ of homoclinic tangency between $W^u(\Gamma_o)$ and $W^s(\Gamma_o)$, and the curve $\mathbf{TC_+^-}$ of heteroclinic tangency between $W^u(\Gamma_o)$ and $W^s(\Gamma^*_o)$. Of course, along each of these three curves, which all involve the pair of saddle periodic orbits $\Gamma_o$ and $\Gamma^*_o$ that are born at the subcritical Hopf bifurcation $\mathbf{H}$, the respective counterparts under the reflectional symmtery exist as well. Notice that the curve $\mathbf{TC_+^s}$ ends at the codimension-two point $\mathbf{HTC_+^s}$ on the Hopf bifurcation curve $\mathbf{H}$, where $\Gamma_o$ and $\Gamma_o^s$ disappear, and the asymmetric equilibria become saddles. Similar to the continuation of the curve $\mathbf{F}$ as $\mathbf{F^*}$ past the Hopf bifurcation curve $\mathbf{H}$, beyond the point $\mathbf{HTC_+^s}$ there is a pair of tangencies between the two-dimensional unstable manifolds of these saddle with $W^s(\Gamma^s_o)$; the corresponding curve is not shown in \fref{fig:BifFoldHFpe} because it appears to lie very close to $\mathbf{H}$, which makes it quite challenging to continue numerically.

The curves $\mathbf{TC_+^s}$, $\mathbf{TC_+^+}$ and $\mathbf{TC_+^-}$ all intersect the curve $\mathbf{Het^+}$, and the respective codimension-two points $\mathbf{DSC_+^s}$, $\mathbf{DSC_+^-}$ and $\mathbf{DSC_+^+}$ are labelled in \fref{fig:BifFoldHFpe}. At these points one finds the singular heteroclinic cycles shown as CTC plots in \fref{fig:CodTwoDSCall}, which have the symbolic representations 
\begin{eqnarray*}
\mathbf{DSC_+^s}:\qquad
p \xrightarrow{ {\color[HTML]{000000} \phantom{va \pitchfork va} }} \Gamma_o \xrightarrow{ {\color[HTML]{000000} \phantom{va} \cup \phantom{va} }}  \Gamma_o^s \xRightarrow{ {\color[HTML]{000000} \phantom{va  \pitchfork  va} }} p,\\
p \xrightarrow{ {\color[HTML]{000000} \phantom{va \pitchfork va} }} \Gamma^*_o \xrightarrow{ {\color[HTML]{000000} \phantom{va} \cup \phantom{va} }}  \Gamma_o^s \xRightarrow{ {\color[HTML]{000000} \phantom{va  \pitchfork  va} }} p;\\[3mm]
\mathbf{DSC_+^+}: \qquad
p \xrightarrow{ {\color[HTML]{000000} \phantom{va \pitchfork va} }} \Gamma_o \xrightarrow{ {\color[HTML]{000000} \phantom{va} \cup \phantom{va} }}  \Gamma_o \xRightarrow{ {\color[HTML]{000000} \phantom{va  \pitchfork  va} }} p, \\
p \xrightarrow{ {\color[HTML]{000000} \phantom{va \pitchfork va} }} \Gamma^*_o \xrightarrow{ {\color[HTML]{000000} \phantom{va} \cup \phantom{va} }}  \Gamma^*_o \xRightarrow{ {\color[HTML]{000000} \phantom{va  \pitchfork  va} }} p;\\[3mm] \mathbf{DSC_+^-}:\qquad
p \xrightarrow{ {\color[HTML]{000000} \phantom{va \pitchfork va} }} \Gamma_o \xrightarrow{ {\color[HTML]{000000} \phantom{va} \cup \phantom{va} }}  \Gamma^*_o \xRightarrow{ {\color[HTML]{000000} \phantom{va  \pitchfork  va} }} p.\\
p \xrightarrow{ {\color[HTML]{000000} \phantom{va \pitchfork va} }} \Gamma^*_o \xrightarrow{ {\color[HTML]{000000} \phantom{va} \cup \phantom{va} }}  \Gamma_o \xRightarrow{ {\color[HTML]{000000} \phantom{va  \pitchfork  va} }} p.
\end{eqnarray*}
Notice that for all these singular cycles, the last connection from the respective periodic orbit back to $p$ is transversal. This is due to the fact that the respective curves in the $(f,\delta)$-plane of \fref{fig:BifFoldHFpe} lie below the fold curves where these heteroclinic connections are created. More specifically, the curves $\mathbf{TC_+^+}$ and $\mathbf{TC_+^-}$ lie below the curve $\mathbf{F}$ where heteroclinic connections from $\Gamma_o$ and $\Gamma_o^*$ to $p$ emerge. Moreover, the curve $\mathbf{TC_+^s}$ lies below the fold curve of connections from $\Gamma_o^s$ to $p$, which is really close to the $\mathbf{SNP^s}$ curve (not shown in this figure).

Observe from the bifurcation diagram in \fref{fig:BifFoldHFpe} that the curve $\mathbf{TC_+^s}$ is tangent to the isolas of constant kneading sequence bounded by
the codimension-one curves $\mathbf{Hep^{s}_i}$ of EtoP connections to $\Gamma^s_o$. Moreover, these tangency points along $\mathbf{TC_+^s}$ converge for increasing $\mathbf{i}$ to the point $\mathbf{DSC_+^s}$ on the curve $\mathbf{Hep}$; this also means that the curves $\mathbf{Hep^{S}_i}$ themselves and the isolas they bound accumulate on the segment of $\mathbf{Hep}$ below $\mathbf{DSC_+^s}$. This property of the curve $\mathbf{TC_+^s}$ is explained by the fact that the tangent heteroclinic connection on it is the one between periodic orbits in the middle of the cycle, specifically between the pair $\Gamma_o$ and $\Gamma_o^*$, respectively, and the S-invariant $\Gamma^s_o$. 

Similarly, the fold curve $\mathbf{TC_+^+}$ in \fref{fig:BifFoldHFpe} is tangent to the curves $\mathbf{Hep^{+}_i}$ of EtoP connections of $W^u_+(p)$ to $\Gamma_o$ (with a single excursion to the region of phase space with $|A|^2 < |B|^2$), with the points of tangency on $\mathbf{TC_+^+}$ accumulating on the codimension-two point $\mathbf{TC_+^+}$; this agrees with the fact that the tangent connection at $\mathbf{DSC_+^+}$ shown in \fref{fig:CodTwoDSCall}(b) from $\Gamma_o$ back to itself has a single loop with $|A|^2 < |B|^2$. Likewise, the fold curve $\mathbf{TC_+^-}$ is tangent to the curves $\mathbf{Hep^{-}_i}$ of EtoP connections of $W^u_+(p)$ to $\Gamma_o^*$ (with a single transition into the region of phase space with $|A|^2 < |B|^2$); the respective tangency points on $\mathbf{TC_+^+}$ also converge on the codimension-two point $\mathbf{TC_+^-}$, where the tangent heteroclinic connection shown in \fref{fig:CodTwoDSCall}(c) is directly from $\Gamma_o$ to its symmetric counterpart $\Gamma_o^*$. 

Overall, we observe that the codimension-two singular cycles at the points $\mathbf{DTC_+^s}$, $\mathbf{DTC_+^+}$ and $\mathbf{DTC_+^-}$ play a role for the organization of the respective families of EtoP connections similar to that of the codimension-two singular EtoP cycles at the points $\mathbf{CF}$ and $\mathbf{CF^s}$ for the respective families of Shilnikov bifurcations \cite{ Kirk1, Lohr1}. While the investigation here sheds considerable light on their role and generalizing nature of the $\mathbf{CF}$ point, determining the unfoldings of such more complicated singular cycles remains a challenge for future research.

\section{Regions of chaotic dynamics and role of degenerate singular cycles} \label{sec:RegionDegSin}

\begin{figure}
\centering
\includegraphics[scale=0.95]{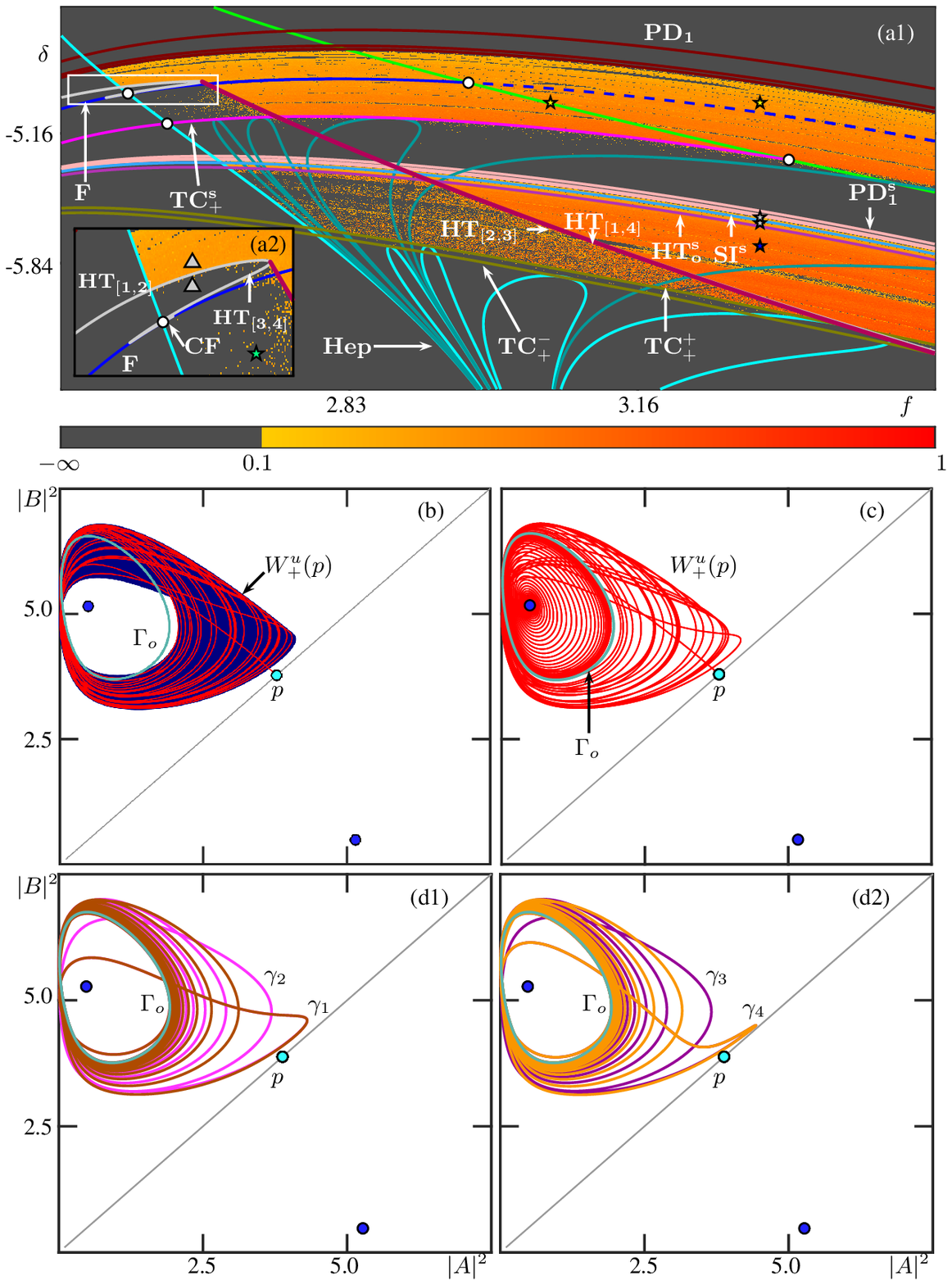} 
\caption{Coloring of the $(f,\delta)$-plane in panel~(a1) by the maximum Lyapunov exponent associated with $W^u_+(p)$, with previously shown bifurcation curves from \fref{fig:BifFoldHFs} and \fref{fig:BifFoldHFpe}; at the yellow stars one finds the chaotic attractors in \fref{fig:chaos}, and at the grey and blue stars those in \fref{fig:ManifoldsTransSInv}. Shown additionally are the curves $\mathbf{HT_{[i,j]}}$ for $\mathbf{i,j}=1,...,4$ of fold bifurcations of homoclinic orbits to $\Gamma_o$; see also the enlargement panel~(a2), where two parameter points are indicated by grey triangles: at $(f,\delta) = (2.6, -4.89)$ as in panel~(b), $W^u_+(p)$ accumulates on a (nonsymmetric) chaotic attractor, while at $(f,\delta) = (2.6, -4.91)$ as in panel~(c), $W^u_+(p)$ converges to an attracting nonsymmetric equilibrium. Panels~(d1) and~(d2) show the four transversal homoclinic orbits to $\Gamma_o$, labeled $\gamma_1$ to $\gamma_4$, that exist simulaneously for $(f,\delta)=(2.65,-4.985)$, the green star in panel~(a2).} 
\label{fig:LyapunovBifUltraZoom}
\end{figure} 

From \fref{fig:BifFoldHFpe} one might expect chaotic dynamics in all the regions of the $(f,\delta)$-plane with many intermingling colors, representing a sensitivity in the kneading sequences generated by $W^u_+(p)$. To check whether this is the case, we now complement our analysis by computing the maximum Lyapunov exponent \cite{Lya1} associated with the trajectory $W^u_+(p)$ over a fine grid in our region of interest of the $(f,\delta)$-plane. \Fref{fig:LyapunovBifUltraZoom}(a1) shows the resulting coloring of the $(f,\delta)$-plane, together with the previous bifurcation curves $\mathbf{Hep}$, $\mathbf{H}$, $\mathbf{PD_i}$, $\mathbf{F}$, $\mathbf{F^*}$, $\mathbf{TC_+^s}$, $\mathbf{PD_i^s}$, $\mathbf{SI^s}$, $\mathbf{HT_o^s}$, $\mathbf{TC_+^+}$ and $\mathbf{TC_+^-}$ that appeared already in \fref{fig:BifFoldHFs} and \fref{fig:BifFoldHFpe}. As is to be expected from the discussion in the previous sections, the accumulation curve of the curves $\mathbf{PD_i}$ and $\mathbf{PD_i^s}$ of period doublings, as well as the curves of the $\mathbf{TC_+^s}$ and $\mathbf{TC_+^-}$ of heteroclinic tangencies involving $\Gamma_o$, are seen to bound regions where the Lyapunov exponent is quite consistently positive, which indicates chaotic dynamics. Notice that the folds associated with the symmetry-increasing bifurcations do not affect the Lyapunov exponent as expected. Indeed, the chaotic attractors from \fref{fig:chaos} and \fref{fig:ManifoldsTransSInv} lie in these regions, at the yellow stars and at the grey and blue stars, respectively. We stress that the Lyapunov exponent shown in panels~(a) is that generated by $W^u_+(p)$. Hence, a negative Lyapunov exponent implies that $W^u_+(p)$ does not accumulate on a chaotic attractor; it does not necessarily indicate that no chaotic attractor exists. Closer inspection of the spotty area in panel~(a), below the curves $\mathbf{HT_{[1,2]}}$ and $\mathbf{HT_{[1,4]}}$, shows that  $W^u(\mathbf{0})$ is attracted to a stable equilibrium or periodic orbit; thus, the positive Lyapunov exponent in this region is an artefact caused by the need for a very large integration time to determine the fate of $W^u_+(\mathbf{0})$.  Also, notice the gray regions inside the solid orange area. They correspond to periodic windows when an attracting periodic orbit is created by saddle-node bifurcations; for example, the biggest window corresponds to a period-three orbit. Interestingly, below the fold curve $\mathbf{F}$ the windows correspond to S-invariant period orbits; for example, in the biggest window one finds the attracting periodic orbit with kneading $S= (\overline{++--})$.

Somewhat surprisingly, \fref{fig:LyapunovBifUltraZoom}(a1) shows that chaotic dynamic is accessible to $W^u_+(p)$ only for sufficiently large $f$, as determined by the additional bifurcation curves labeled $\mathbf{HT_{[1,2]}}$ and $\mathbf{HT_{[1,4]}}$; see also the enlargement in panel~(a2). Notice, in particular, that the curve $\mathbf{Hep}$ represents a transition to chaotic dynamics accessed by $W^u_+(p)$ only above its intersection point with the new curve $\mathbf{HT_{[1,2]}}$. The loss of the chaotic attractor in the transition across the curve $\mathbf{HT_{[1,2]}}$ is illustrated with \fref{fig:LyapunovBifUltraZoom}(b) and~(c), which show the fate of $W^u_+(p)$ in the $(|A|^2,|B|^2)$-plane for the parameter points indicated by the two grey triangles in panel~(a2). Above $\mathbf{HT_{[1,2]}}$, as in panel~(b) for $(f,\delta) = (2.6, -4.89)$, the branch $W^u_+(p)$ accumulates on a nonsymmetric chaotic attractor; below $\mathbf{HT_{[1,2]}}$, as in panel~(c) for $(f,\delta) = (2.6, -4.91)$, on the other hand, $W^u_+(p)$ converge to the upper nonsymmetric attracting equilibrium. 

The curves $\mathbf{HT_{[1,2]}}, \mathbf{HT_{[3,4]}}, \mathbf{HT_{[1,4]}}$ and $\mathbf{HT_{[2,3]}}$ correspond to fold bifurcations between the manifolds $W^u(\Gamma_o)$ and $W^s(\Gamma_o)$ of the saddle periodic orbit $\Gamma_o$. In particular, the bifurcation $\mathbf{HT_{[1,2]}}$ generates a pair of transverse homoclinic orbits to $\Gamma_o$; they are shown in \fref{fig:LyapunovBifUltraZoom}(d1), and we refer to them as $\gamma_1$ and $\gamma_2$, respectively. When the curve $\mathbf{HT_{[3,4]}}$ in panel~(a2) is also crossed, then there is another fold bifurcation that creates a second pair of transverse homoclinic orbits to $\Gamma_o$, which we show in \fref{fig:LyapunovBifUltraZoom}(d2) and label $\gamma_3$ and $\gamma_4$, respectively. Note that all four transverse homoclinic orbits $\gamma_1$ to $\gamma_4$ coexist in the region below $\mathbf{HT_{[3,4]}}$; those in panels~(d1) and~(d2) are for the same parameter point indicated by the green star in panels~(a2). Moreover, at the two branches $\mathbf{HT_{[2,3]}}$ and $\mathbf{HT_{[1,4]}}$ the transverse homoclinic orbits $\gamma_2$ and $\gamma_3$, and $\gamma_1$ and $\gamma_4$ coalesce, respectively. Note that the curves $\mathbf{HT_{[2,3]}}$ and $\mathbf{HT_{[1,4]}}$ are very close to each other in the $(f,\delta)$-plane, which means that these homoclinic orbits to $\Gamma_o$ appear in quick succession when $f$ is increased; nevertheless, the two curves can be identified and distinguished clearly during continuation, because they concern quite different objects in phase space. Notice that the two curves $\mathbf{HT_{[1,2]}}$ and $\mathbf{HT_{[3,4]}}$ are branches of a single smooth curve of folds that intersects, at the codimension-two point $\mathbf{FF}$, a second single smooth curve which is formed by $\mathbf{HT_{[2,3]}}$ and $\mathbf{HT_{[1,4]}}$. Determining this intersection conclusively, both theoretically \cite{Kirk1} and numerically, is a challenging task. Nevertheless, we find that the point $\mathbf{FF}$ lies very near the point of largest curvature of either of the two fold curves, and our computations strongly suggest that at $\mathbf{FF}$ the branch $\mathbf{HT_{[1,2]}}$ becomes $\mathbf{HT_{[3,4]}}$ and the branch $\mathbf{HT_{[2,3]}}$ becomes $\mathbf{HT_{[1,4]}}$. 

\begin{figure}[t!]
\centering
\includegraphics[scale=0.95]{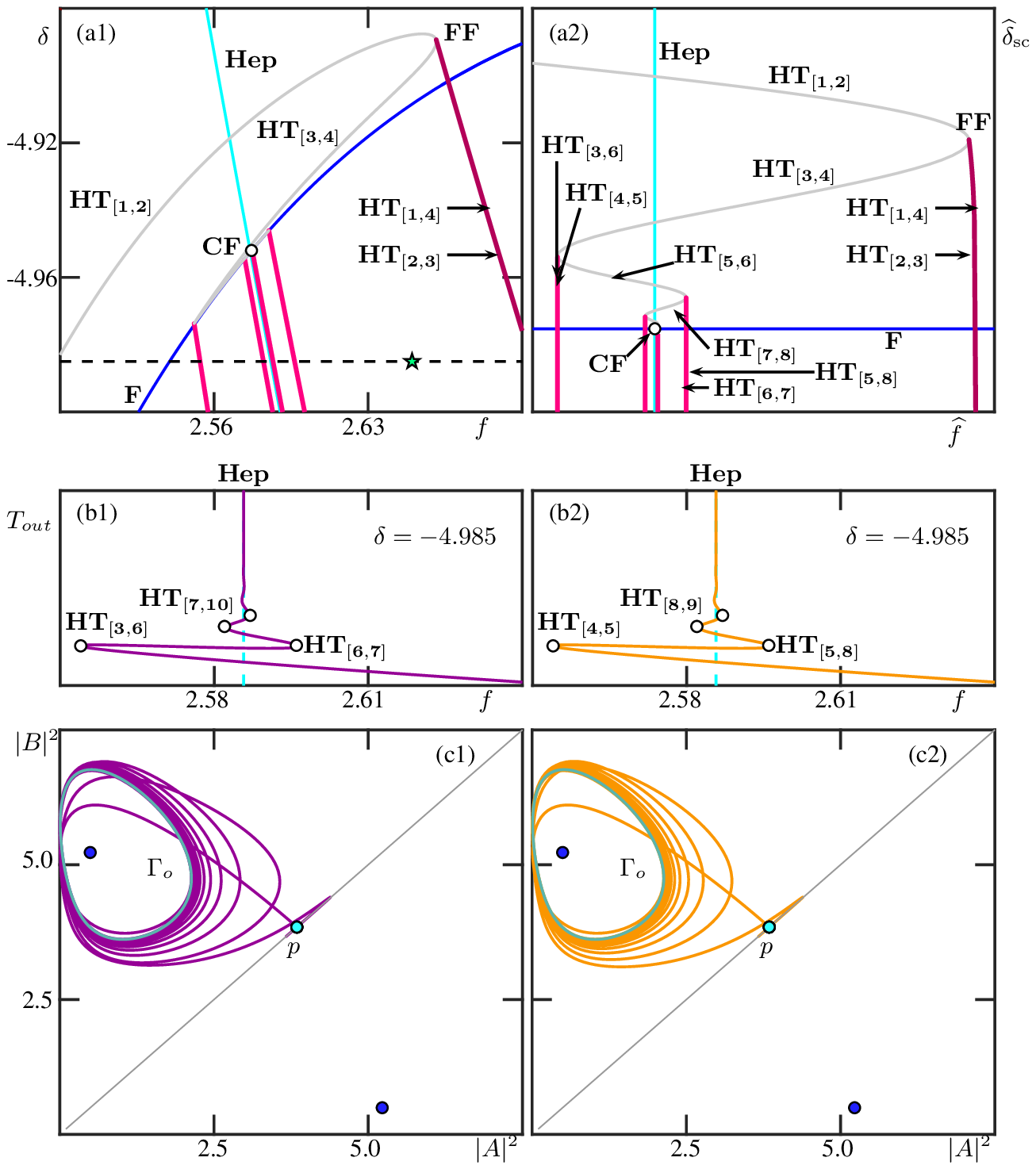} 
\caption{Bifurcation diagram near the codimension-two point $\mathbf{CF}$ in the $(f,\delta)$-plane. Panel~(a1) is an enlargement of \fref{fig:LyapunovBifUltraZoom}(a2) near the point, and panel~(a1) shows the same data after a coordinate change to straightened parameters $\widehat{f}$ and $\widehat{\delta}$, where the curves $\mathbf{Hep}$ and $\mathbf{F}$ are the coordinate axes. The grey curve, the continuation of $\mathbf{HT_{[3,4]}}$ oscillates into the point $\mathbf{CF}$, which generates further pairs of curves $\mathbf{HT_{[i,i+1]}}$, $\mathbf{i}=2,4,6$ and $\mathbf{HT_{[i,i+3]}}$, $\mathbf{i}=1,3,5$ (fuchsia) of fold bifurcations of homoclinic orbits to $\Gamma_o$. Panels~(b1) and~(b2) show the continuation in $f$ for $\delta=-4.985$ of the two transversal homoclinic orbits $\gamma_3$ and $\gamma_4$ from \fref{fig:LyapunovBifUltraZoom}(d2), where the time $T_{out}$ is the time that the homoclinic orbit spends outside a tubular neighbourhood of $\Gamma_o$. Panels~(c1) and~(c2) show the two respective homoclinic orbits with large $T_{out}$ as they approach the bifurcation $\mathbf{Hep}$.} 
\label{fig:singlecodim2EtoP} 
\end{figure} 

Observe in \fref{fig:LyapunovBifUltraZoom}(a2) that the branch $\mathbf{HT_{[3,4]}}$ appears to oscillate into the point $\mathbf{CF}$. \Fref{fig:singlecodim2EtoP} shows that this is indeed the case and that infinitely many additional pairs of fold bifurcations are generated in the process. Panel~(a1) is an enlargement of the $(f,\delta)$-plane near the points $\mathbf{CF}$ and $\mathbf{FF}$. Also shown are branches $\mathbf{HT_{[i,j]}}$ (not labeled) of further folds of homoclinic orbits to $\Gamma_o$. To clarify the bifurcation diagram near the global codimension-two point $\mathbf{CF}$ in the style of an unfolding, we present in panel~(a2) the same curves in a straightened out manner. This image was generated from the computed data by a smooth coordinate change, constructed by splines, that maps the curves  $\mathbf{F}$ and $\mathbf{Hep}$ in the $(f,\delta)$-plane to the coordinate axes in the $(\widehat{f}, \widehat{\delta})$-plane; additionally, we rescale the positive $\widehat{\delta}$-axis to  $\widehat{\delta}_{\rm sc}$ so that the oscillating nature of the grey fold curve becomes more visible. Subsequent intersection points of the grey fold curve with $\mathbf{Hep}$ now clearly show exponential convergence to the point $\mathbf{CF}$.

This numerically computed bifurcation diagram agrees with the conjectured unfolding in \cite[Fig. 15~(b)]{Kirk1}, which was developed for a geometrical two-dimensional map reflecting the dynamics close to a degenerate singular cycle between a saddle focus equilibrium and a saddle periodic orbit. It was shown in \cite{Kirk1} that to leading order the grey curve satisfies
\begin{equation}\label{eq:leadingOrder}
\widehat{f} = -c \widehat{\delta}^{\, \eta_1} \cos \left(\Phi- \eta_2\ln{\widehat{\delta}}\right),
\end{equation}
where $c$ and $\Phi$ are constants, and $\eta_1 = | \text{Re}(\lambda_s) |/\lambda_u$ and $\eta_2= \text{Im}(\lambda_s)/\lambda_u$; here, $\lambda_s$ and $\lambda_u$ are the complex stable and the real unstable eigenvalues of the saddle focus equilibrium at the point $\mathbf{CF}$, respectively. In particular, if one takes the maxima in $\widehat{f}$ of Eq.~\cref{eq:leadingOrder} then it is clear that $\widehat{f} = -c \widehat{\delta}^{\, \eta_1}$, while minima in $\widehat{f}$ satisfies $\widehat{f} = c \widehat{\delta}^{\, \eta_1}$. Applying a least-squares approximation to the logarithm of both the maxima and minima of our numerical data in the $(\widehat{f},\widehat{\delta})$-plane yields the estimate $\widetilde{\eta_1} \approx 0.3495$, which is very close to the computed eigenvalue ratio  $\eta_1 \approx 0.3633$. By considering this curve in the $(\widehat{f},\ln(\widehat{\delta}))$-plane, where the maxima and minima of $\widehat{f}$ are evenly spaced as $\widehat{\delta}$ increases, one approximately recover the period of the oscillation; the average distance in $\widehat{\delta}$ of succesive maxima gives the estimate $\widetilde{\eta_2} \approx 0.9913$, while the eigenvalue ratio is  $\eta_2 \approx 0.9835$. This shows that the conjectured unfolding in \cite{Kirk1} and our numerical computations for the vector field case agree at a quantitative level. \Fref{fig:singlecodim2EtoP}(a2) also shows more clearly the additional pairs of branches $\mathbf{HT_{[i,i+1]}}$ with $\mathbf{i}=2,4,6$ and $\mathbf{HT_{[i,i+3]}}$ with $\mathbf{i}=1,3,5$ that emerge from additional fold-fold points on the grey curve (located extremely close to folds with respect to $\widehat{f}$). Going along the grey curve past each fold-fold point generates a new branch $\mathbf{HT_{[i,i+1]}}$, with $\mathbf{i}=5,7,\ldots$, which generates a new and coexisting pair $\gamma_i$ and $\gamma_{i+1}$ of transverse homoclinic orbits to $\Gamma_o$. The overall geometry is that of a folded, two-sheeted surface (think of a folded piece of paper) \cite{arnold1985singularities,Ali1} that folds back and forth with decreasing amplitude. 

Having an explicit system that exhibits this degenerate singular cycle $\mathbf{CF}$ allows us to understand the nature of the folds that accumulate on it. We consider now the situation where we continue in $f$ the homoclinic orbits $[3]$ and $[4]$ through their respective folds for fixed $\delta$ below the point $\mathbf{FF}$. Specifically for $\delta = -4.985$, which is the horizontal line through the green star in \fref{fig:LyapunovBifUltraZoom}(a2), which also appears in \fref{fig:singlecodim2EtoP}(a1). The results of these two continuations are shown in \fref{fig:singlecodim2EtoP}(b1) and~(b2), where the respective homoclinic orbit is represented by the time $T_{out}$ that it spends outside a tubular neighborhood of $\Gamma_o$. Notice from these panels that we observe what one might call a `Shilnikov bifurcation of EtoP cycles', characterized by the accumulation of a homoclinic orbit to a periodic orbit onto heteroclinic cycles, here involving $p$ and $\Gamma_o$. The last computed homoclinic orbits for large $T_{out}$ that approach the bifurcation $\mathbf{Hep}$ are shown in panels~(c1) and~(c2), respectively. Notice that each of these limits is a different codimension-one EtoP cycle between $p$ and $\Gamma_o$; as the point $\mathbf{CF}$ is approached, these two distinct EtoP cycles approach each other and then coalesce at $\mathbf{CF}$, where they disappear. The connections from $\Gamma_o$ to $p$ are structurally stable, in both panels, and they become the same tangent connection at $\mathbf{CF}$ when continued to $\mathbf{Hep}$. While a more complete study of the codimension-two global bifurcation point $\mathbf{CF}$ --- the focus case of the degenerate singular cycle studied in \cite{ Kirk1, Lohr1} --- remains an interesting challenge for future research, we conclude here that \fref{fig:singlecodim2EtoP}(a2) represents a concise and consistent numerical unfolding (of which there may be several). Furthermore, the explicit system~\eref{eq:Couplednondim} offers opportunities for future study of the other degenerate singular cycle $\mathbf{CF^s}$ (which exists by $\Z_2$-equivariance), and the singular cycles $\mathbf{DSC^s_+}$, $\mathbf{DSC^+_+}$ and $\mathbf{DSC^-_+}$ when the tangent connections occur between saddle periodic orbits.

\section{Perspectives and Conclussions} \label{sec:conclusions}

\begin{figure}
\centering
\includegraphics[scale=0.95]{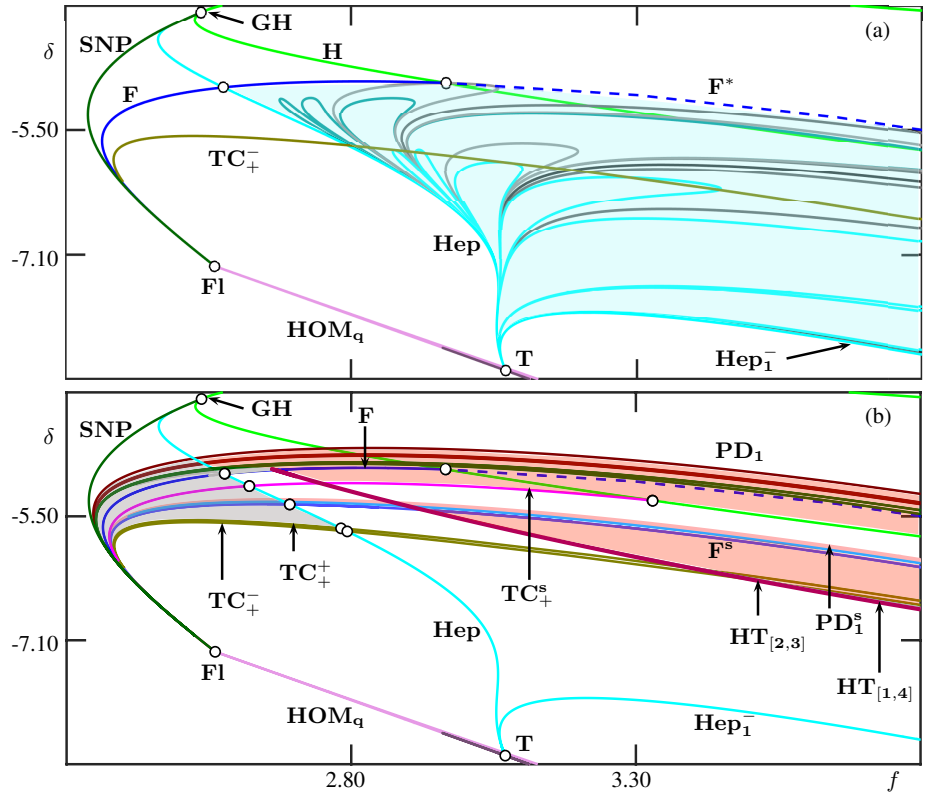} 
\caption{Bifurcation diagrams in the $(f,\delta)$-plane of system~\eref{eq:Couplednondim} showcasing the curves that emanate from the Bykov T-point $\mathbf{T}$ in panel~(a), and the ones that emanate from the flip bifurcation $\mathbf{Fl}$ in panel~(b). Shilnikov bifurcations may be encounterd only in the blue highlighted region in panel~(a), and $W^u_+(p)$ may converges to a chaotic attractor only in the red highlighted region in panel~(b).} 
\label{fig:LastBif}
\end{figure} 

We finish by providing an overall perspective of the bifurcation structure we found in system~\eref{eq:Couplednondim}. To this end, \fref{fig:LastBif} shows two final bifurcation diagrams that highlight the roles of the Bykov T-point $\mathbf{T}$ and the flip bifurcation $\mathbf{Fl}$, respectively, as organizing centers of the $(f,\delta)$-plane. Panel~(a) shows that all of the different curves of codimension-one connecting orbits we studied --- Shilnikov homoclinic orbits and EtoP connections involving different periodic orbits --- emerge from the point $\mathbf{T}$; more specifically, near $\mathbf{T}$ they exist in a wedge (cyan region) bounded by the two primary EtoP connections: $\mathbf{Hep}$ with kneading sequence $S=(+\overline{+})$ and $\mathbf{Hep_1^-}$ with kneading sequence $S=(+\overline{-})$. Along a sufficiently small circle around this point we expect all possible global bifurcations to be present that are generated by $\mathbf{T}$ and involve the saddle equilibrium $p$. On the other hand, further away certain global bifurcation curves form loops, which means that they are no longer encountered along a larger circle. As we have seen, such loops are closely associated with different types of fold curves of certain global invariant manifolds. In particular, we find that the curves of Shilnikov bifurcations are tangent to the fold curve $\mathbf{F}$ and $\mathbf{F^*}$ \cite{Kirk1, Lohr1} and curves of EtoP connections to the fold curves $\mathbf{TC_+^s}$, $\mathbf{TC_+^+}$ and $\mathbf{TC_+^-}$.

In turn, \fref{fig:LastBif}(b) highlights the fold curve $\mathbf{F}$ and all the other fold curves we considered, as well as the local bifurcations associated with the periodic orbits we have studied in this paper. Notice that all these bifurcations can all be continued as curves to the flip bifurcation point $\mathbf{Fl}$. We have numerical evidence that at this point the two-dimensional stable manifold of the symmetric real saddle equilibrium involved in the $\mathbf{T}$-point is at an inclination flip configuration. To our knowledge, the only results concerning a flip bifurcation in a $\mathbf{Z_2}$-equivariant system are those in \cite{Golmakani2011}, where it is studied how Lorenz-like attractors arise from a symmetric orbit flip bifurcation in three-dimensional vector field. Our results suggest the ingredients of an unfolding, in the sense that all curves we found to emerge from $\mathbf{Fl}$ in between the curves $\mathbf{SNP}$ and $\mathbf{TC^-_+}$ will be encountered along a sufficiently small circle around this point. 

When global bifurcation curves and, in particular, the fold curves of invariant manifolds cross the curve $\mathbf{Hep}$ then additional bifurcations emerge from the respective codimension-two crossing points. One specific example we considered in more detail is the emergence from the crossing point of $\mathbf{F}$ and $\mathbf{Hep}$ of additonal curves $\mathbf{HT_{[i,j]}}$ of fold bifurcations of homoclinic orbits to a peridic orbit; these turned out to play a role for bounding the region where the branch $W^u_+(p)$ converges to a chaotic attractor, and this region is highlighted in \fref{fig:LastBif}(b).

Overall, we discovered an intringuing bifurcation structure due to the interaction of the points $\mathbf{T}$ and $\mathbf{Fl}$ in the two-site $\Z_2$-equivariant and four-dimensional semiclassical Bose--Hubard dimer model. More specifically, certain curves emanating from $\mathbf{T}$ and $\mathbf{Fl}$, respectively, cross transversely to create many additional codimension-two points, generating cascades of additional global bifurcations. Our numerical evidence suggests that there exist infinitely many codimension-two points along the curve $\mathbf{Hep}$ in between the fold curves $\mathbf{F}$ and $\mathbf{TC_+^-}$. These points organize finer details in the organization of the kneading invariant in the bifurcation diagram that we showed, and they should be associated with saddle periodic orbits that emanate from the local bifurcations. It is natural to suspect that there is an underlying organizing center of codimension-three, characterized by the confluence of the codimension points $\mathbf{T}$ and $\mathbf{Fl}$ along the curve $\mathbf{HOM_q}$ in \fref{fig:LastBif}. The further study of this and other global bifurcations is beyond the scope of this paper and remains an interesting and challenging subject for future research. 

By considering the global bifurcation associated with the saddle focus $p$ and the saddle periodic orbits $\Gamma_o$, $\Gamma^s_o$ and $\Gamma^s$ only, we already obtained a detailed characterization of different dynamics. In particular, this includes different types of chaotic behavior in the red region of \fref{fig:LastBif}(b), whose different switching properties implies that it might be possible to distinguish them in an experiment. Localized, \emph{non-switching chaotic behavior} arises through the period-doubling cascade $\mathbf{PD_i}$ of asymmetric periodic orbits; it undergoes a boundary crisis at the homoclinic fold bifurcation curves $\mathbf{HT_{[1,2]}}$ and $\mathbf{HT_{[1,4]}}$. The symmetry increasing transition when the fold bifurcation $\mathbf{F}$ is crossed then leads to \emph{chaotic behavior with chaotic switching}; the associated chaotic attractor undergoes a boundary crisis at the curves  $\mathbf{HT_{[1,4]}}$ or $\mathbf{TC_+^s}$. Similarly, \emph{chaotic behavior with regular switching} arises from the period-doubling cascade $\mathbf{PD^s_i}$ of the pair of asymmetric periodic orbits $\widehat{\Gamma}_a$ and $\widehat{\Gamma}^*_a$, and it features a subsequent symmetry increasing bifurcation at the curve $\mathbf{SI^s}$. An interior crisis bifurcation of the associated chaotic attractor near the curves $\mathbf{F^s}$ and $\mathbf{HT^s_o}$ then creates \emph{chaotic behavior with intermittent regular and chaotic switching}; this attractor undergoes a boundary crisis at the $\mathbf{HT_{[1,4]}}$, $\mathbf{TC_+^+}$ or $\mathbf{TC_+^-}$.

Finally, it will be an interesting challenge to identify the different types of dynamics, including the different types of chaotic switching between the two cavities, in an actual experiment. This will require an increase of the strength of the (optical) driving beyond what has been achieved so far. A related question is how chaotic and other behavior identified in the semiclassical approximation studied here manifests itself in measurements of systems near the quantum limit, such as photonic nanocavities with only quite few photons \cite{Haddadi_2014}. Providing insights about quantum chaos and the transition between the classical and quantum regimes will require the analysis of statistical properties of open quantum systems; this is a challenging topic of ongoing research \cite{PhysRevE.67.066203, PhysRevResearch.2.033131, PhysRevA.102.063702, wimberger2014nonlinear}.

\section*{Acknowledgments}
We thank Bruno Garbin, Ariel Levenson and Alejandro Yacomotti for sharing their insight, especially concerning their experiments. We also thank Cris Hasan, Stefan Ruschel and Kevin Stitely for insightful discussions. 

\bibliographystyle{siam}
\bibliography{GBK_BHchaosArXiv}

\begin{thebibliography}{10}

\bibitem{TIDES2012}
{\sc A.~Abad, R.~Barrio, F.~Blesa, and M.~Rodr\'iguez}, {\em Algorithm 924:
  Tides, a taylor series integrator for differential equations}, ACM Trans.
  Math. Software, 39 (2012).

\bibitem{PhysRevLett.95.010402}
{\sc M.~Albiez, R.~Gati, J.~F\"olling, S.~Hunsmann, M.~Cristiani, and M.~K.
  Oberthaler}, {\em {Direct observation of tunneling and nonlinear
  self-trapping in a single bosonic Josephson junction}}, Phys. Rev. Lett., 95
  (2005), p.~010402.

\bibitem{arnold1985singularities}
{\sc V.~Arnold, A.~Varchenko, and S.~Gusein-Zade}, {\em Singularities of
  Differentiable Maps: Volume I: The Classification of Critical Points Caustics
  and Wave Fronts}, Monographs in Mathematics, Birkh{\"a}user Boston, 1985.

\bibitem{Ashwin1990}
{\sc P.~Ashwin}, {\em {Symmetric chaos in systems of three and four forced
  oscillators}}, Nonlinearity, 3 (1990), pp.~603--617.

\bibitem{Shil4}
{\sc R.~Barrio and A.~Shilnikov}, {\em {Parameter--sweeping techniques for
  temporal dynamics of neuronal systems: case study of Hindmarsh--Rose model}},
  J. Math. Neurosci., 1 (2011).

\bibitem{Shil3}
{\sc R.~Barrio, A.~Shilnikov, and L.~Shilnikov}, {\em Kneadings, symbolic
  dynamics and painting {L}orenz chaos}, Internat. J. Bifur. Chaos Appl. Sci.
  Engrg., 22 (2012).

\bibitem{Ben-Tal2002}
{\sc A.~Ben-Tal}, {\em {Symmetry restoration in a class of forced
  oscillators}}, Physica D, 171 (2002), pp.~236--248.

\bibitem{MaiaThesis}
{\sc M.~Brunstein}, {\em Nonlinear Dynamics in {III-V} Semiconductor Photonic
  Crystal Nano-cavities}, PhD thesis, Universit\'e Paris Sud - Paris XI, 2011.

\bibitem{Call1}
{\sc R.~C. Calleja, E.~J. Doedel, A.~R. Humphries, A.~Lemus-Rodr\'iguez, and
  E.~B. Oldeman}, {\em Boundary-value problem formulations for computing
  invariant manifolds and connecting orbits in the circular restricted three
  body problem}, Celest. Mech. Dyn. Astron., 114 (2012), pp.~77--106.

\bibitem{BinMahmud}
{\sc B.~Cao, K.~W. Mahmud, and M.~Hafezi}, {\em Two coupled nonlinear cavities
  in a driven-dissipative environment}, Phys. Rev. A, 94 (2016), p.~063805.

\bibitem{Carmichael:1631392}
{\sc H.~Carmichael}, {\em {An open systems approach to quantum optics}},
  {Lecture Notes in Physics Monographs}, Springer, Berlin, 1993.

\bibitem{Casteels2017}
{\sc W.~Casteels and C.~Ciuti}, {\em {Quantum entanglement in the
  spatial--symmetry-breaking phase transition of a driven-dissipative
  Bose--Hubbard dimer}}, Phys. Rev. A, 95 (2017), p.~013812.

\bibitem{PhysRevA.93.033824}
{\sc W.~Casteels, F.~Storme, A.~Le~Boit\'e, and C.~Ciuti}, {\em Power laws in
  the dynamic hysteresis of quantum nonlinear photonic resonators}, Phys. Rev.
  A, 93 (2016), p.~033824.

\bibitem{Kirk1}
{\sc A.~R. Champneys, V.~Kirk, E.~Knobloch, B.~E. Oldeman, and J.~D.~M.
  Rademacher}, {\em Unfolding a tangent equilibrium--to--periodic heteroclinic
  cycle}, SIAM J. Appl. Dyn. Syst., 8 (2009), pp.~1261--1304.

\bibitem{san2}
{\sc A.~R. Champneys, Y.~Kuznetsov, and B.~Sandstede}, {\em A numerical toolbox
  for homoclinic bifurcation analysis}, Internat. J. Bifur. Chaos Appl. Sci.
  Engrg., 6 (1996), pp.~867--887.

\bibitem{Chossat1988}
{\sc P.~Chossat and M.~Golubitsky}, {\em {Symmetry--increasing bifurcation of
  chaotic attractors}}, Physica D, 32 (1988), pp.~423--436.

\bibitem{Lya1}
{\sc F.~Christiansen and H.~H. Rugh}, {\em {Computing Lyapunov spectra with
  continuous Gram--Schmidt orthonormalization}}, Nonlinearity,  (1997),
  pp.~1063--1072.

\bibitem{PhysRevE.64.025202}
{\sc P.~Coullet and N.~Vandenberghe}, {\em {Chaotic self-trapping of a weakly
  irreversible double Bose condensate}}, Phys. Rev. E, 64 (2001), p.~025202.

\bibitem{Dellnitz1995}
{\sc M.~Dellnitz and C.~Heinrich}, {\em {Admissible symmetry increasing
  bifurcations}}, Nonlinearity, 8 (1995), pp.~1039--1066.

\bibitem{Doe1}
{\sc E.~J. Doedel}, {\em {AUTO}: A program for the automatic bifurcation
  analysis of autonomous systems}, Congr. Numer., 30 (1981), pp.~265--284.

\bibitem{Doe3}
{\sc E.~J. Doedel, B.~Krauskopf, and H.~M. Osinga}, {\em Global invariant
  manifolds in the transition to preturbulence in the lorenz system}, Indag.
  Math., 22 (2011), pp.~223--241.

\bibitem{Doe2}
{\sc E.~J. Doedel and B.~E. Oldeman}, {\em {AUTO}-\textup{07}p: Continuation
  and Bifurcation Software for Ordinary Differential Equations}, Department of
  Computer Science, Concordia University, Montreal, Canada, 2010.
\newblock With major contributions from A. R. Champneys, F. Dercole, T. F.
  Fairgrieve, Y. Kuznetsov, R. C. Paffenroth, B. Sandstede, X. J. Wang and C.
  H. Zhang; available at \url{http://www.cmvl.cs.concordia.ca/}.

\bibitem{DrummondWalls1980}
{\sc P.~D. Drummond and D.~F. Walls}, {\em {Quantum theory of optical
  bistability. I. Nonlinear polarisability model}}, J. Phys. A: Math. Gen, 13
  (1980), pp.~725--741.

\bibitem{PhysRevE.67.066203}
{\sc C.~Emary and T.~Brandes}, {\em {Chaos and the quantum phase transition in
  the Dicke model}}, Phys. Rev. E, 67 (2003), p.~066203.

\bibitem{Gibbs1985}
{\sc H.~M. Gibbs}, {\em Optical Bistability: Controlling Light with Light},
  Academic Press, 1985.

\bibitem{Ali1}
{\sc R.~Gilmore and M.~Lefranc}, {\em The Topology of Chaos: Alice in Stretch
  and Squeezeland}, Wiley-Interscience, 2002.

\bibitem{And3}
{\sc A.~Giraldo, B.~Krauskopf, N.~G.~R. Broderick, A.~M. Yacomotti, and J.~A.
  Levenson}, {\em {The driven--dissipative Bose--Hubbard dimer: phase diagram
  and chaos}}, New J. Phys.,  (2020).

\bibitem{And1}
{\sc A.~Giraldo, B.~Krauskopf, and H.~M. Osinga}, {\em Saddle invariant objects
  and their global manifolds in a neighborhood of a homoclinic flip bifurcation
  of case {B}}, SIAM J. Appl. Dyn. Syst., 16 (2017), pp.~640--686.

\bibitem{And2}
\leavevmode\vrule height 2pt depth -1.6pt width 23pt, {\em Cascades of global
  bifurcations and chaos near a homoclinic flip bifurcation: A case study},
  SIAM J. Appl. Dyn. Syst., 17 (2018), pp.~2784--2829.

\bibitem{Glendinning1984}
{\sc P.~Glendinning}, {\em {Bifurcations near homoclinic orbits with
  symmetry}}, Phys. Lett. A, 103 (1984), pp.~163--166.

\bibitem{Golmakani2011}
{\sc A.~Golmakani and A.~J. Homburg}, {\em {Lorenz attractors in unfoldings of
  homoclinic--flip bifurcations}}, Dynamical Systems, 26 (2011), pp.~61--76.

\bibitem{GREBOGI1983181}
{\sc C.~Grebogi, E.~Ott, and J.~A. Yorke}, {\em Crises, sudden changes in
  chaotic attractors, and transient chaos}, Physica D, 7 (1983), pp.~181--200.

\bibitem{Haddadi_2014}
{\sc S.~Haddadi, P.~Hamel, G.~Beaudoin, I.~Sagnes, C.~Sauvan, P.~Lalanne, J.~A.
  Levenson, and A.~M. Yacomotti}, {\em Photonic molecules: tailoring the
  coupling strength and sign}, Optics Express, 22 (2014), p.~12359.

\bibitem{Hamel2015}
{\sc P.~Hamel, S.~Haddadi, F.~Raineri, P.~Monnier, G.~Beaudoin, I.~Sagnes,
  A.~Levenson, and A.~M. Yacomotti}, {\em {Spontaneous mirror--symmetry
  breaking in coupled photonic--crystal nanolasers}}, Nature Photonics, 9
  (2015), pp.~311--315.

\bibitem{Heinrich2008}
{\sc C.~Heinrich}, {\em Symmetry increasing bifurcations via collisions of
  attractors}, Rocky Mountain J. Math., 29 (2008), pp.~559--608.

\bibitem{HomSan}
{\sc A.~J. Homburg and B.~Sandstede}, {\em Homoclinic and heteroclinic
  bifurcations in vector fields}, in Handbook of Dynamical Systems, H.~W.
  Broer, B.~Hasselblatt, and F.~Takens, eds., vol.~3, {Elsevier, New York},
  2010, pp.~381--509.

\bibitem{Joannopoulos:08:Book}
{\sc J.~D. Joannopoulos, S.~G. Johnson, J.~N. Winn, and R.~D. Meade}, {\em
  Photonic Crystals: Molding the Flow of Light}, Princeton University Press,
  2~ed., 2008.

\bibitem{King1992}
{\sc G.~P. King and S.~T. Gaito}, {\em {Bistable chaos. II. Bifurcation
  analysis M.}}, Phys. Rev. A, 46 (1992), pp.~3100--3110.

\bibitem{Knobloch2013}
{\sc J.~Knobloch, J.~S. Lamb, and K.~N. Webster}, {\em {Using Lin's method to
  solve Bykov's problems}}, J. Differential Equations, 257 (2013),
  pp.~2984--3047.

\bibitem{Krauskopf_2004}
{\sc B.~Krauskopf and B.~E. Oldeman}, {\em A planar model system for the
  saddle{\textendash}node hopf bifurcation with global reinjection},
  Nonlinearity, 17 (2004), pp.~1119--1151.

\bibitem{Kra2}
{\sc B.~Krauskopf and H.~M. Osinga}, {\em Computing invariant manifolds via the
  continuation of orbit segments}, in Numerical Continuation Methods for
  Dynamical Systems: {Path Following and Boundary Value Problems},
  B.~Krauskopf, H.~M. Osinga, and J.~Gal\'an-Vioque, eds., {Springer, The
  Netherlands}, 2007, pp.~117--154.

\bibitem{KraRie1}
{\sc B.~Krauskopf and T.~{Rie{\ss}}}, {\em A {L}in's method approach to finding
  and continuing heteroclinic connections involving periodic orbits},
  Nonlinearity, 21 (2008), pp.~1655--1690.

\bibitem{Kuz1}
{\sc Y.~A. Kuznetsov}, {\em Elements of Applied Bifurcation Theory},
  {Springer-Verlag, New York}, 3rd~ed., 2004.

\bibitem{Lohr1}
{\sc A.~Lohse and A.~Rodrigues}, {\em Boundary crisis for degenerate singular
  cycles}, Nonlinearity, 30 (2017), pp.~2211--2245.

\bibitem{10.5555/1593342}
{\sc A.~Matsko}, {\em Practical Applications of Microresonators in Optics and
  Photonics}, CRC Press, Inc., USA, 1st~ed., 2009.

\bibitem{Matsumoto1987}
{\sc T.~Matsumoto, L.~O. Chua, and M.~Komuro}, {\em {Birth and death of the
  double scroll}}, Physica D, 24 (1987), pp.~97--124.

\bibitem{Melbourne1993}
{\sc I.~Melbourne, M.~Dellnitz, and M.~Golubitsky}, {\em {The structure of
  symmetric attractors}}, Arch. Ration. Mech. Anal., 123 (1993), pp.~75--98.

\bibitem{Palis1}
{\sc J.~Palis and W.~de~Melo}, {\em Geometric Theory of Dynamical Systems},
  {Springer, New York}, 1982.

\bibitem{RADEMACHER2005390}
{\sc J.~D.~M. Rademacher}, {\em Homoclinic orbits near heteroclinic cycles with
  one equilibrium and one periodic orbit}, J. Differential Equations, 218
  (2005), pp.~390--443.

\bibitem{RADEMACHER2010305}
\leavevmode\vrule height 2pt depth -1.6pt width 23pt, {\em {Lyapunov--Schmidt
  reduction for unfolding heteroclinic networks of equilibria and periodic
  orbits with tangencies}}, J. Differential Equations, 249 (2010),
  pp.~305--348.

\bibitem{Shil5}
{\sc L.~P. Shilnikov}, {\em A case of the existence of a denumerable set of
  periodic motions}, Sov. Math. Dokl., 6 (1965), pp.~163--166.

\bibitem{PhysRevResearch.2.033131}
{\sc K.~C. Stitely, A.~Giraldo, B.~Krauskopf, and S.~Parkins}, {\em Nonlinear
  semiclassical dynamics of the unbalanced, open dicke model}, Phys. Rev.
  Research, 2 (2020), p.~033131.

\bibitem{PhysRevA.102.063702}
{\sc K.~C. Stitely, S.~J. Masson, A.~Giraldo, B.~Krauskopf, and S.~Parkins},
  {\em Superradiant switching, quantum hysteresis, and oscillations in a
  generalized dicke model}, Phys. Rev. A, 102 (2020), p.~063702.

\bibitem{Wigg1}
{\sc S.~Wiggins}, {\em Introduction to Applied Nonlinear Dynamical Systems and
  Chaos}, {Springer-Verlag, New York}, 2nd~ed., 2003.

\bibitem{wimberger2014nonlinear}
{\sc S.~Wimberger}, {\em Nonlinear dynamics and quantum chaos}, Springer, 2014.

\bibitem{Shil1}
{\sc T.~Xing, R.~Barrio, and A.~Shilnikov}, {\em Symbolic quest into homoclinic
  chaos}, Internat. J. Bifur. Chaos Appl. Sci. Engrg., 24 (2014).

\bibitem{Kirk2}
{\sc W.~Zhang, B.~Krauskopf, and V.~Kirk}, {\em How to find a codimension-one
  heteroclinic cycle between two periodic orbits}, Discrete Contin. Dyn. Syst.,
  32 (2012), pp.~2825--2851.

\end{thebibliography}

\end{document}